\documentclass[12pt]{article}
\usepackage{amsmath}
\usepackage{amssymb}
\usepackage{amsthm}
\usepackage{amsfonts}
\usepackage{algorithm}
\usepackage{graphicx}
\usepackage{algpseudocode}
\usepackage{longtable}
\usepackage{natbib}
\usepackage{url} 
\usepackage{caption}
\usepackage{subcaption}
\usepackage{tikz}
\usetikzlibrary{calc}
\usepackage{hyperref}
\usepackage[mathscr]{euscript}
\usepackage[nodisplayskipstretch]{setspace}
\usepackage{xcolor}
\usepackage[nolabel, final]{showlabels}
\usepackage{comment}
\usepackage{lineno, blindtext, bm, url}
\usepackage{enumitem}

\newcommand{\blind}{0}

\newcommand{\tikzmark}[1]{\tikz[overlay,remember picture] \node (#1) {};}
\newcommand{\DrawBox}[4][]{%
    \tikz[overlay,remember picture]{%
        \coordinate (TopLeft)     at ($(#2)+(-0.2em,0.9em)$);
        \coordinate (BottomRight) at ($(#3)+(0.2em,-0.3em)$);
        \path (TopLeft); \pgfgetlastxy{\XCoord}{\IgnoreCoord};
        \path (BottomRight); \pgfgetlastxy{\IgnoreCoord}{\YCoord};
        \coordinate (LabelPoint) at ($(\XCoord,\YCoord)!0.5!(BottomRight)$);
        \draw [red,#1] (TopLeft) rectangle (BottomRight);
        \node [below, #1, fill=none, fill opacity=1] at (LabelPoint) {#4};
    }
}

\newtheorem{Lemma}{Lemma}
\newtheorem{lemma}{Lemma}
\newtheorem{Proposition}{Proposition}
\def\R{\mathbb{R}}
\def\MBA{\mathbf{A}}

\def\MBD{\mathbf{D}}
\def\MBE{\mathbf{E}}
\def\MBF{\mathbf{F}}
\def\MBG{\mathbf{G}}
\def\MBH{\mathbf{H}}
\def\MBI{\mathbf{I}}
\def\MBJ{\mathbf{J}}

\def\MBM{\mathbf{M}}
\def\MBN{\mathbf{N}}

\def\MBQ{\mathbf{Q}}
\def\MBR{\mathbf{R}}
\def\MBS{\mathbf{S}}

\def\MBU{\mathbf{U}}
\def\MBV{\mathbf{V}}
\def\MBW{\mathbf{W}}
\def\MBX{\mathbf{X}}
\def\MBY{\mathbf{Y}}
\def\MBZ{\mathbf{Z}}
\def\MBId{\mathbf{Id}}

\def\MB0{\mathbf{0}}

\DeclareMathOperator*{\rank}{rank}
\DeclareMathOperator*{\minimize}{minimize}
\DeclareMathOperator*{\argmin}{argmin}
\newtheorem{remark}{Remark}
\date{}

\begin{document}

\def\spacingset#1{\renewcommand{\baselinestretch}%
{#1}\small\normalsize} \spacingset{1}


\if0\blind
{
  \title{\bf Double-matched matrix decomposition for multi-view data}
  \author{Dongbang Yuan\thanks{
    The authors gratefully acknowledge the support from the National Science Foundation grants DMS-1712943 and  DMS-2044823.}\hspace{.2cm}\\
    Department of Statistics, Texas A\&M University\\
    and \\
    Irina Gaynanova\\
    Department of Statistics, Texas A\&M University}
  \maketitle
} \fi

\if1\blind
{
  \bigskip
  \bigskip
  \bigskip
  \begin{center}
    {\LARGE\bf Title}
\end{center}
  \medskip
} \fi

\bigskip
\begin{abstract}
We consider the problem of extracting joint and individual signals from multi-view data, that is, data collected from different sources on matched samples. While existing methods for multi-view data decomposition explore single matching of data by samples, we focus on double-matched multi-view data (matched by both samples and source features). Our motivating example is the miRNA data collected from both primary tumor and normal tissues of the same subjects; the measurements from two tissues are thus matched both by subjects and by miRNAs. Our proposed double-matched matrix decomposition allows us to simultaneously extract joint and individual signals across subjects, as well as joint and individual signals across miRNAs. Our estimation approach takes advantage of double-matching by formulating a new type of optimization problem with explicit row space and column space constraints, for which we develop an efficient iterative algorithm. Numerical studies indicate that taking advantage of double-matching leads to superior signal estimation performance compared to existing multi-view data decomposition based on single-matching. We apply our method to miRNA data as well as data from the English Premier League soccer matches and find joint and individual multi-view signals that align with domain-specific knowledge.
\end{abstract}

\noindent%
\vfill

\newpage
\spacingset{1.5} 
\section{Introduction}
\label{sec:intro}

Multi-view data (collected on the same samples from multiple views or data sources) are increasingly common with advances in multi-omics and other data collection technologies. In matrix form, each view $d$ corresponds to a matrix $\MBX_d$ with $n$ rows for the matched samples, and $p_d$ columns for corresponding measurements. While typically the distinct views are only matched by samples, in some cases the views are double-matched: matched by both samples (matched rows) and view features (matched columns). A motivating example is the miRNA data from The Cancer Genome Atlas (TCGA) project collected from both primary tumor and normal tissues of the same subjects; the measurements from two tissues represent two views $\MBX_1,\ \MBX_2\in \R^{n \times p}$ that are matched both by $n$ subjects and $p$ miRNAs. Our goal is to extract common (across tissues) as well as individual (tissue-specific) signals from each view, where common/individual signals have two meanings: common/individual signals across subjects, and common/individual signals across miRNAs. 

Several methods have been proposed that allow to extract common (joint) structure from the multi-view data. Canonical correlation analysis \citep{CCA} seeks linear combinations of features from each view that have maximal correlation. Similarly, partial least squares (PLS) \citep{rosipal2005overview} maximizes the covariance, with OnPLS \citep{OnPLS}, multiple coinertia analysis \citep{mengMCIA} and inter-battery analysis \citep{IBFAviaPLS} considering extensions to more than two views. JIVE \citep{JIVE} decomposes each view into joint and individual signals, where joint signals are due to matched samples. CIFE \citep{CIFE} and AJIVE \citep{AJIVE} consider the same joint and individual decomposition as JIVE, however use a different estimation procedure. Multi-Omics Factor Analysis (MOFA) \citep{MOFA, MOFA+} disentangles common and individual information using group factor analysis with the sparsity structure. DISCO-SCA \citep{DISCO-SCA} uses simultaneous components model \citep{VanDeun:2009tu} with subsequent rotation. iNMF \citep{iNMF} is a non-negative matrix factorization extension of JIVE. SLIDE \citep{SLIDE} 
allows for partially-common structures when the number of views is larger than two.  

Despite the considerable developments in multi-view data decompositions that extract joint and individual signals, these methods (JIVE, CIFE, AJIVE, MOFA, DISCO-CCA, SLIDE, iNMF) are designed for single-matched multi-view data (matched by samples) rather than double-matched in our motivating example. Applying these methods to double-matched data will lead to extraction of joint signals only in one direction. Let $\MBX_1,\ \MBX_2\in \R^{n \times p}$ be data matrices corresponding to double-matched views, and let $\widehat \MBA_1,\ \widehat \MBA_2 \in \R^{n \times p}$ be estimated signal matrices obtained by applying one of the existing approaches (e.g. JIVE, CIFE, AJIVE, etc). Then $\widehat \MBA_1$ and $\widehat \MBA_2$ will have joint signal in their column spaces (corresponding to matched rows), but no joint signal in their row spaces (corresponding to matched columns). A naive approach to estimate joint signal in the row space is to apply the same method to transposed $\MBX_1^{\top}, \MBX_2^{\top}\in \R^{p \times n}$ leading to $\widetilde \MBA_1,\ \widetilde \MBA_2 \in \R^{p \times n}$ with the joint signal corresponding to matched $p$ features. However, \textit{there is no guarantee} that the estimated signals agree with each other, that is in general $\widetilde \MBA_d \neq \widehat \MBA_d^{\top}$, which we confirm in our simulation studies (Section~\ref{s:Signalid}). Furthermore, some signal rank estimations methods, e.g. permutation approach in \citet{JIVE} or bi-cross-validation approach in \citet{SLIDE}, can lead to different estimated ranks
for the same $\MBX_1$, $\MBX_2$ depending on whether the matching by rows or the matching by columns is used (Section~\ref{s:Rankest}). 

Several methods consider the problem of extracting signal from double-matched multi-view data. Population value decomposition \citep{PVD} is an extension of singular value decomposition to double-matched data, however it only allows to extract joint signal, and does not extract individual signal. Similarly, 3-way PCA \citep{wold3wayPCA} and tensor decompositions \citep{zhouTensor} extract joint signals, but not individual. Linked matrix factorization \citep{LMF} and bidimensional integrative factorization \citep{BIDIFAC} are designed for the case where each pair of views is either matched by rows or by columns, but not simultaneously by both as in our motivating example. 

In this work, we propose the double-matched matrix decomposition (DMMD) for multi-view data that allows to extract joint and individual signals in both row and column directions simultaneously, in contrast to existing approaches. First, we prove that DMMD decomposition exists, and characterize conditions for its uniqueness. Second, we propose an estimation approach that takes advantage of the fact that the signal matrices must coincide whether the joint and individual signals are considered in the row direction, or in the column direction. We pose this estimation as a new type of optimization problem with explicit row space and column space constraints, for which we develop an efficient iterative algorithm. Third, we show that DMMD has superior signal estimation performance compared to existing methods for single-matched data even when underlying true signal ranks are known (Section~\ref{s:Signalid}), thus confirming the advantage of taking into account double-matched structure in estimation.

The rest of the paper is organized as follows. In Section~\ref{s:method}, we formulate the proposed double-matched matrix decomposition, and derive an algorithm for its estimation. In Section~\ref{sec:simulation}, we compare DMMD to existing methods on simulated data. In Section~\ref{sec:application}, we illustrate DMMD on the double-matched miRNA data from TCGA, and double-matched English Premier league soccer match data. In Section~\ref{sec:discus} we conclude with discussion. 

\section{Method}
\label{s:method}

\subsection{Notation}
For a matrix $\MBA \in \mathbb{R}^{n \times p}$,  we let $\MBA^T$ be its transpose, $\mathcal{C}(\MBA)$ be its column space and $\mathcal{R}(\MBA)$ be its row space. We use $\|\MBA\|_F = \sqrt{\sum_{i=1}^{n}{\sum_{j=1}^{p}{a^2_{ij}}}}$ to denote its Frobenius norm. For two matrices $\MBA_1 \in \mathbb{R}^{n \times p_1}$ and $\MBA_2 \in \mathbb{R}^{n \times p_2}$, we write $[\MBA_1,\MBA_2] \in \mathbb{R}^{n \times (p_1 + p_2)}$ to denote the column-wise concatenation. We say $\mathcal{C}(\MBA_1)$ is orthogonal to $\mathcal{C}(\MBA_2)$ if for any vector $\boldsymbol{x}_1 \in \mathcal{C}(\MBA_1)$ and any vector $\boldsymbol{x}_2 \in \mathcal{C}(\MBA_2)$, it holds that $\boldsymbol{x}_1 \perp \boldsymbol{x}_2$. We use $\MBId$ to denote an identity matrix. We use $\boldsymbol{e}_i = (0,\cdots,0,1,0,\cdots,0)^T$ with only the $i$-th element being one to denote the standard basis vector. We use script-style letter $\mathcal{U}$ to denote a vector space formed by the matrix $\MBU$ with columns corresponding to orthonormal basis vectors, $\mathcal{C}(\MBU) = \mathcal{U}$. 

\subsection{Model}
\label{sec:model}
We consider two double-matched data matrices $\MBX_1\in \mathbb{R}^{n \times p}$ and $\MBX_2 \in \R^{n \times p}$. We assume additive decomposition $\MBX_k = \MBA_k + \MBE_k$, $k = 1,2$, where $\MBA_k$ is the signal matrix and $\MBE_k$ is the noise matrix. We further assume each signal matrix $\MBA_k$ is low-rank, which is common in the literature \citep{LowRank}. Our goal is to estimate $\MBA_k$ from $\MBX_k$, and identify parts of the signal that are joint/individual across row dimension $n$ (samples) as well as parts of the signal that are joint/individual across column dimension $p$ miRNAs).

Existing methods for estimation of $\MBA_k$ in single-matched multi-view data \citep{JIVE, AJIVE,CIFE, iNMF, SLIDE} are based on separating the signal matrix into joint and individual parts with respect to the matched dimension, that is $\MBA_k = \MBJ_k + \MBI_k$. For example, in JIVE model \citep{JIVE, AJIVE, CIFE}, the joint matrices $\MBJ_1$ and $\MBJ_2$ share the same column space, i.e., $\mathcal{C}(\MBJ_1) = \mathcal{C}(\MBJ_2) = \mathcal{C}(\MBJ)$. The individual matrices $\MBI_1$ and $\MBI_2$ are orthogonal to the joint space and have zero intersection of their respective column spaces, i.e., $\mathcal{C}(\MBJ) \perp \mathcal{C}(\MBI_k), \cap_{j=1}^2{\mathcal{C}(\MBI_k)} = \{\MB0\}$. Furthermore, given the signal matrices $\MBA_k$, the JIVE decomposition is unique \citep{JIVE, AJIVE}.

Our proposal is based on the observation that for double-matched signal matrices $\MBA_k$, the JIVE decomposition must hold with respect to both dimensions (row and column) simultaneously.
We formalize this observation in the following lemma, which is a generalization of Lemma~1 from \citet{AJIVE}.
\begin{Lemma}
\label{l:model}
\textit{Given two signal matrices $\MBA_1, \MBA_2 \in \mathbb{R}^{n \times p}$ , there are unique sets of matrices $\{\MBJ_{c1},\MBJ_{c2}\}$, $\{\MBI_{c1},\MBI_{c2}\}$, $\{\MBJ_{r1},\MBJ_{r2}\}$ and $\{\MBI_{r1},\MBI_{r2}\}$ such that}\\
(1)\quad $\MBA_k = \MBJ_{ck} + \MBI_{ck} = \MBJ_{rk} + \MBI_{rk},\quad k = 1,2$\\
(2)\quad $\mathcal{C}(\MBJ_{ck}) = \mathcal{M} \subset\mathcal{C}(\MBA_k) , \quad k = 1,2$\\
(3)\quad $\mathcal{R}(\MBJ_{rk}) = \mathcal{N} \subset\mathcal{R}(\MBA_k) , \quad k = 1,2$\\
(4)\quad $\mathcal{M} \perp \mathcal{C}(\MBI_{ck}), \mathcal{N} \perp \mathcal{R}(\MBI_{rk}), \quad k = 1,2$\\
(5)\quad $\mathcal{C}(\MBI_{c1})\cap{\mathcal{C}(\MBI_{c2})} = \{\MB0\}, \mathcal{R}(\MBI_{r1})\cap{\mathcal{R}(\MBI_{r2})} = \{\MB0\}$
\end{Lemma}
Here $\mathcal{M}$ represents the joint column structure (common signal for $n$ samples) and $\mathcal{N}$ represents joint row structure (common signal for $p$ features) of the signal matrices $\{\MBA_1,\MBA_2\}$. Similarly, $\MBI_{ck}$ and $\MBI_{rk}$ represent the individual column signals and individual row signals, respectively. Lemma \ref{l:model} applies to double-matched matrices $\{\MBA_1,\cdots,\MBA_K\}$ from more than two views $ (K > 2)$; we only present case $K=2$ as it is sufficient for motivating datasets.

In light of Lemma~\ref{l:model}, we consider the following Double-Matched Matrix Decomposition (DMMD) for observed $\MBX_1\in \mathbb{R}^{n \times p}$ and $\MBX_2 \in \R^{n \times p}$
\begin{equation}\label{eq:dmmd}
    \MBX_k = \underbrace{\MBJ_{ck} + \MBI_{ck}}_{\MBA_k}+ \MBE_k = \underbrace{\MBJ_{rk} + \MBI_{rk}}_{\MBA_k} + \MBE_k, \quad k=1,2;
\end{equation}
where $\MBJ_{ck}$, $\MBJ_{rk}$, $\MBI_{ck}$, $\MBI_{rk}$ satisfy the above conditions. The main novelty of DMMD is that the signal $\MBA_k$ is constrained to be the same whether it is decomposed in column or in row direction. In what follows, we use $r_k = \rank(\MBA_k)$, $k=1,2$, to denote the total rank of each signal matrix; $\MBM$ and $\MBN$ to denote the matrices that contain basis vectors of $\mathcal{M}$ and $\mathcal{N}$ column-wise, respectively; $r_c = \rank(\MBM)$ to denote the rank of joint column structure, and $r_r = \rank(\MBN)$ to denote the rank of joint row structure. Figures \ref{fig:demo_col} and \ref{fig:demo_row} show an example of the decomposition~\eqref{eq:dmmd} on a simulated data. 

\begin{figure}[!t]
\centering
\begin{subfigure}{0.5\textwidth}
\includegraphics[width=1\linewidth]{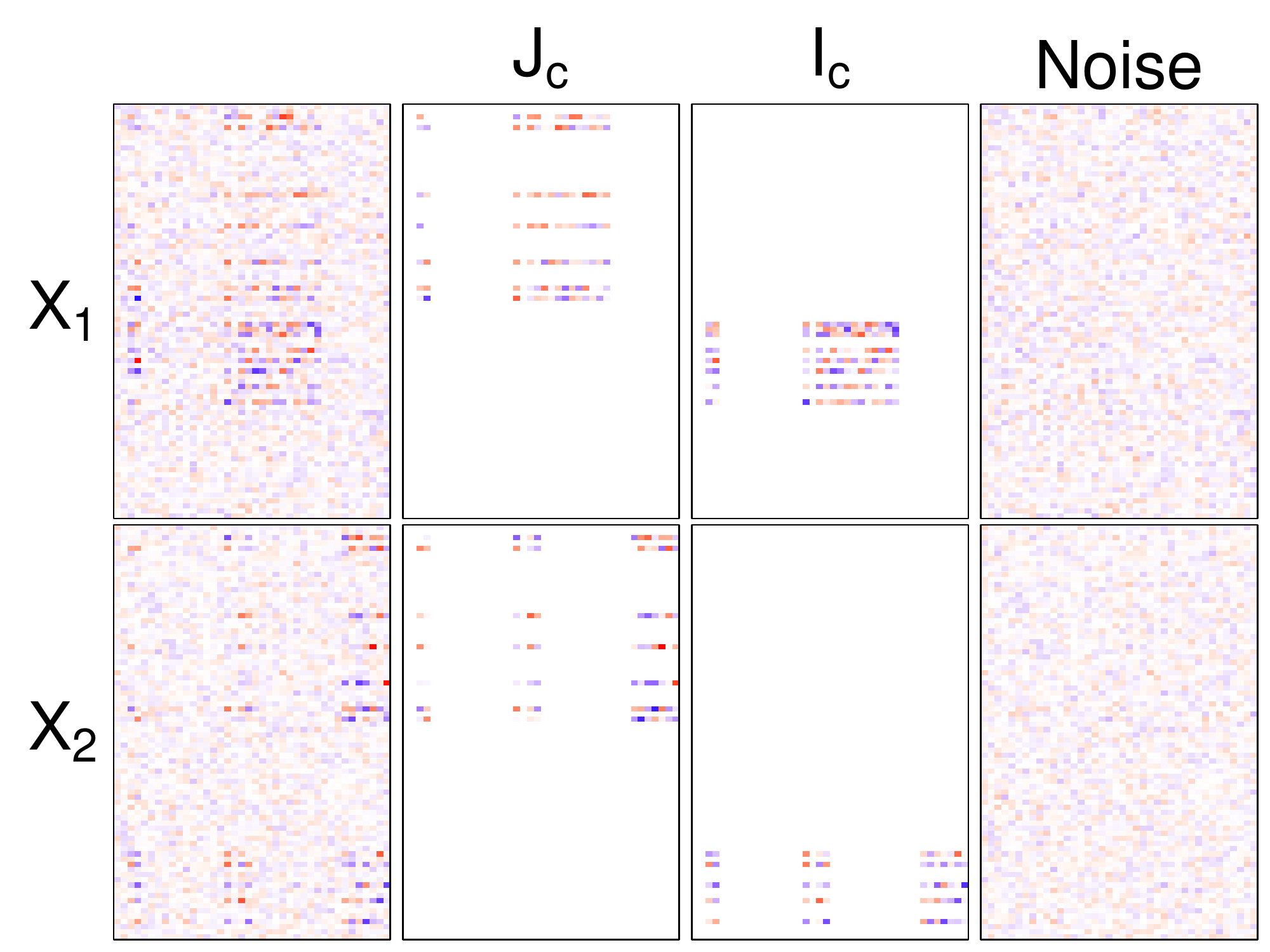}
\caption{Column decomposition}
\label{fig:demo_col}
\end{subfigure}%
\begin{subfigure}{0.5\textwidth}
\includegraphics[width=1\linewidth]{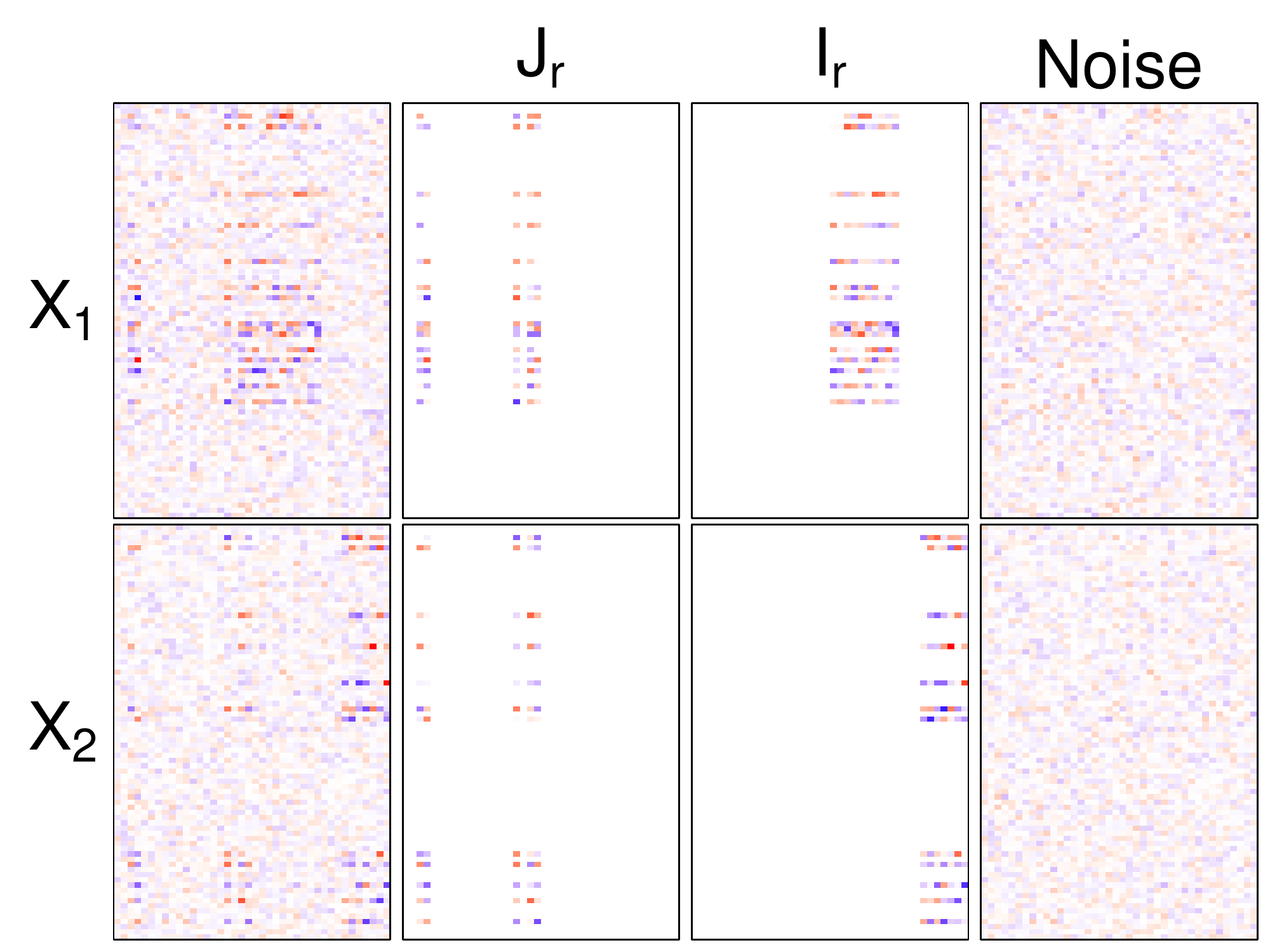}
\caption{Row decomposition}
\label{fig:demo_row}
\end{subfigure}
\caption{Two double-matched matrices are decomposed into joint structure, individual structures and noise in both row and column directions according to DMMD model~\eqref{eq:dmmd}, here $n=80$, $p=40$, $\rank(\MBA_1)=15$, $\rank(\MBA_2)=12$,  $\rank(\MBM) =7$ and $\rank(\MBN) = 5$.}
\end{figure}

\subsection{Estimation}

To fit model~\eqref{eq:dmmd}, we propose the following estimation approach:
\begin{description}
  \item[\textbf{Step 1: Estimate proxy signals.}] Estimate the total ranks of $\MBA_1$ and $\MBA_2$, and construct proxy signal matrices $\MBZ_1$ and $\MBZ_2$ from $\MBX_1$ and $\MBX_2$ given those ranks.  
  \item[\textbf{Step 2: Estimate joint structure.}] Use proxy signals $\MBZ_1$ and $\MBZ_2$ to estimate basis vectors of $\mathcal{M}$ (joint column structure) and $\mathcal{N}$ (joint row structure). 
  \item[\textbf{Step 3: Estimate signals with given joint structure.}] Fit model~\eqref{eq:dmmd} conditionally on the estimated $\mathcal{M}$, $\mathcal{N}$ from step 2 and estimated total ranks from step 1. 
\end{description}
Figure \ref{fig:DMMD_demo} shows the flow chart summarizing DMMD estimation steps. In Supplement S5, we describe a variation of Step 3 that allows for iterative updates of initial $\mathcal{M}$, $\mathcal{N}$ from Step 2 leading to iterative DMMD (DMMD-i). Numerically, the two approaches are very similar, but DMMD-i has a significantly higher computational cost (Section~\ref{sec:simulation}).

\begin{figure}[!t]
\centering
\includegraphics[width=1\linewidth]{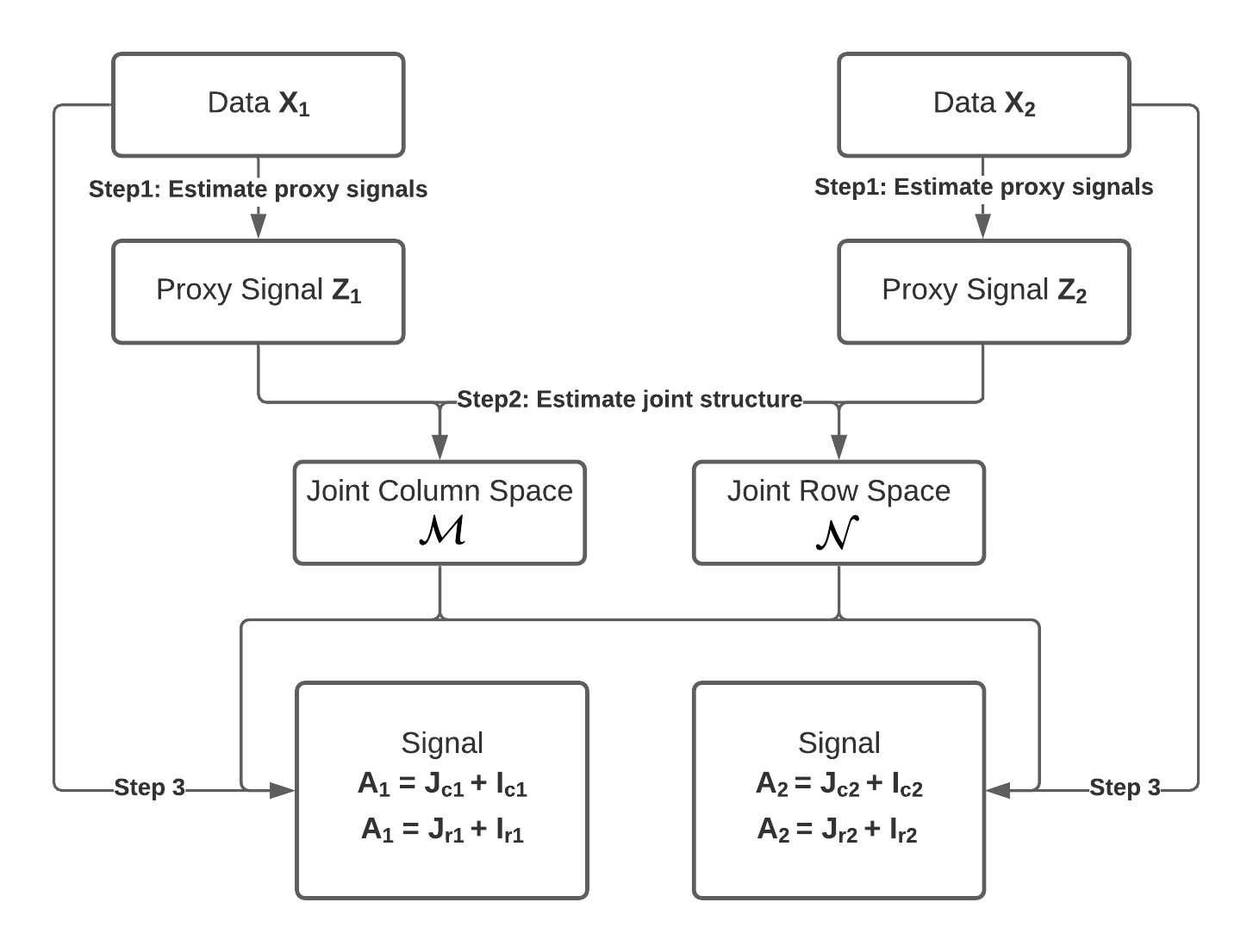}
\caption{Summary of the proposed estimation approach for fitting DMMD model~\eqref{eq:dmmd}.}
\label{fig:DMMD_demo}
\end{figure}

\subsubsection{Estimation of proxy signals}
\label{s:pl}
We estimate proxy low-rank signal matrices $\MBZ_k$ from observed $\MBX_k$ using low-rank singular value decomposition of $\MBX_k$ \citep{jha2010denoising}. We propose to use the profile likelihood approach \citep{ProfileLik2006Zhu} to estimate the total rank of the signal. 

Let $d_1 \geq d_2 \geq \cdots \geq d_m$ be the ordered singular values of matrix $\MBX_1$, where $m = \min(n,p)$. Given a fixed $q$ with $1\leq q \leq m$, define sets $D_1 = \{d_1, d_2,\cdots, d_q\}$ and $D_2 = \{d_{q+1},\cdots, d_m\}$. \citet{ProfileLik2006Zhu} assume that the elements of $D_1$ and $D_2$ come from the normal distributions $N(\mu_1,\sigma^2)$ and $N(\mu_2,\sigma^2)$, respectively. Let $f(\cdot;\mu,\sigma^2)$ be the probability density function of $N(\mu,\sigma^2)$. Then the log-likelihood is $$l(q,\mu_1,\mu_2,\sigma^2) = \sum_{i=1}^{q}{\log f(d_i;\mu_1,\sigma^2)} + \sum_{j=q+1}^{m}{\log f(d_j;\mu_2,\sigma^2)}.$$
Given $q$,  the MLEs are $\widehat{\mu}_1 = \sum_{i=1}^q{d_i}/q$, $\widehat{\mu}_2 = \sum_{j=q+1}^m{d_j}/(m-q)$ and $\widehat{\sigma}^2 = [(q-1)s_1^2 + (m-q-1)s_2^2]/m$, where $s_1^2$ and $s_2^2$ are the sample variances of elements in $D_1$ and $D_2$, respectively. We estimate the rank of signal $\MBA_1$ by maximizing the profile likelihood over $q$ and set $r_1 = \widehat q$, where $\widehat q$ is the maximizer. The same approach is used for $\MBX_2$. Given $r_k$, we obtain proxy signal matrix $\MBZ_k$ by corresponding rank-$r_k$ SVD  of observed $\MBX_k$.

\begin{remark}
We use profile likelihood approach for rank estimation as it is fast and performs well in our simulations, however an alternative rank estimation approach can be used in this step. Some examples are permutation method \citep{JIVE}, edge distribution method \citep{onatski2010determining} and Bi-Cross-Validation method \citep{BCV}. We compare these approaches in simulations in Section~\ref{s:Rankest}.
\end{remark}

\subsubsection{Estimation of joint structure}
\label{s:joint}
In this section we estimate the joint column structure $\mathcal{M}$ and the joint row structure $\mathcal{N}$ based on the proxy signals $\MBZ_1$ and $\MBZ_2$. From the proof of Lemma~\ref{l:model}, the joint column structure $\mathcal{M} = \mathcal{C}(\MBA_1)\cap{\mathcal{C}(\MBA_2)}$ and the joint row structure $\mathcal{N} = \mathcal{R}(\MBA_1)\cap{\mathcal{R}(\MBA_2)}$. Thus, a naive way to estimate $\mathcal{M}$ is to consider the intersection of column spaces of proxy signals $\MBZ_1$ and $\MBZ_2$, $\mathcal{C}(\MBZ_1)\cap{\mathcal{C}(\MBZ_2)}$, however $\MBZ_k$ is only an estimate of $\MBA_k$. Thus, in practice $\mathcal{C}(\MBZ_1)\cap{\mathcal{C}(\MBZ_2)} = \{\MB0\}$ due to the corruption of true joint structure by noise. To circumvent this difficulty, we propose to use principal angles to measure the similarity between $\mathcal{C}(\MBZ_1)$ and $\mathcal{C}(\MBZ_2)$. Both CCA\citep{CCA-principle-angle} and inter-battery factor analysis \citep{IBFAviaPLS} also use the cosines of principal angles to measure similarity. We propose to separate the principal angles into two groups: small angles indicating common signals (albeit not exactly equal) and large angles indicating individual signals. Similar idea is used in \citet{AJIVE}, however our approach for determining the angle cutoff is different. 

We first review principal angles. Let $\mathcal{U}$ and $\mathcal{V}$ be subspaces with $dim(\mathcal{U}) = h_1, dim(\mathcal{V}) = h_2$ in $\mathbb{R}^{n}$. Let $h = \min(h_1,h_2)$, then the principal angles $\Theta(\mathcal{U},\mathcal{V}) = \{\theta_k \in [0,\frac{\pi}{2}]| k = 1,2,\cdots,h\}$ between $\mathcal{U}$ and $\mathcal{V}$ are recursively defined by 
\begin{align*}
    &\cos{\theta_k} = \max_{\boldsymbol{x} \in \mathcal{U}}{\max_{\boldsymbol{y} \in \mathcal{V}}{|\boldsymbol{x}^T\boldsymbol{y}|}} = |\boldsymbol{x}_k^T\boldsymbol{y}_k| \\
    & \mbox{subject to}\quad \|\boldsymbol{x}\| = \|\boldsymbol{y}\| = 1, \boldsymbol{x}^T\boldsymbol{x}_i = 0, \boldsymbol{y}^T\boldsymbol{y}_i = 0, \quad i = 1,2,\cdots,k-1.
\end{align*}
The vectors $\{\boldsymbol{x}_1, \cdots, \boldsymbol{x}_h\}$ and $\{\boldsymbol{y}_1, \cdots, \boldsymbol{y}_h\}$ are called principal vectors. Principal angles can be calculated using singular value decomposition \citep{knyazev2002principal}. Let $\MBX \in \mathbb{R}^{n\times h_1}$ and $\MBY \in \mathbb{R}^{n \times h_2}$ be the orthogonal matrices formed by concatenating orthonormal basis vectors of $\mathcal{U}$ and $\mathcal{V}$ column-wise, respectively. Let SVD of $\MBX^T\MBY$ be $\MBU\mathbf{\Sigma}\MBV^T$ where $\mathbf{\Sigma}$ is a $h_1$ by $h_2$ diagonal matrix with singular values $s_1(\MBX^T\MBY),\cdots,s_h(\MBX^T\MBY)$ in non-increasing order. Then 
$
\cos \Theta(\mathcal{U}, \mathcal{V}) = \{s_1(\MBX^T\MBY),\cdots,s_h(\MBX^T\MBY)\}.
$
Moreover, the corresponding principal vectors are given by the first $h$ columns of $\MBX\MBU$ and $\MBY\MBV$.

Using principal angles, we estimate joint column structure $\mathcal{M}$ from $\MBZ_1$ (with rank $r_1$) and $\MBZ_2$ (with rank $r_2$) as follows. First we calculate the principal angles $\theta_1, \cdots, \theta_l, l = \min(r_1,r_2)$, and principal vectors $\{\boldsymbol{u}_1, \cdots, \boldsymbol{u}_l\}$ and $\{\boldsymbol{v}_1, \cdots, \boldsymbol{v}_l\}$ between two column spaces $\mathcal{C}(\MBZ_{1})$ and $\mathcal{C}(\MBZ_{2})$. Since the joint rank could be 0 or $l$, we add artificial $\theta_0 = 0$ and $\theta_{l+1} = \pi/2$ angles into the principal angle vector. We then separate the angles into two groups by using profile likelihood as described in Section~\ref{s:pl} to estimate the optimal cutoff $\widehat{q}$. The estimated joint column rank is then $r_c = \widehat{q} - 1$. We estimate the basis for joint column structure $\mathcal{M}$ by calculating the element-wise average of principal vectors, e.g., $\boldsymbol{w}_i = \frac{1}{2}(\boldsymbol{u}_i + \boldsymbol{v}_i), i = 1,2,\cdots,r_c,$ corresponding to the smallest $r_c$ principal angles. The basis of joint row structure $\mathcal{N}$ together with its rank $r_r$ is determined similarly from the row spaces $\mathcal{R}(\MBZ_{1})$ and $\mathcal{R}(\MBZ_{2})$. Let $\MBM$ and $\MBN$ be the matrices containing averaged principal vectors corresponding to $\mathcal{M}$ and $\mathcal{N}$, respectively. To form basis vectors, we orthogonalize $\MBM$ and $\MBN$ using Gram-Schmidt process. The full procedure  is summarized in Figure \ref{fig:Caljoint}. 

\begin{figure}[!t]
\centering
\includegraphics[width=1\linewidth]{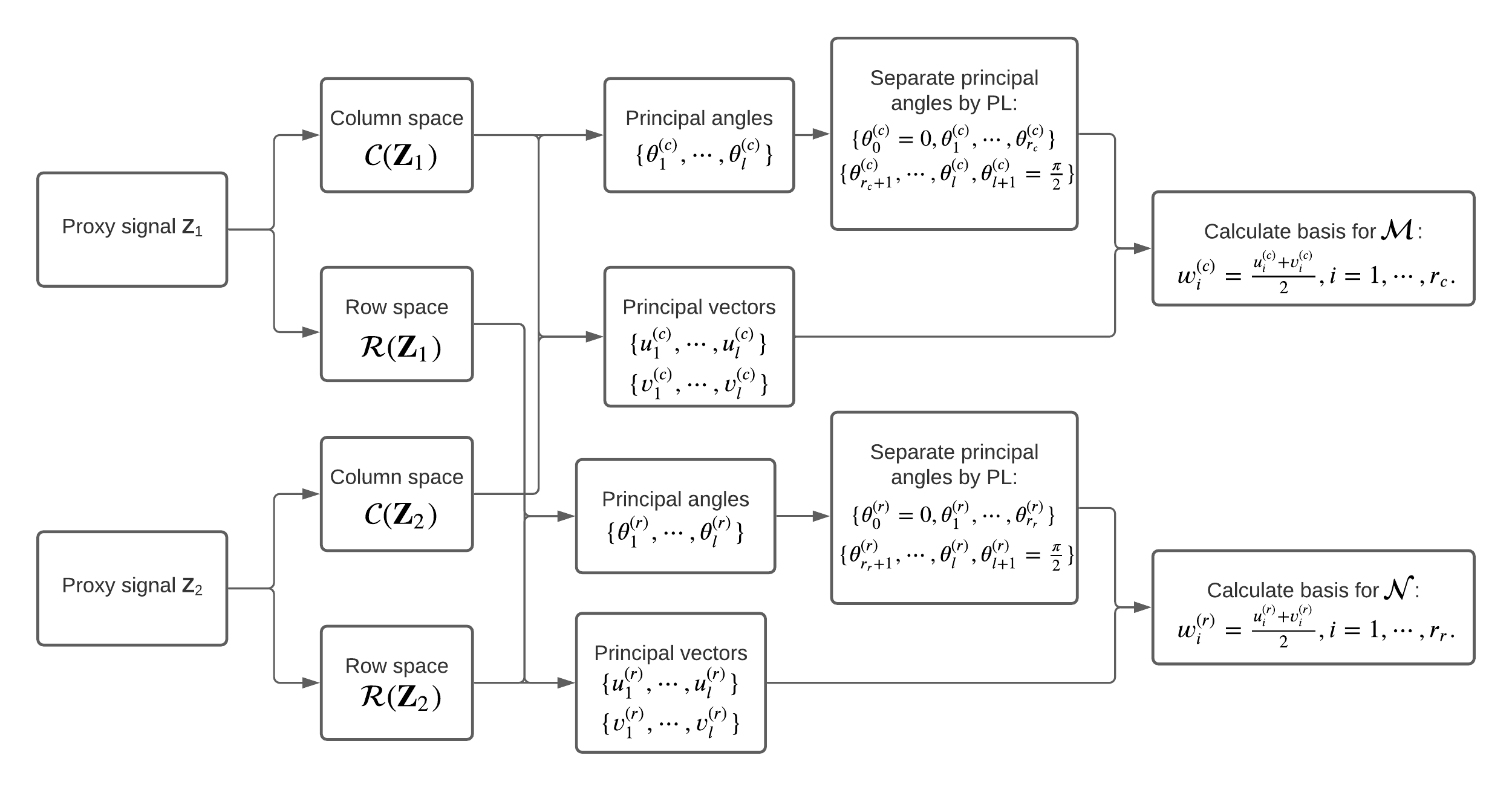}
\caption{Procedure of calculating joint structure in DMMD}
\label{fig:Caljoint}
\end{figure}

\begin{remark}
An alternative rank estimation approach can be used in this step. Some examples are permutation method \citep{JIVE}, Bi-Cross-Validation method \citep{BCV} and Wedin bound method \citep{AJIVE}. See Section~\ref{s:Rankest} for comparison.
\end{remark}

\subsubsection{Estimation of signals with given joint structure}

We consider the estimation of the signals $\MBA_k$, $k = 1, 2$, in~\eqref{eq:dmmd} given the joint structures $\mathcal{M}$ and $\mathcal{N}$. Our goal is to find the closest matrix to $\MBX_k$ that simultaneously contains both given joint column structure $\mathcal{M}$ and given joint row structure $\mathcal{N}$, that is to solve
\begin{align}
    &\minimize_{\MBA_k\in \R^{n\times p}}{\|\MBX_k - \MBA_k\|^2_{F}} \label{eq:DMopt}\\
    & \mbox{such that}\quad \mathcal{C}(\MBM) \subset \mathcal{C}(\MBA_k),\ \mathcal{C}(\MBN) \subset \mathcal{R}(\MBA_k),\ \text{rank}(\MBA_k) = r_k,\quad k = 1,2 \nonumber,
\end{align}
where $r_k$ is the total rank for signal $\MBA_k$ as in Section~\ref{s:pl}, $\MBM $ is the joint column space with $r_c$ basis in $\mathbb{R}^n$ and $\MBN$ is the joint row space with $r_r$ basis in $\mathbb{R}^p$ estimated as in Section~\ref{s:joint}. In Supplement S5, we describe a variation of \eqref{eq:DMopt} with additional minimization over joint structures that allows for iterative updates of initial $\MBM$, $\MBN$, leading to iterative DMMD (DMMD-i). We compare DMMD and DMMD-i in Section~\ref{sec:simulation}.

To solve \eqref{eq:DMopt}, we first consider a simplified problem by removing the row-space constraint, for which we derive a closed-form solution.

\begin{Lemma}
\label{l:SMMD}
Given $\MBX\in \R^{n\times p}$, $\MBM\in \R^{n\times r_c}$ with orthonormal columns with $\text{rank}(\MBM\MBM^T\MBX) = r_c$, and rank $r$ with $r_c\leq r \leq \min(n, p$), consider
\begin{equation}\label{eq:simpler}
    \minimize_{\MBA \in \R^{n\times p}}{\|\MBX - \MBA\|_{F}^2}
     \quad \mbox{such that}\quad \mathcal{C}(\MBM) \subset \mathcal{C}(\MBA), \quad rank(\MBA) = r. 
\end{equation}
Let $\MBA_M^* = \MBM\MBM^T\MBX + \MBR\MBR^T\MBX$, where the columns of $\MBR$ are the first $r - r_c$ left singular vectors of $(\MBId - \MBM\MBM^T)\MBX$. Then $\MBA_M^{*}$ is the global minimizer of~\eqref{eq:simpler}. Furthermore, if matrix $(\MBId - \MBM\MBM^T)\MBX$ has distinct $(r - r_c)$-th and $(r - r_c + 1)$-th singular values, then $\MBA_M^{*}$ is the unique minimizer of~\eqref{eq:simpler}.
\end{Lemma}

From Lemma~\ref{l:SMMD}, the columns of matrix $[\MBM, \MBR]$ are the basis vectors of the column-space of the solution $\MBA_M^{*}$ to~\eqref{eq:simpler}.
Similarly, consider a simplified problem \eqref{eq:DMopt} with row-space constraint, but no column-space constraint. 
\begin{Lemma}
\label{l:SMMDN}
Given $\MBX\in \R^{n\times p}$, $\MBN\in \R^{p\times r_r}$ with orthonormal columns with $\text{rank}(\MBX\MBN\MBN^T) = r_r$, and rank $r$ with $r_r\leq r \leq \min(n, p$), consider
\begin{equation}\label{eq:simpler2}
    \minimize_{\MBA \in \R^{n\times p}}{\|\MBX - \MBA\|_{F}^2}
     \quad \mbox{such that}\quad \mathcal{C}(\MBN) \subset \mathcal{R}(\MBA), \quad rank(\MBA) = r. 
\end{equation}
Let $\MBA_N^* = \MBX\MBN\MBN^T + \MBX\MBS\MBS^T$, where the columns of $\MBS$ are the first $r - r_r$ right singular vectors of $\MBX(\MBId - \MBN\MBN^T)$. Then $\MBA_N^{*}$ is the global minimizer of~\eqref{eq:simpler2}. Furthermore, if matrix $\MBX(\MBId - \MBN\MBN^T)$ has distinct $(r - r_r)$-th and $(r - r_r + 1)$-th singular values, then $\MBA_N^{*}$ is the unique minimizer of~\eqref{eq:simpler2}.
\end{Lemma}

The columns of matrix $[\MBN, \MBS]$ are the basis vectors of the row-space of the solution $\MBA_N^{*}$ to~\eqref{eq:simpler2}.
Given the closed form solutions to~\eqref{eq:simpler} and~\eqref{eq:simpler2}, we propose an iterative algorithm for the full problem~\eqref{eq:DMopt}, where we alternate the update of column space with the update of the row space. The full algorithm is summarized in Algorithm \ref{a:opt}. In words, we first initialize the full column space $\widetilde \MBM_k$ of each signal matrix, and update the row space that is not captured by $\MBN$. Given the updated full row space of each signal matrix $\widetilde \MBN_k$, we then  update the column space that is not captured by $\MBM$. At each step, the current estimated signal matrix is $\MBA_k^{(t)} =\widetilde{\MBM}_k^{(t)}\widetilde{\MBM}_k^{(t)T}\MBX_k\widetilde{\MBN}_k^{(t)}\widetilde{\MBN}_k^{(t)T}$, which is feasible as long as it has full rank $r_k$ (this is always satisfied in our numerical studies because of noisy $\MBX_k$). Once the estimated signal matrices $\MBA_k$ are obtained, by construction it holds that for $k = 1, 2$: $\mathbf{J}_{ck} =  \mathbf{M}\mathbf{M}^\top\mathbf{A}_k, \mathbf{I}_{ck} =  (\mathbf{Id} - \mathbf{M}\mathbf{M}^\top)\mathbf{A}_k.$
Similarly, $\mathbf{J}_{rk} =  \mathbf{A}_k\mathbf{N}\mathbf{N}^\top, \mathbf{I}_{rk} = \mathbf{A}_k(\mathbf{Id} - \mathbf{N}\mathbf{N}^\top).$

 Let $L_k^{(t)}$ be the objective function of \eqref{eq:DMopt} at iteration step $t$, $L_k^{(t)} = L_k(\MBR^{(t)},\MBS^{(t)}) = \|\MBX_k - \MBA^{(t)}_k\|^2_F = \|\MBX_k - (\MBM\MBM^T + \MBR^{(t)}\MBR^{(t)T})\MBX_k(\MBN\MBN^T+\MBS^{(t)}\MBS^{(t)T})\|_F^2$. We show that Algorithm~\ref{a:opt} is guaranteed to converge as it leads to non-increasing sequence of $L_k^{(t)}$.   


\begin{Proposition}
\label{convergence}
If at each iteration step $t$ in Algorithm \ref{a:opt}, $(\mathbf{M}\mathbf{M}^T + \mathbf{R}_k^{(t)}\mathbf{R}_k^{(t)T})\mathbf{X}_k(\mathbf{Id} - \mathbf{N}\mathbf{N}^T)$ is of rank at least $r_k-r_r$ and $(\mathbf{Id} - \mathbf{M}\mathbf{M}^T)\mathbf{X}_k(\MBN\MBN^T + \MBS_k^{(t)}\MBS_k^{(t)T})$ is of rank at least $r_k-r_c$, then the sequence of objective values $L_k^{(t)}$ is non-increasing.
\end{Proposition}
Proposition~\ref{convergence} only guarantees the convergence of objective values, but in practice we found that that the sequences of $\MBR_k^{(t)}$ and $\MBS_k^{(t)}$ also converge. The convergence to the global minimizer is not guaranteed since problem~\eqref{eq:DMopt} is nonconvex, thus the output may depend on initial $\MBR_k^{(0)}$. The proposed $\MBR_k^{(0)}$ corresponds to global solution when $r_r = 0$ due to Lemma~2. Empirically this choice leads to a smaller objective value at convergence compared to a random $\MBR_k^{(0)}$, and excellent signal estimation performance.

\begin{algorithm}[!t]
\caption{Iterative algorithm for \eqref{eq:DMopt}}
\label{a:opt}
\begin{algorithmic}[1]
\State Given: $\MBX_k \in \R^{n\times p}$, $r_k$, $k=1, 2$; $\MBM\in \R^{n\times r_c}$, $\MBN\in \R^{p\times r_r},t_{max},\epsilon > 0$
\For{$k = 1,2$}
\State SVD: $(\MBId - \MBM\MBM^T)\MBX_k = \MBU_k\MBD_k\MBV_k^T$
\State $\MBR_k^{(0)} \gets$ first $r_k - r_c$ columns of $\MBU_k$
\State $\widetilde{\MBM}_k^{(0)} \gets [\MBM,\MBR_k^{(0)}]$
\State $t \gets 0$
\While{$t\neq t_{max}$ and $|L_k^{(t)} - L_k^{(t-1)}|>\epsilon$}
\State SVD: $\widetilde{\MBM}_k^{(t)}\widetilde{\MBM}_k^{(t)T}\MBX_k(\MBId- \MBN\MBN^T) = \MBU^{(t)}_{1,k}\MBD^{(t)}_{1,k}\MBV^{(t)T}_{1,k}$
\State $\MBS_k^{(t+1)} \gets$ first $r_k - r_r$ columns of $\MBV^{(t)}_{1,k}$
\State $\widetilde{\MBN}_k^{(t+1)} \gets [\MBN,\MBS_k^{(t+1)}]$
\State SVD: $(\MBId - \MBM\MBM^T)\MBX_k\widetilde{\MBN}_k^{(t+1)}\widetilde{\MBN}_k^{(t+1)T} = \MBU^{(t)}_{2,k}\MBD^{(t)}_{2,k}\MBV^{(t)T}_{2,k}$
\State $\MBR_k^{(t+1)} \gets$ first $r_k - r_c$ columns of $\MBU^{(t)}_{2,k}$
\State $\widetilde{\MBM}_k^{(t+1)} \gets [\MBM,\MBR_k^{(t+1)}]$
\State $t \gets t + 1$
\State $L_k^{(t)} = \|\MBX_k - \widetilde{\MBM}_k^{(t)}\widetilde{\MBM}_k^{(t)T}\MBX_k\widetilde{\MBN}_k^{(t)}\widetilde{\MBN}_k^{(t)T}\|^2_F$
\EndWhile
\State \Return {$\MBA_k^* = \widetilde{\MBM}_k^{(t)}\widetilde{\MBM}_k^{(t)T}\MBX_k\widetilde{\MBN}_k^{(t)}\widetilde{\MBN}_k^{(t)T}$}
\EndFor
\end{algorithmic}
\end{algorithm}

\section{Simulation studies}
\label{sec:simulation}
We generate the signal matrices $\MBA_k \in \mathbb{R}^{n \times p}$ given the sample size $n$, the number of features $p$, the total signal ranks $r_k \leq \min(n, p)$, $k=1,2$, the rank of joint column structure $r_c\leq \min(r_1, r_2)$ and the rank of joint row structure $r_r\leq \min(r_1, r_2)$ in accordance with Lemma~\ref{l:model} (Supplement S3). We then set
$
\MBX_k = \MBA_k + \MBE_k,
$
where $\MBE_k$ has independent entries $e_{kij} \sim \mathcal{N}(0, \sigma_k^2),\ i \in \{1,\cdots,n\},\ j \in \{1,\cdots,p\}$. We define the signal to noise ratio as $$\text{SNR} =\frac{\|\MBA_k\|^2_F}{\mathbb{E}(\|\MBE_k\|^2_F)} = \frac{r_k}{np\sigma^2_k},$$ 
and choose $\sigma_k$ to control the SNR at pre-specified levels. 

\subsection{Rank estimation}
\label{s:Rankest}
We investigate the performance of profile likelihood (PL) method from Sections~\ref{s:pl}--\ref{s:joint} on estimating the total signal ranks $r_k$, the joint column rank $r_c$ and the joint row rank $r_r$. We compare with permutation method used in \textsf{r.jive} package \citep{JIVE_R} and Bi-Cross-Validation (BCV) method \citep{BCV} implemented in \textsf{SLIDE} package \citep{SLIDE_R}.  JIVE and SLIDE rank selection methods are applied in two ways: (i) column space decomposition based on matched rows (samples); (ii) row space decomposition based on matched columns (features). When we perform JIVE on matched rows, we denote it as JIVE (Row), similarly for SLIDE. For total rank estimation, we also consider edge distribution (ED) method \citep{onatski2010determining}. We implement ED method in R ourselves by translating python code from \citet{shu2020d}. For joint rank estimation, we also consider the Wedin threshold method \citep{AJIVE} as implemented in \textsf{ajive} R package \citep{AJIVE_R}. The Wedin method is applied with the given true total ranks rather than estimated total ranks as the latter is not implemented in \textsf{ajive}.

We consider three settings, with 140 replications for each.
\begin{description}
  \item[\textbf{Setting 1}] $n = 240, p = 200$, $\mbox{SNR}=1$, $r_1, r_2$ sampled from $\{2,3,\cdots,20\}$ with replacement and $r_c, r_r$ sampled from $\{1,2,\cdots,\min(r_1, r_2,5)\}$ with replacement. 
  \item[\textbf{Setting 2}]Same as \textbf{Setting 1} with $\mbox{SNR}=0.5$.
  \item[\textbf{Setting 3}] $n = 240, p = 200$, $\mbox{SNR}=1$, $r_1, r_2$ sampled from $\{2,3,\cdots,20\}$ with replacement. In the first 35 replications, $r_c, r_r = 0$. In the next 35 replications, $r_c = 0, r_r = \min(r_1, r_2)$. In the third 35 replications, $r_c = \min(r_1, r_2), r_r = 0$. For the last 35 replications, $r_c = r_r = \min(r_1, r_2)$. This setting is used to demonstrate cases where there is either no joint structure or no individual structure. 
\end{description}
Figure~\ref{fig:set1_total} displays the difference between the estimated total rank and the true rank $r_k$ for each method in Setting 1. ED works the best, followed by the proposed PL. The permutation approach in \textsf{r.jive} works poorly in this setting, and is also not consistent (different ranks are estimated depending on whether the matching is done by rows or by columns). SLIDE rank estimation based on BCV also works poorly, however this is likely due to automatic centering implemented in the package which will perturb the column-space of the true non-centered signal. Figure~\ref{fig:set1_joint} displays the difference between the estimated joint rank and the true joint rank (either $r_c$ or $r_r$). The Wedin bound method works perfectly, however it uses the knowledge of true total ranks. The proposed PL works as well as Wedin bound without such knowledge with the exception of two cases. Both methods are significantly more accurate compared to other approaches. 
The results in Setting 2 are qualitatively similar to results in Setting 1 (Supplement S3), however the performance tends to be worse due to lower SNR. On total rank estimation, the median performance of ED and PL is still superior to the permutation method used in \textsf{r.jive} package and to BCV , however ED tends to underestimate the total ranks, whereas PL tends to overestimate the total ranks. On joint ranks estimation, the Wedin bound method still works perfectly. PL estimates joint ranks perfectly over 90\% of the times, however significantly overestimates the rank in remaining cases. JIVE on average correctly estimates the joint ranks but has higher IQR compared to PL. SLIDE consistently underestimates the joint ranks, which is likely again due to its automatic centering within bi-cross-validation.

\begin{figure}[!t]
\centering
\begin{subfigure}{0.54\textwidth}
\includegraphics[scale = 0.31]{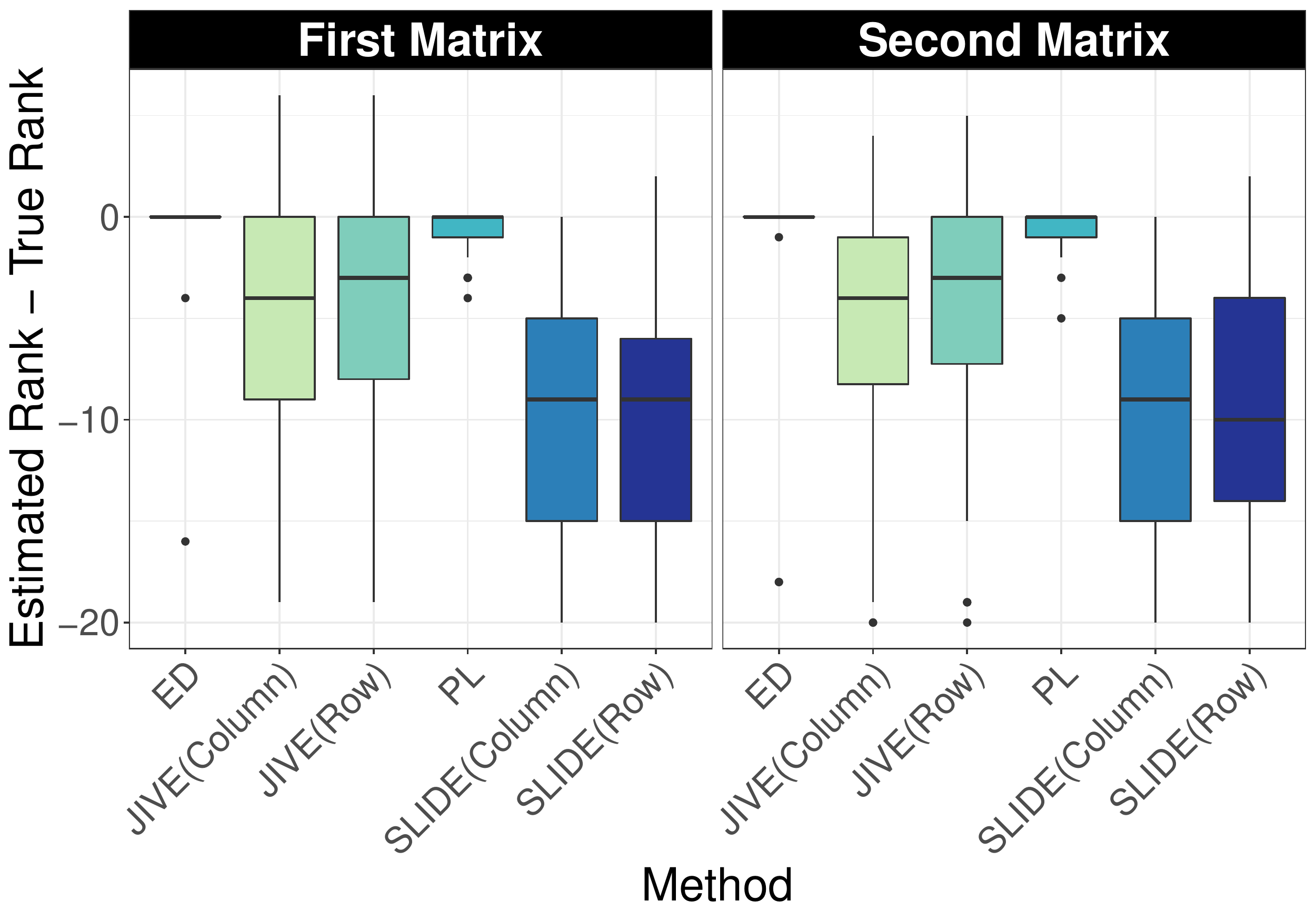}
\caption{Total rank estimation}
\label{fig:set1_total}
\end{subfigure}%
\begin{subfigure}{0.46\textwidth}
\includegraphics[scale = 0.29]{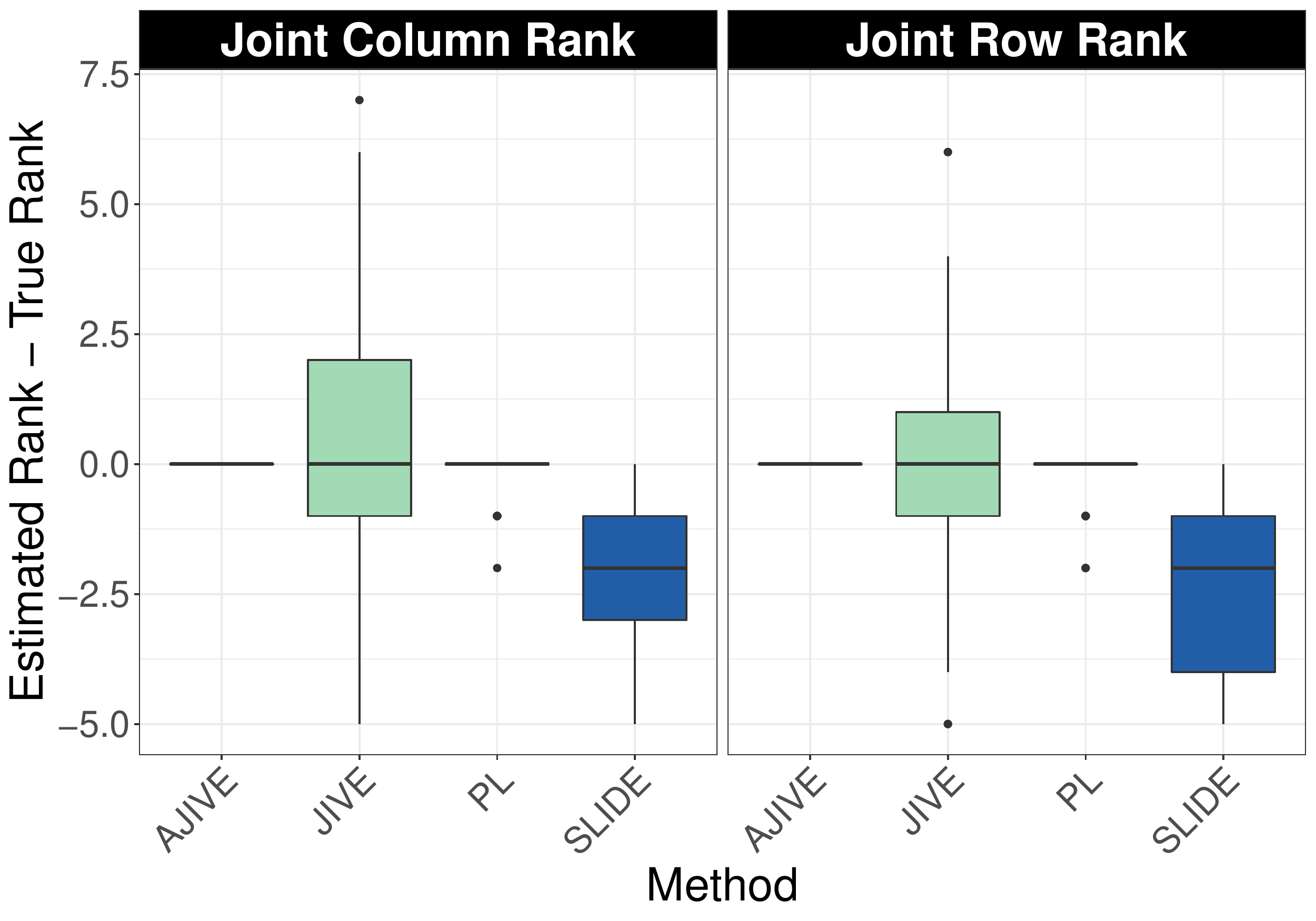}
\caption{Joint rank estimation}
\label{fig:set1_joint}
\end{subfigure}
\caption{Comparison of rank estimation for Setting 1 over 140 replications. $n = 240, p = 200$, $2 \leq r_1, r_2 \leq 20, 1 \leq r_c, r_r \leq 5$, $\mbox{SNR} = 1$. JIVE (Column) or SLIDE (Column) estimates the total rank when columns are matched and vice versa.}
\end{figure}

Figures~\ref{fig:set3_total} and~\ref{fig:set3_joint}
show the results on total rank estimation and joint rank estimation, respectively, in Setting 3. As in Setting~1, ED and PL methods are significantly more accurate in estimating total ranks $r_k$ compared to the permutation approach, and the results of the latter are again dependent on whether the matching is based on rows or columns. Unlike Setting~1, PL is slightly better than ED, as occasionally ED grossly underestimates the total rank. On joint rank estimation, Wedin bound works best, however it uses the knowledge of true total ranks. All methods correctly identify zero joint rank when $r_c=r_r=0$. Overall PL is more accurate than JIVE and SLIDE when $r_r=r_c=r_{\min}$, however in few cases it underestimates the joint rank more severely than JIVE. When $r_c=0,\ r_r = r_{\min}$, PL and JIVE are comparable in estimating $r_r$ on average, but PL has lower variance across replications. When $r_c = r_{\min},\ r_r=0$, PL slightly underestimates $r_c$ compared to the permutation approach.

\begin{figure}[!t]
\begin{subfigure}{0.53\textwidth}
\includegraphics[width=0.99\textwidth]{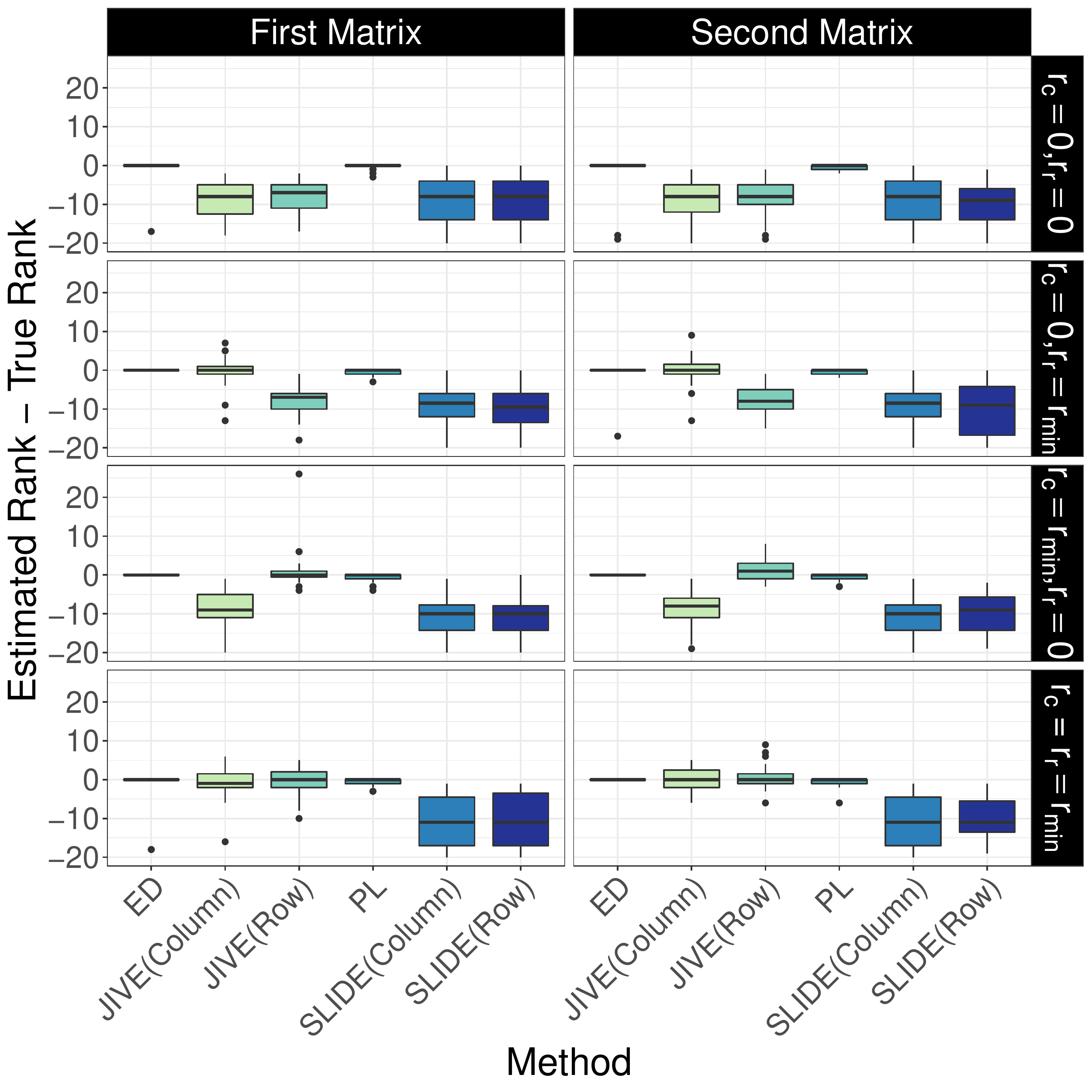}
\caption{Total rank estimation}
\label{fig:set3_total}
\end{subfigure}%
\begin{subfigure}{0.47\textwidth}
\includegraphics[width=0.99\textwidth]{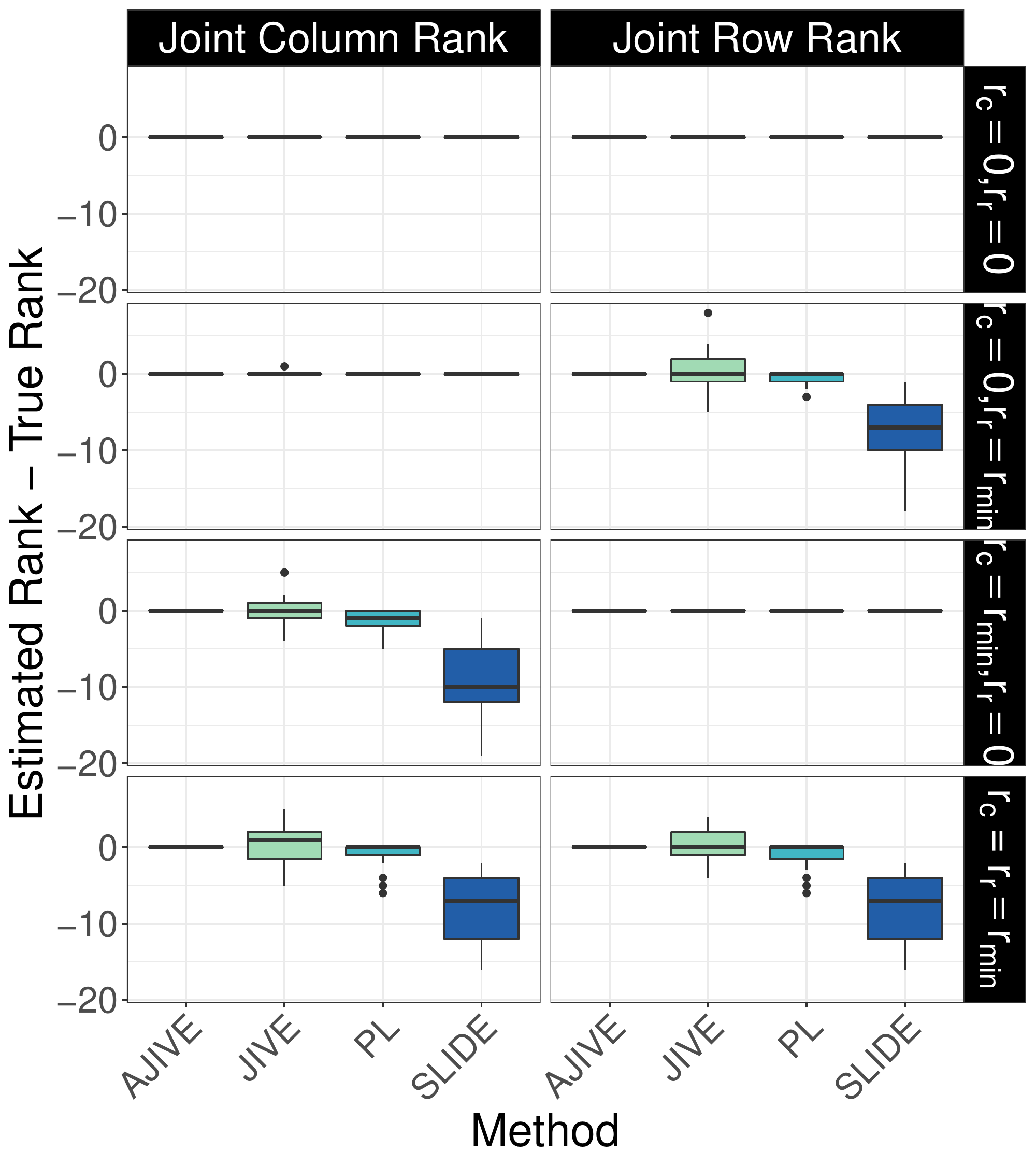}
\caption{Joint rank estimation}
\label{fig:set3_joint}
\end{subfigure}
\caption{Comparison of rank estimation for Setting 3 over 140 replications. $n = 240, p = 200, 2 \leq r_1, r_2 \leq 20, r_{min} = \min(r_1,r_2)$. 
}
\end{figure}

Overall, we found that ED and PL work best in total rank estimation. While in Setting~1 ED works better than PL, ED assumes that the maximal possible signal rank is bounded by $0.1\min(n,p)$ \citep{ahn2013eigenvalue}, and this assumption is satisfied in all of our settings. Since in practice this assumption may be violated, we use PL method as default. On joint rank estimation, Wedin bound method works perfectly in all of the settings, however it does so by using true total ranks. Since in practice the true ranks are unknown, and PL works similarly, we use PL as default. We reach the same conclusions in a high-dimensional setting that mimics TCGA data set in Section~\ref{s:TCGA} (Supplement S3).

\subsection{Signal identification}
\label{s:Signalid}
We investigate the performance of DMMD on estimating signals $\MBA_k$ in model~\eqref{eq:dmmd} if the true ranks are known. We also consider DMMD-i (Supplement S5) which allows iterative adjustment of initial $\MBM$ and $\MBN$ from Step 2. We measure the performance as
$$
\text{Relative Error}(\widehat \MBA_k, \MBA_k) = \frac{\|\widehat {\MBA}_k - \MBA_k\|^2_F}{\|\MBA_k\|^2_F}.
$$
We also measure the relative error separately on joint $\MBJ_{ck}$, $\MBJ_{rk}$ and individual $\MBI_{ck}$, $\MBI_{rk}$ in model \eqref{eq:dmmd}. We compare with JIVE \citep{JIVE}, SLIDE \citep{SLIDE} and AJIVE \citep{AJIVE} using the same implementation as in Section~\ref{s:Rankest}.
AJIVE, JIVE and SLIDE are fitted in two ways: (i) column space decomposition based on matched rows (samples); (ii) row space decomposition based on matched columns (features). We use JIVE (Row) to indicate model based on matched rows, and similarly JIVE (Column).

We consider three settings with 140 replications for each.
\begin{description}
\item[\textbf{Setting 4}] $n = 240, p = 200$, $r_1 = 20, r_2 = 18, r_c = 4, r_r = 3$, $\mbox{SNR} = 1$.
\item[\textbf{Setting 5}] Same as \textbf{Setting 4} with $\mbox{SNR}=0.5$.
\item[\textbf{Setting 6}] $n = 240, p = 200$, $\mbox{SNR} = 1$, $r_1 = 20, r_2 = 18$. In the first 35 replications, $r_c, r_r = 0$. In the next 35 replications, $r_c = 0, r_r = \min(r_1, r_2) = r_2$. In the third 35 replications, $r_c = \min(r_1, r_2) = r_2, r_r = 0$. For the last 35 replications, $r_c = r_r = \min(r_1, r_2) = r_2$. Like Setting 3 in Section~\ref{s:Rankest}, this setting is used to demonstrate cases where there is either no joint structure, or no individual structure.
\end{description}
All DMMD, JIVE, AJIVE and SLIDE use true total ranks $r_k$, true joint rank $r_c$ (for column decomposition) and true joint rank $r_r$ (for row decomposition), $k=1,2$, as input. Ranks misspecifications are investigated in Supplement S3, where DMMD performs best in total signal estimation.

Figures~\ref{fig:set4_col} and~\ref{fig:set4_row} show relative errors of all methods in Setting 4 for estimated signals based on matched rows ($\MBJ_{ck},\MBI_{ck}$) and matched columns ($\MBJ_{rk},\MBI_{rk}$), respectively. The errors for total signal are the same for DMMD as it enforces model~\eqref{eq:dmmd}. In contrast, the errors for JIVE, AJIVE and SLIDE depend on matching (by rows or by columns) as it affects the estimated signal. For joint signals, DMMD and SLIDE perform similar, and are both more accurate than JIVE and AJIVE. DMMD has the smallest errors on full signals and individual signals in all scenarios, confirming that taking into account double matching leads to more accurate signal estimation.  The same conclusion holds in Setting 5 with smaller SNR (see Supplementary Materials Section 2). 

\begin{figure}[!t]
\begin{subfigure}{0.5\textwidth}
\includegraphics[width=1\textwidth]{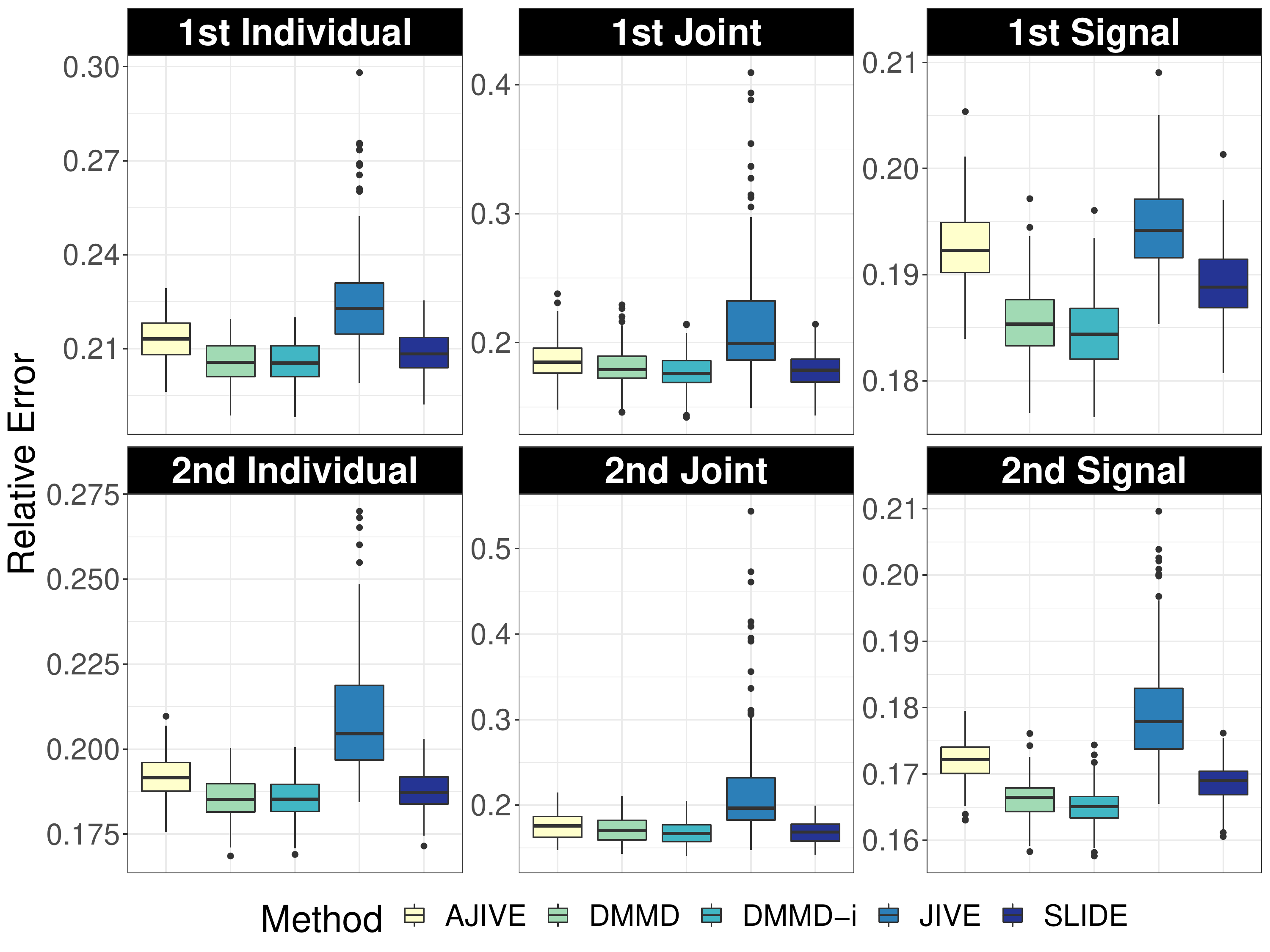}
\caption{Column space decomposition (matched rows)}
\label{fig:set4_col}
\end{subfigure}%
~
\begin{subfigure}{0.5\textwidth}
\includegraphics[width=1\textwidth]{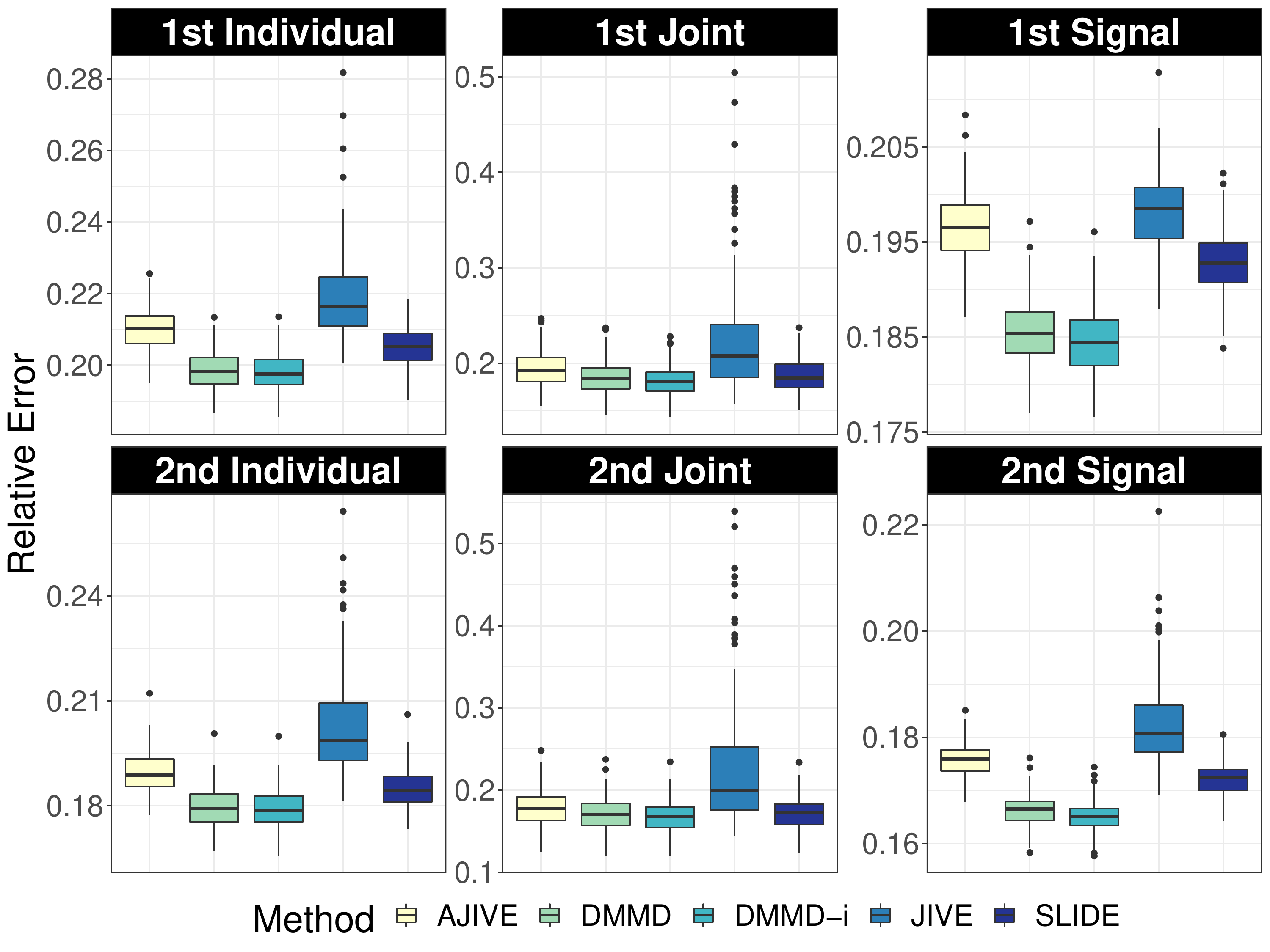}
\caption{Row space decomposition (matched columns)}
\label{fig:set4_row}
\end{subfigure}
\caption{Comparison of signal identification for Setting 4 over 140 replications, $n = 240,\ p = 200$, $r_1 = 20,\ r_2 = 18,\ r_c = 4,\ r_r = 3$, $\mbox{SNR} = 1$. }
\end{figure}

In Setting 6, either joint signal matrix or individual signal matrix is exactly equal to zero, thus we use the absolute error, $\|\text{Estimated Signal} - \text{True  Signal}\|^2_F$, rather than the relative error to measure the performance. Figures~\ref{fig:set6_col} and~\ref{fig:set6_row} show absolute errors of the four methods based on column space decomposition due to matched rows ($\MBJ_{ck},\MBI_{ck}$), or row space decomposition due to matched columns ($\MBJ_{rk},\MBI_{rk}$), respectively. When both joint column and row structures are absent ($r_c = r_r =0$), all methods perform the same, which is expected as the estimation is completely separate across two views. When $r_c=r_r=r_2$, DMMD gives smallest errors as it takes advantage of double matching. When $r_c = 0, r_r = r_2$, JIVE and SLIDE work better than DMMD on estimating joint row structure, whereas when $r_r = 0, r_c = r_2$, they work better than DMMD on estimating joint column structure. 
A possible explanation for this is a different approach for estimating the joint structures used by the methods. DMMD uses element-wise averaging of pairs of basis vectors from each view with smallest principal angles as in AJIVE \citep{AJIVE}, whereas JIVE and SLIDE extract basis vectors from concatenated matrix of view-specific residuals after subtracting individual structures. 

Overall, we find that DMMD has the smallest signal estimation error, with DMMD-i being slightly better than DMMD. This remains true in a high-dimensional setting that mimics TCGA data set in Section~\ref{s:TCGA} (Supplement S3). While in Setting 6 JIVE and SLIDE sometimes lead to better performance, these cases correspond to absent individual structures in either row or column directions and absent joint structures in the other directions, which is rarely the case for real data. When individual structures are present, DMMD always leads to improved errors as it enforces equality in estimated total signals from row and column decomposition of double-matched data, which subsequently leads to more accurate estimation of individual structures, and consequently, of the total signal.

\begin{figure}[!t]
\begin{subfigure}{0.48\textwidth}
\includegraphics[width=1\textwidth]{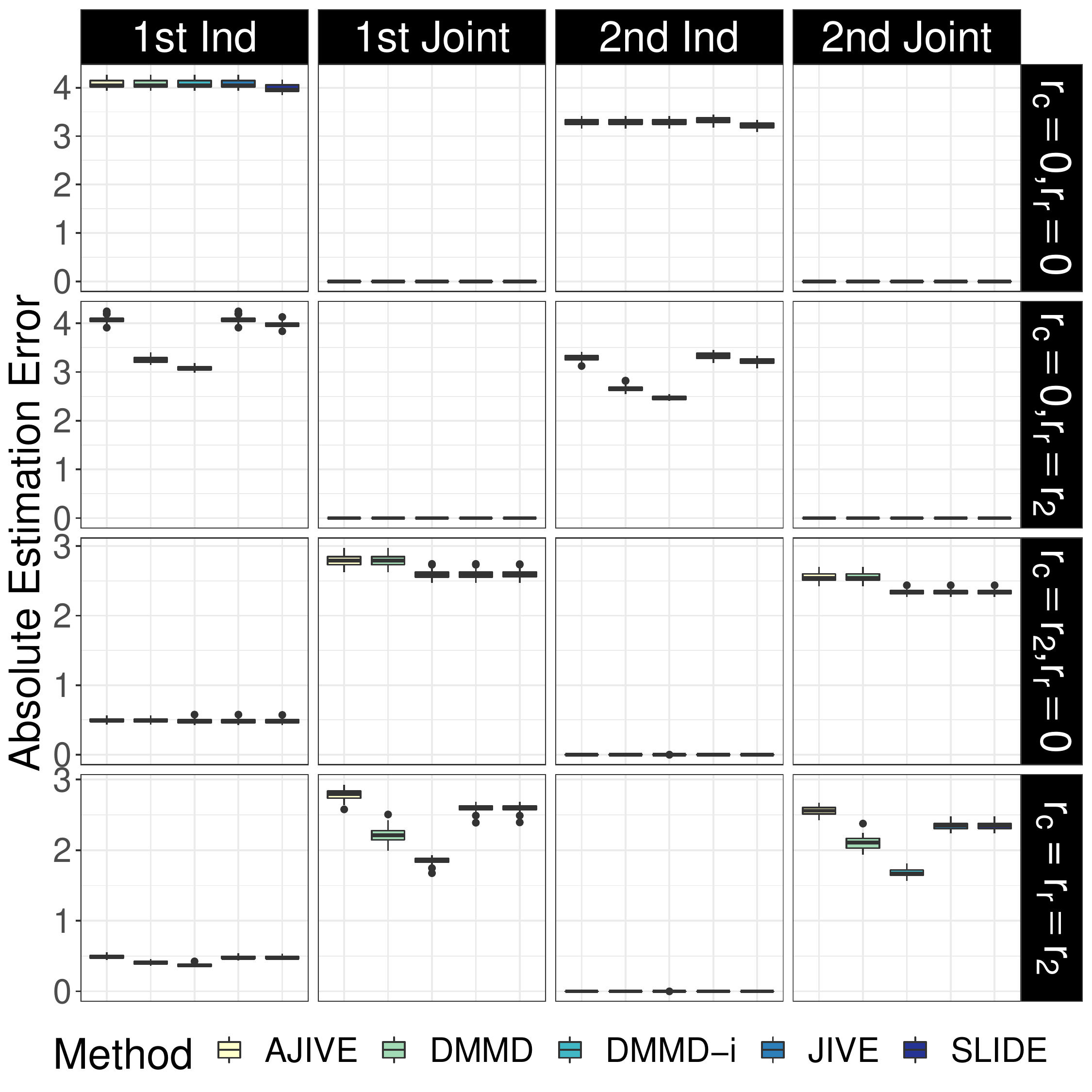}
\caption{Matched rows decomposition}
\label{fig:set6_col}
\end{subfigure}%
~
\begin{subfigure}{0.48\textwidth}
\includegraphics[width=1\textwidth]{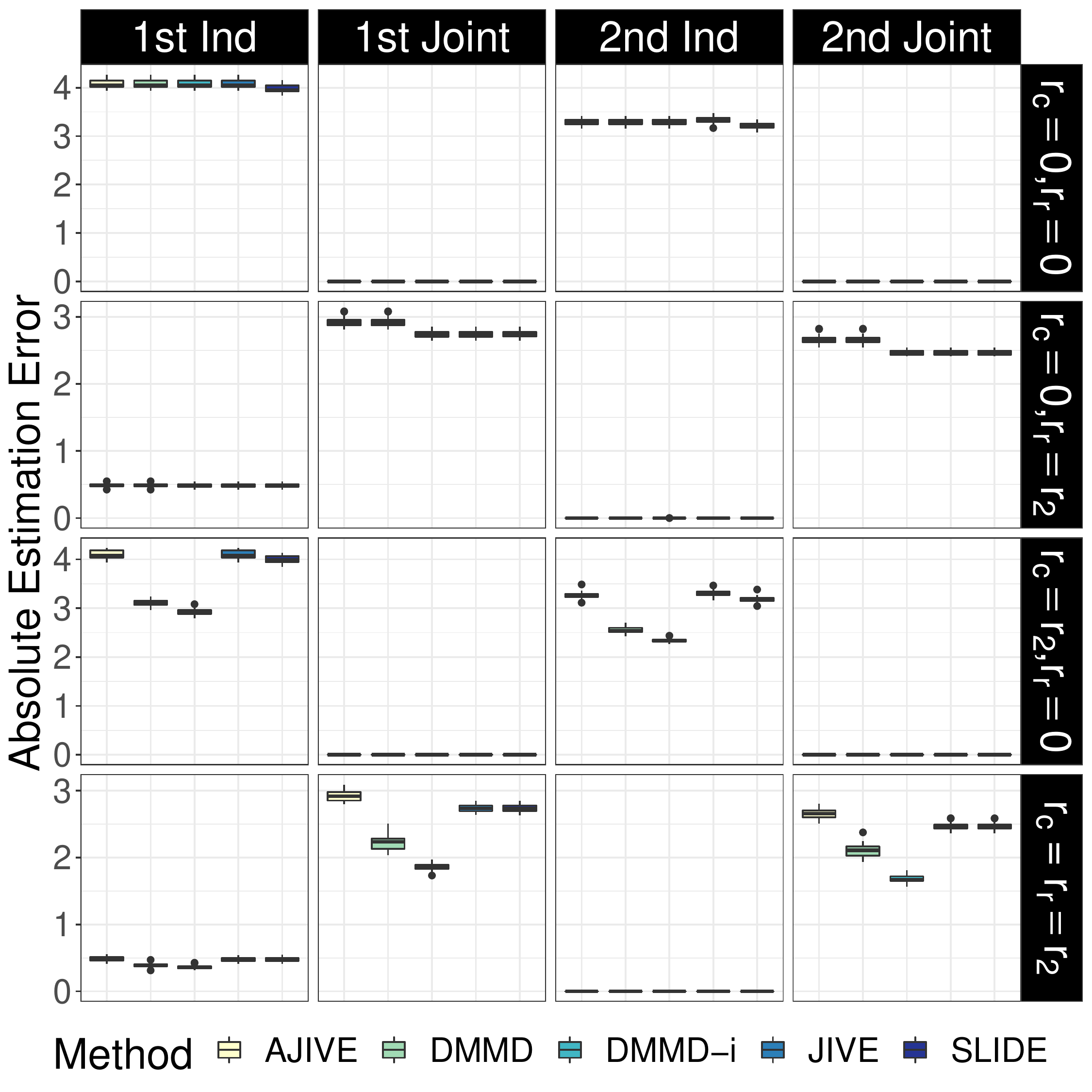}
\caption{Matched columns decomposition}
\label{fig:set6_row}
\end{subfigure}
\caption{Comparison of signal identification for Setting 6 over 140 replications. $n = 240,\ p = 200,\ r_1 = 20,\ r_2 = 18,\ \mbox{SNR} = 1$.}
\end{figure}

Computationally, DMMD is significantly faster than the competitors; its total run time (rank estimation plus signal identification) on data with $n= 100$, $p=800$ is 24 seconds on Intel(R) Core(TM) i5-7300U CPU @ 2.60GHz. As expected, DMMD-i comes with a significantly higher computational cost taking 618 seconds on the same data. More details are in Supplement S3. Given the relatively small improvement of DMMD-i over DMMD, in practice we recommend to use DMMD with large datasets for computational efficiency.

\section{Application}
\label{sec:application}
\subsection{Application to TCGA data}
\label{s:TCGA}

We consider data from The Cancer Genome Atlas (TCGA) repository corresponding to the Breast Invasive Carcinoma (BRCA) cancer type. We use \textsf{TCGA-Assembler 2} software pipeline \citep{wei2018tcga} to obtain miRNA read counts corresponding to the primary tumor tissue (view~1) and to the normal tissue (view~2) of the same subjects. We log-transform the counts to account for skewness, and remove the samples and features with zero variance for both views. We then apply double-standardization as in \citet{efron2012large} so that all rows and columns in each matrix have mean zero and sample variance one; such processing of multi-view data is also used in \citet{risk2021sing}. The double-standardization makes the estimated decomposition mean and scale-invariant. The final double-matched $\MBX_1$ (primary tumor tissue) and $\MBX_2$ (normal tissue) each contain $p=734$ matched miRNAs from $n=87$ matched samples. To evaluate possible biological relevance of obtained decomposition, we match each sample with one of the cancer subtypes: Luminal A (LumA), Luminal B (LumB), Basal-like (Basal), HER2-enriched (H) obtained from \url{https://www.cbioportal.org/study/clinicalData?id=brca_tcga_pub}. For 46 out of 87 subjects there are missing clinical records which are denoted as unknown. 

Our goal is to extract common (across tissues) as well as individual (tissue-specific) signals from each view, where common/individual signals have two meanings: (i) across subjects, and (ii) across miRNAs. Here we consider two methods: JIVE (rank selection via permutation test with subsequent fitting of the JIVE model) and the proposed DMMD (rank selection via profile likelihood with subsequent fitting of model~\eqref{eq:dmmd}). For JIVE, we consider both matching by subjects (column space decomposition), and matching by miRNAs (row space decomposition). 

First, we compare the estimated ranks. Permutation approach used by JIVE leads to inconsistent total ranks as the estimates depend on the type of matching: matching by samples leads to $\widehat r_1 = 11$, $\widehat r_2 = 9$, whereas matching by miRNAs leads to $\widehat r_1 = 14, \widehat r_2 = 11$. The PL method gives smaller estimated ranks $\widehat r_1 = 8, \widehat r_2 = 6$. Despite the discrepancy in total ranks between JIVE and DMMD, both lead to the same estimated joint ranks with $\widehat r_c = 0$ and $\widehat r_r = 2$. 

Next, we compare the variance explained by each method, together with the variance explained separately by joint/individual parts of the estimated signal. 
Figure~\ref{fig:variation} shows the percent variance explained by each part of the estimated decomposition separately for tumor and normal tissues. The total variance explained by estimated signal (joint plus individual) is the same for DMMD regardless of the type of matching considered, whereas it changes for JIVE due to discrepancy in estimated signals. The overall variance explained is higher for JIVE as it estimates higher total ranks compared to DMMD. Both DMMD and JIVE show that the variance explained by joint structure is higher for normal tissue compared to primary tumor tissue. We believe that this is in agreement with what would be expected from biological knowledge since tumor tissue evolves from originally normal tissue, and becomes more heterogeneous as cancer develops. 

\begin{figure}[!t]
\includegraphics[width=0.5\textwidth]{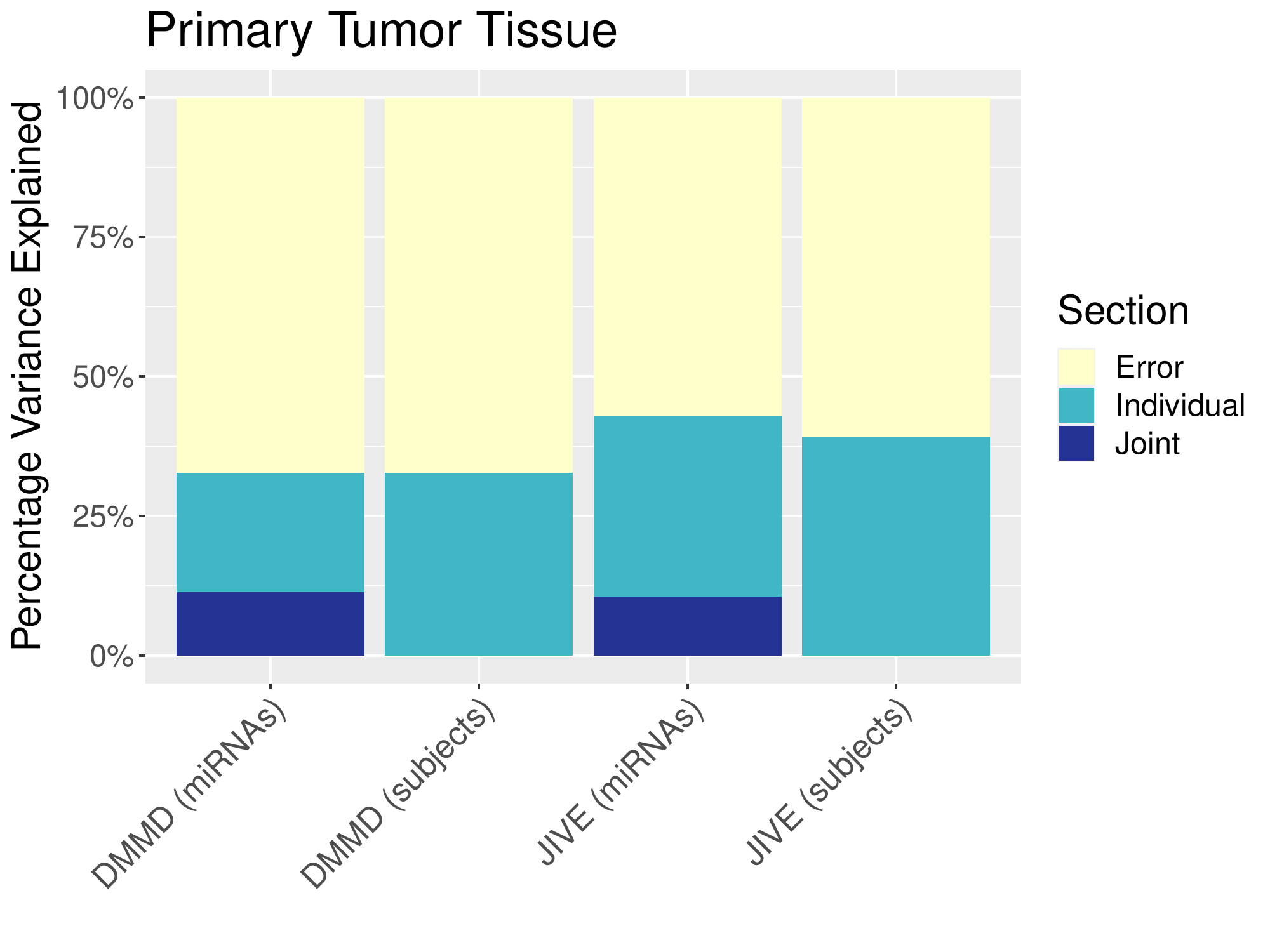}
\includegraphics[width=0.5\textwidth]{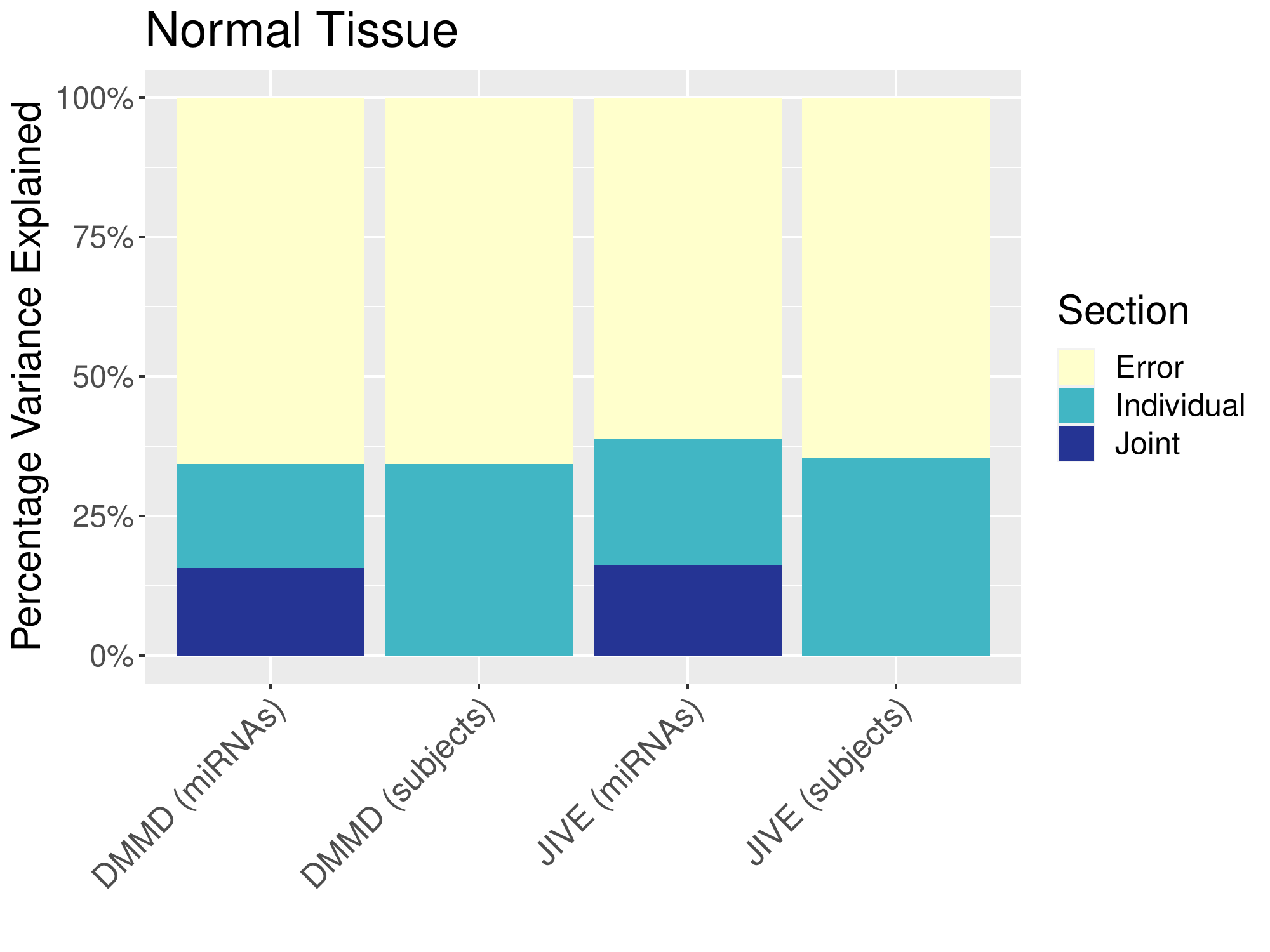}
\caption{Percentage variance explained by extracted DMMD and JIVE decompositions on TCGA BRCA matched miRNA data from primary tumor and normal tissues.}\label{fig:variation}
\end{figure}

We next display the found joint row structure ($\widehat r_r = 2$) corresponding to matched miRNAs in Figures~\ref{fig:joint_miRNA_tumor} and~\ref{fig:joint_miRNA_normal}. An alternative vertical alignment of these heatmaps is in Supplement Figure S3. The order of samples and miRNAs is the same in both tissues, and are determined based on the hierarchical clustering of the joint structure of the primary tumor tissue. Visually, the joint row structure captures the division of miRNAs in 3 clusters. While the displayed cluster partition is based on tumor tissue, the heatmap of normal tissue in Figure~\ref{fig:joint_miRNA_normal} has block structure based on the same partition.

\begin{figure}[!t]
\begin{subfigure}{0.5\textwidth}
\includegraphics[width =1\textwidth]{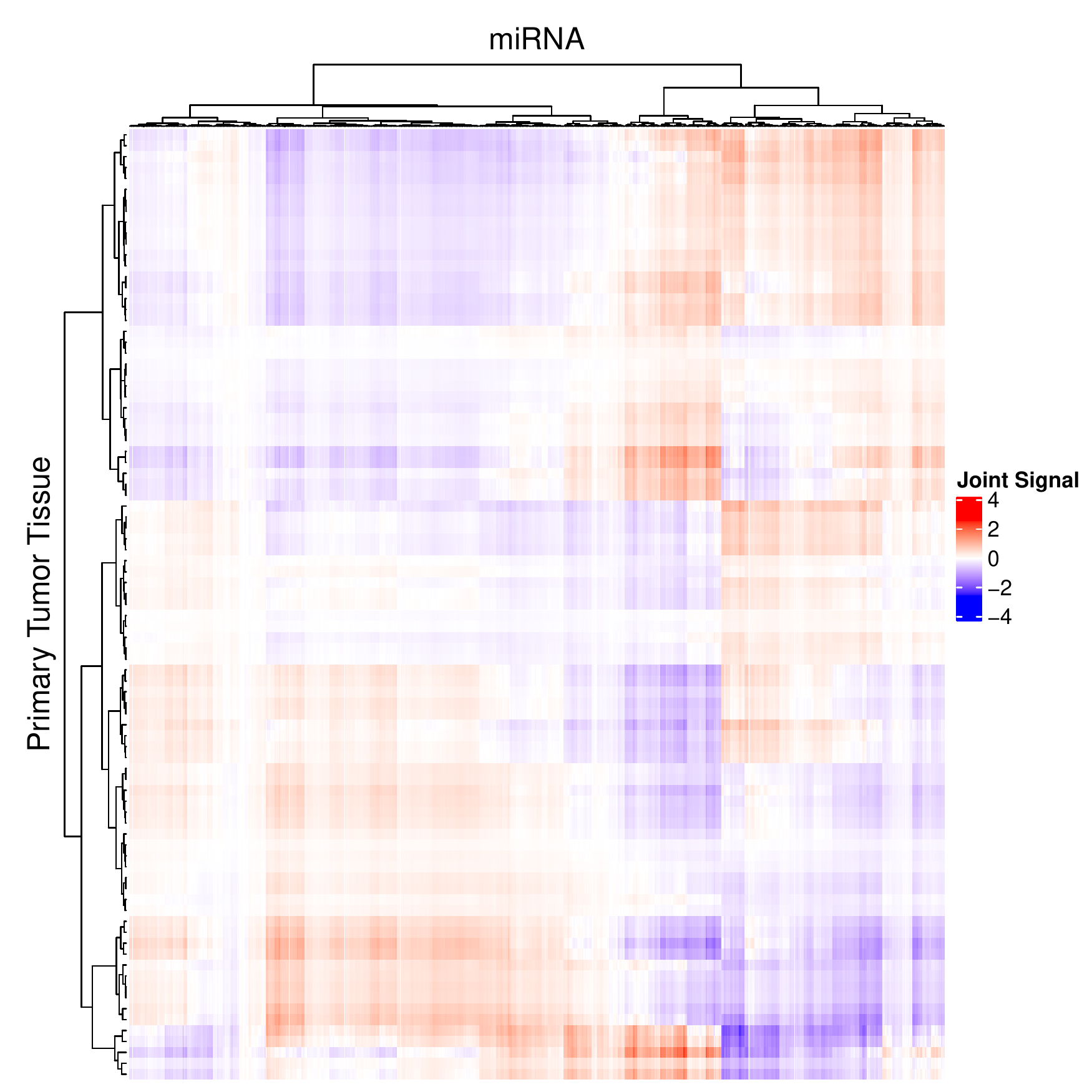}
\caption{}
\label{fig:joint_miRNA_tumor}
\end{subfigure}%
\begin{subfigure}{0.5\textwidth}
\includegraphics[width=1\textwidth]{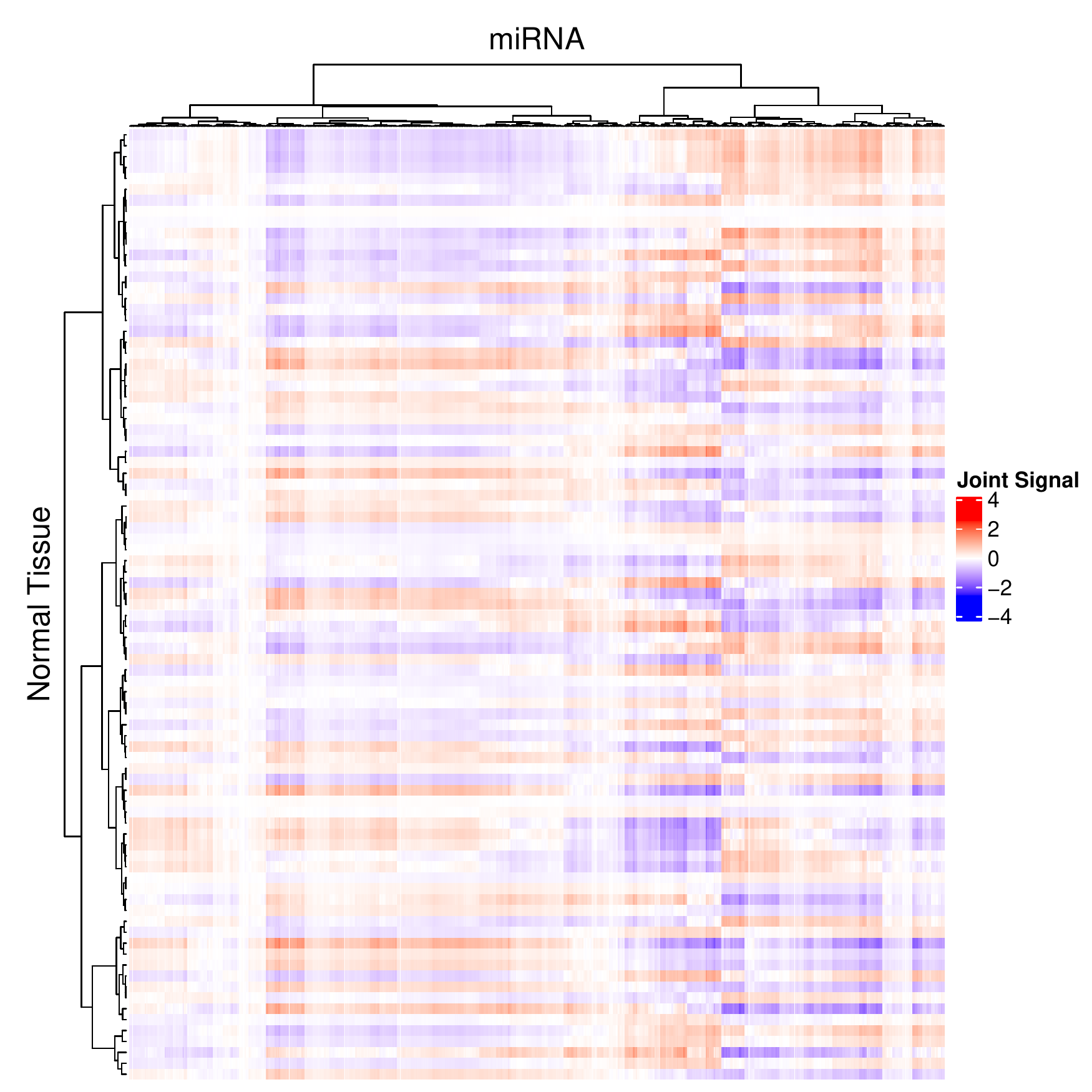}
\caption{}
\label{fig:joint_miRNA_normal}
\end{subfigure}
\caption{Joint row (miRNA) structures extracted by DMMD for primary tumor and normal tissues from matched TCGA-BRCA miRNA data. The order of samples and miRNAs in both figures is the same, which is determined by the joint structure of primary tumor tissue.}
\end{figure}

To provide further interpretation of estimated decomposition, we next consider individual structures from DMMD with respect to matched subjects. Figure \ref{fig:TCGA_ind_row} displays the heatmaps of estimated $\MBI_{ck}^{\top}\in \R^{p \times n}$, $k=1,2$, from model~\eqref{eq:dmmd} with samples sorted according to subtype information. Here $\rank(\MBI_{c1}) = 8$ and $\rank(\MBI_{c2}) = 6$. To better visualize the corresponding individual column spaces, Figure~\ref{fig:TCGA_ind_rowb} shows the leading left singular vectors of $\MBI_{c1}$ and $\MBI_{c2}$, respectively, which are the leading basis vectors for corresponding individual column spaces. For the primary tumor tissue, the basis vector displays a strong contrast between the Basal and LumA subtypes, expresses 27.3\% of the whole variation in individual structure. The effect of this basis vector can be seen in the whole signal heatmap in Figure \ref{fig:TCGA_ind_row}, where the contrast in the same direction is observed roughly in the top half of the miRNAs in primary tumor tissue, and is observed in the opposite direction in the bottom half of the miRNAs. In contrast, the individual structure for normal tissue does not display this contrast. Furthermore, the leading basis vector for the individual structure of normal tissues does not appear to separate any of the cancer subtypes, which is in agreement with what would be expected from biological knowledge since it captures individual subjects structure in normal tissues that is not present in primary tumor tissue.

\begin{figure}[!t]
\begin{subfigure}{1\textwidth}
\includegraphics[width=0.5\textwidth]{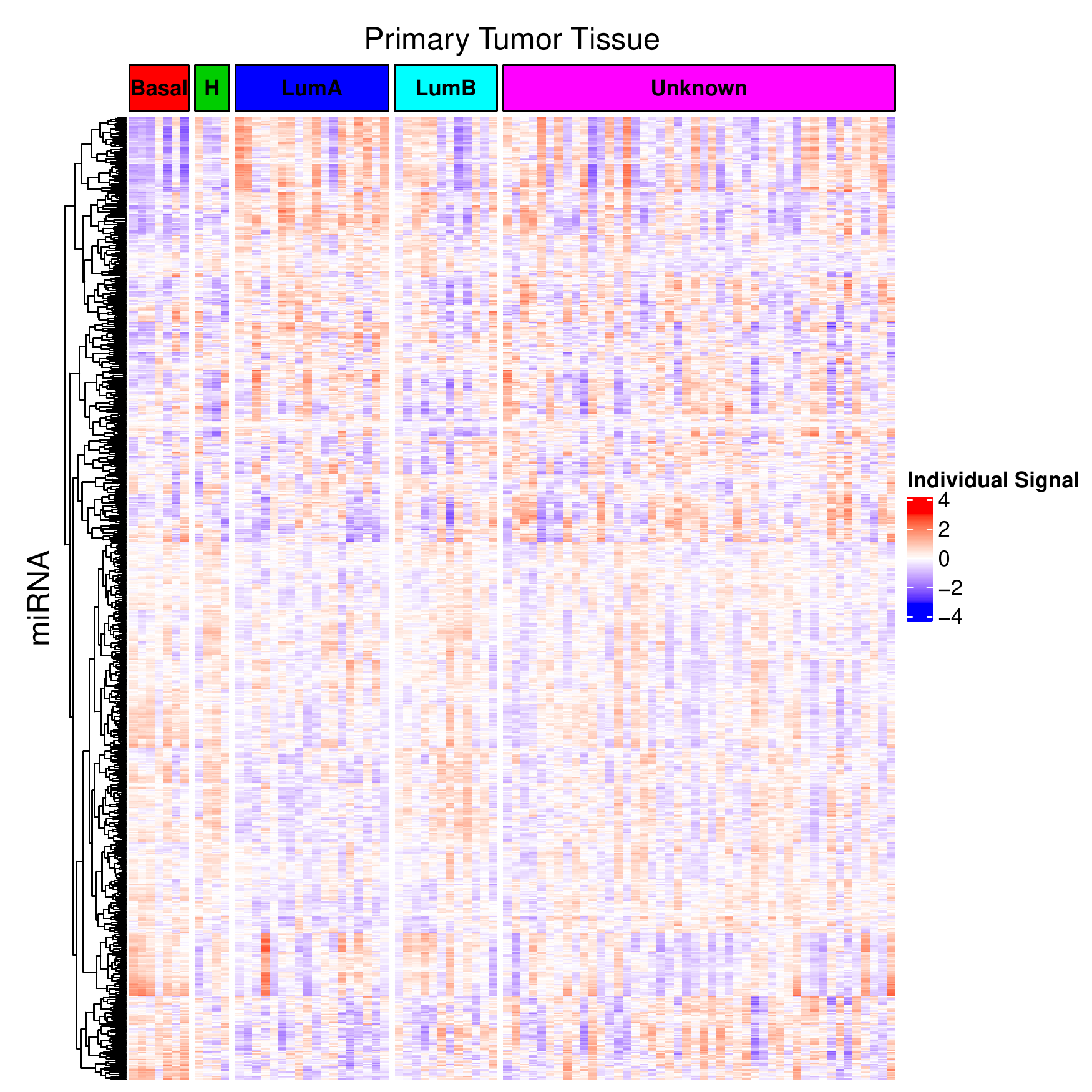}
\includegraphics[width=0.5\textwidth]{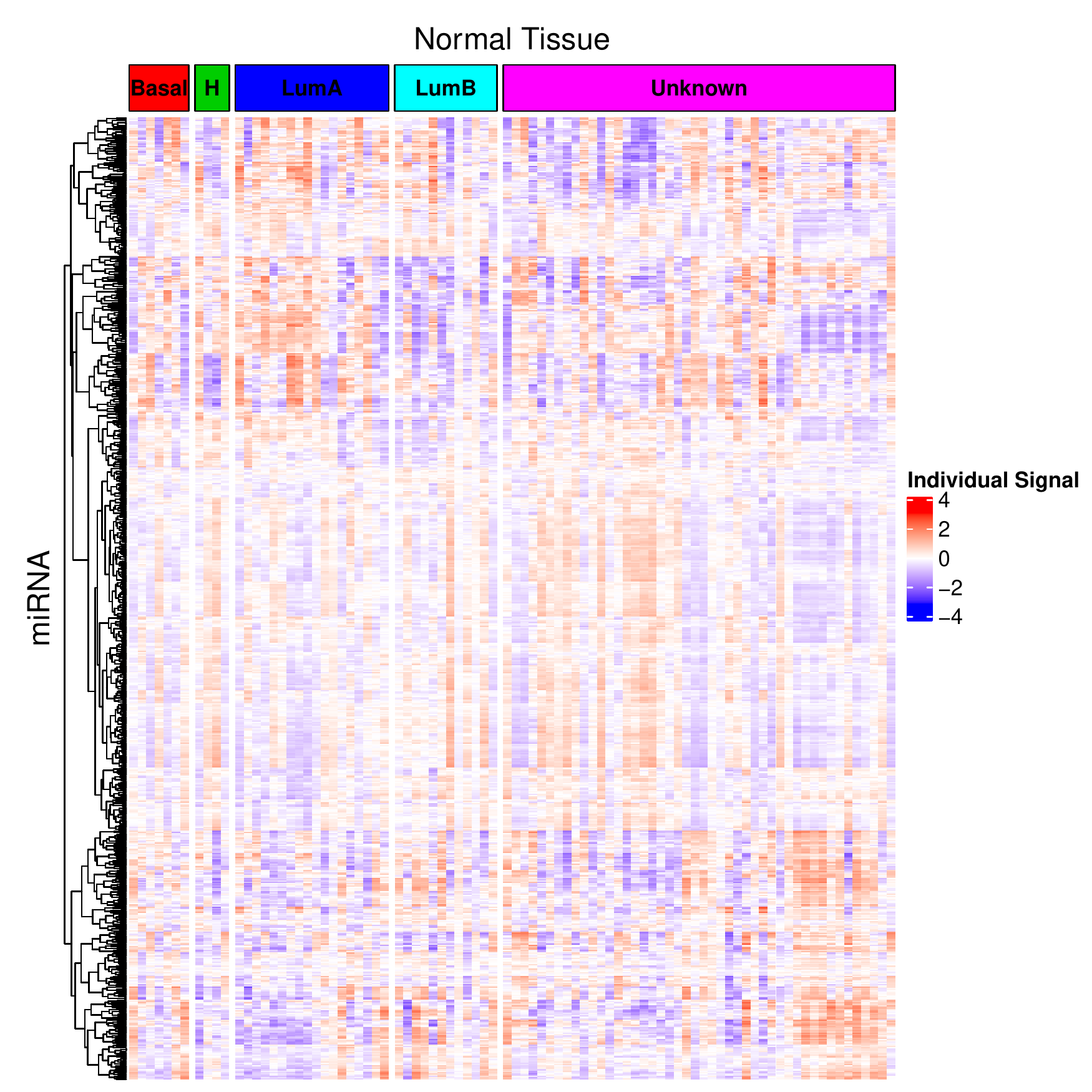}
\caption{Individual signal matrices (subjects)}
\label{fig:TCGA_ind_row}
\end{subfigure}
\begin{subfigure}{1\textwidth}
\includegraphics[width=0.5\textwidth]{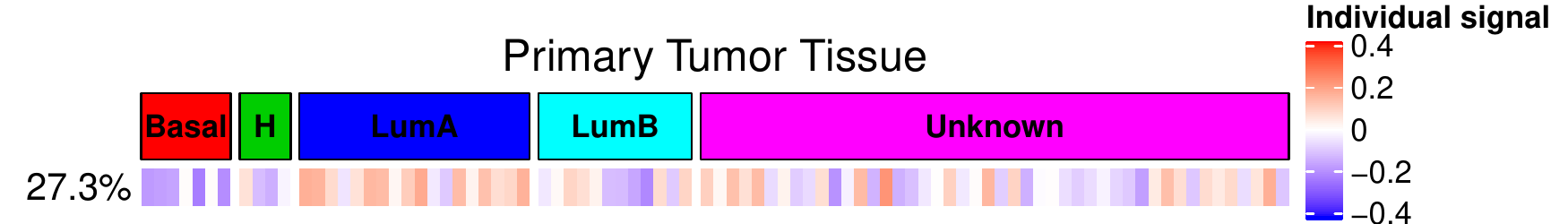}
\includegraphics[width=0.5\textwidth]{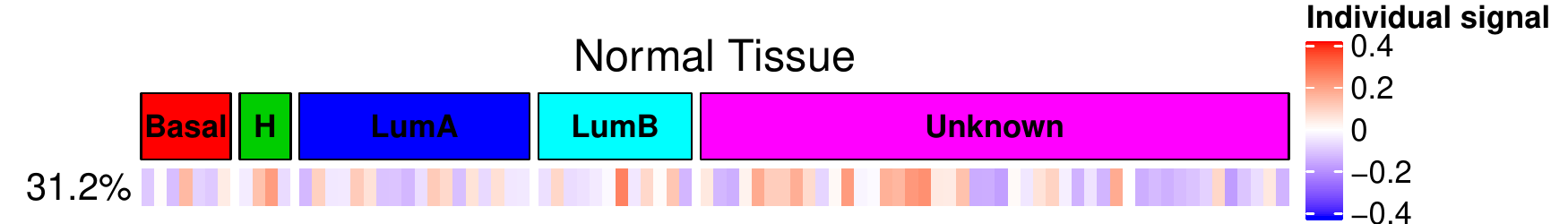}
\caption{Leading basis vectors for individual signals (subjects)}
\label{fig:TCGA_ind_rowb}
\end{subfigure}
\caption{Individual column (matched subjects) structures extracted by DMMD for primary tumor and normal tissues from double-matched TCGA-BRCA miRNA data. The samples are ordered according to the cancer subtype. The top row shows full individual signals, whereas the bottom row shows the leading individual basis vectors along with the percentage of variance explained (relative to the full individual signal).}
\end{figure}

\subsection{Application to soccer data}
\label{s:soccer}
We consider data from soccer matches in the English Premier League obtained from \url{https://www.kaggle.com/kenzeng24/premier-league-matches}. 
Each row of data represents a soccer match played in the English Premier League, and each column represents a unique feature recorded for that match (e.g. Date, Home Team, Full Time Goals for each team, etc). First, we filter the data by removing the matches with data quality issues (decimal values recorded for the number of yellow cards or red cards), and removing the matches corresponding to draw games. Secondly, we determine the winning and losing team for each of the matches, and extract ten match statistics recorded for each team representing numbers of full-time goals, half-time goals, shots, shots on target, hit woodwork, corners, fouls, offsides, yellow cards, and red cards. To aid interpretation of extracted signals in terms of original untransformed match statistics, we do not apply double standardization from  Section~\ref{s:TCGA} to soccer data. In the end, we obtain two double-matched $558\times 10$ matrices, one for the winning team and one for the losing team, where each row corresponds to a match, and each column corresponds to a team statistic from the match. Our goal is to investigate the relationship between match statistics that are (i) common across teams, and (ii) individual to the winning team.

First, we estimate the ranks of underlying signals. Both PL and ED automatically select $\widehat r_1 = \widehat r_2 = 1$, and PL determines $\widehat r_c = \widehat r_r = 1$, which is one of the special cases considered in simulation setting 6. DMMD works best in that setting, so we only report DMMD results. Table~\ref{t:Soccer11} displays the coefficient of the extracted row basis vector normalized to have coefficient 1 for Full time goals to assist interpretation. As this basis vector corresponds to joint structure across both winning and losing teams, we conclude that in English Premier League there is on average 1 goal for every 8.08 shots in the game. Also, on average, there are more goals in the second half game compared to the first half game. 

\begin{table}[!t]
\caption{Joint row basis for winning and losing teams in English Premier League when $r_1 = r_2 = r_r = 1$.}
\begin{center}
\footnotesize
\begin{tabular}{p{1.0cm}p{1.0cm}p{0.9cm}p{0.9cm}p{1.1cm}p{1.1cm}p{1.2cm}p{1.2cm}p{1.2cm}p{1cm}p{1cm}}
Signal & Full Time Goals & Half Time Goals & Shots & Shots on Target & Hit Woodwork & Corners & Fouls Commited & Offsides & Yellow Cards & Red Cards\\\hline
Joint & 1.00 & 0.45 & 8.08 & 3.90 & 0.23 & 4.03 & 9.80 & 2.50 & 1.10 & 0.07 \\
\end{tabular}
\label{t:Soccer11}
\end{center}
\end{table}

\begin{table}[!t]
\caption{Joint row basis and individual row basis for winning teams in English Premier League when $r_1 = 2, r_2 = r_r = 1$.}
\begin{center}
\footnotesize
\begin{tabular}{p{1.0cm}p{1.0cm}p{0.9cm}p{0.9cm}p{1.1cm}p{1.1cm}p{1.2cm}p{1.2cm}p{1.2cm}p{1cm}p{1cm}}
Signal & Full Time Goals & Half Time Goals & Shots & Shots on Target & Hit Woodwork & Corners & Fouls Commited & Offsides & Yellow Cards & Red Cards \\\hline
Joint & 1.00 & 0.47 & 7.85 & 3.78 & 0.21 & 4.01 & 11.35 & 2.76 & 1.30 & 0.08 \\
Win & 1.00 & 0.33 & 5.09 & 2.93 & 0.21 & 1.62 & -4.98 & -0.39 & -0.88 & -0.08 \\
\end{tabular}
\label{t:Soccer21}
\end{center}
\end{table}

The total ranks estimated by PL and ED are quite low due to low $p=10$. ED inherently restricts the signal rank to be at most $0.1p$, and thus can not possibly estimate larger ranks for these data. PL relies on clustering singular values in two groups, which we suspect is less reliable when the number of singular values is small. Thus here we consider an alternative approach to rank estimation, where for each team we pick the rank to explain 90\% of the variation in the respective dataset. This approach leads $\widehat r_1 = 2$, $\widehat r_2=1$, $\widehat r_r=1$. Thus there is a rank 1 individual structure in winning team not present in the losing team. Table~\ref{t:Soccer21} displays the coefficients for the joint row basis vector, and the individual row basis vector for winning team as estimated by DMMD. Both vectors are normalized to have coefficient 1 corresponding to full time goals to assist interpretation. As in the previous analyses, there are approximately 8 shots for every goal, and approximately 4 corners for every goal. However, the winning team tends to have a higher number of goals per shots while simultaneously having fewer fouls, offsides, yellow cards and red cards. The coefficient for offsides may appear counter-intuitive, however it can be interpreted as the attack of winning team being less interrupted. Somewhat surprisingly, the number of hitting woodwork does not seem to affect the winning or losing conditions. 

\section{Discussion}
\label{sec:discus}
We propose a new decomposition for multi-view data with matched rows and columns, which we call DMMD. The main novelty of our approach is in taking advantage of double-matching property via explicit column and row space constraints in signal estimation, and in deriving the corresponding optimization algorithm. The algorithm relies on estimated joint column and row spaces; the proposed estimation works well numerically but can be improved using iterative DMMD at the expense of a significantly higher computational cost. While our exposition has been limited to the case of two views, DMMD can be applied to more views as Lemma~\ref{l:model}, model~\eqref{eq:dmmd}, estimation of proxy signals in Section~\ref{s:pl} and the Algorithm~\ref{a:opt} still hold. Estimation of joint structures in Section~\ref{s:joint} requires modification for more than two views, which we outline in Supplement~S6.

The method requires estimation of ranks corresponding to different parts of the decomposition, and our empirical studies indicate that the chosen profile likelihood (PL) approach for rank estimation is competitive compared to alternative rank estimation methods. However, we also found that PL can occasionally severely overestimate the ranks, while other methods tend to underestimate the ranks. To help identify whether severe rank overestimation is an issue in practice, we recommend to simultaneously consider several rank estimation approaches (like we did in Section~\ref{sec:application}) to verify that the ranks estimated by PL are not considerably higher than the ranks estimated using other methods. 

Several extensions of DMMD are of interest. First, generalization to more than two views described in Supplement S6 allows to apply DMMD to longitudinal data, albeit the method will treat all the time points interchangeably. An alternative approach for longitudinal data is to treat the view at starting time as base view, and apply DMMD for each consecutive view paired with the base view. This would lead to a sequence of joint and individual structures ordered by time, which could provide additional insights on the time/treatment effects. It would be of interest to pursue such analyses in the future. Second, DMMD is designed for fully observed data, thus imputation is required if some elements of $\MBX_d$ are missing. As the computations rely only on the leading singular vectors and singular values, one possibility is to adapt soft-impute algorithm \citep{mazumder2010spectral} for DMMD. Third, DMMD decomposition is not sparse, however sparse extensions can be pursued by replacing standard SVD in Algorithm~\ref{a:opt} with its sparse analogs, e.g. sparse PCA method of \citet{Shen:2008fk} or penalized orthogonal iteration of \citet{Jung2019}. This, however, would increase computational time, and create additional challenge of choosing appropriate sparsity tuning parameters. Finally, while DMMD is based on the matrix models, an alternative approach is to view double-matched data as a three-way tensor and consider tensor decompositions \citep{zhouTensor}. It is unclear, however, how to extract individual information from the latter. Furthermore, we find that existing tensor decompositions may have difficulties in capturing joint row and column structures simultaneously (Supplement S7). It would be of interest to investigate extensions of tensor decompositions that will provide more flexibility as well as preserve the matrix interpretation of joint and individual structures.

\section*{Acknowledgements}

The authors thank the Editor, AE and two anonymous reviewers for the comments that helped improve this work. The authors thank Himanshu Kumar for help in processing soccer dataset in Section~\ref{s:soccer}. The results shown in Section~\ref{s:TCGA} are based upon data generated by the TCGA Research Network: \href{https://www.cancer.gov/tcga}{https://www.cancer.gov/tcga}. 

\bigskip
\begin{center}
{\large\bf SUPPLEMENTARY MATERIAL}
\end{center}
\begin{description}
\item[Supplementary:]
Proofs of all statements, and additional numerical results (.pdf file). The \textsf{R} code is available at \url{https://github.com/justicesuker/DMMD_Code}.

\end{description}

\bibliographystyle{chicago}
\bibliography{IrinaReferences}

\renewcommand{\thesection}{S\arabic{section}}
\renewcommand{\thetable}{S\arabic{table}}
\renewcommand{\thefigure}{S\arabic{figure}}
\def\theequation{S\arabic{equation}}
\renewcommand{\thelemma}{S.\arabic{lemma}}
\renewcommand{\thealgorithm}{S.\arabic{algorithm}}

\newpage
\begin{center}
\vskip 60pt
\LARGE \bf Supplement to ``Double-matched matrix decomposition for multi-view data" \par
\vskip 2em  
\end{center}

\begin{abstract}
In section~\ref{s:proofs}, we prove all the results stated in the main paper. In section~\ref{s:lemmas}, we prove supplementary lemmas. In section~\ref{s:add_simulation}, we provide additional simulation results. In section~\ref{s:a_TCGA}, we provide additional details on TCGA data application. In section~\ref{s:DMMDi}, we describe iterative DMMD. In section~\ref{s:largeK}, we describe generalization of DMMD to more than two views. In section~\ref{s:tensor}, we illustrate difficulties with interpretation of joint structures in Tucker decomposition.
\end{abstract}

\spacingset{1.5} 

\section{Technical proofs}
\label{s:proofs}

\begin{proof}[Proof of Lemma 1]
We will prove a more general version of the lemma for the case $K \geq 2$. We only prove the decomposition with respect to column spaces, as similar proof can be applied to row spaces by transposing the matrices. This proof follows the proof of Lemma 1 in \citet{AJIVE}, but fills in more details.\\
$\bf{Existence}:$\\
Let $\mathbf{D} = \cap_{j=1}^K{\mathcal{C}(\mathbf{A}_j)} $ and choose $\mathbf{b}_1,\cdots,\mathbf{b}_r \quad (r \leq p)$ to be a basis of $\mathbf{D}$. Construct $\mathbf{J} = \left[\mathbf{b}_1,\cdots,\mathbf{b}_r \right]$, and  its projection matrix $\mathbf{P}_\mathbf{J}$. For every $k\in \{1,2,\cdots,K\}$, let $\mathbf{J}_k = \mathbf{P}_\mathbf{J}\mathbf{A}_k$ and $\mathbf{I}_k = (\mathbf{Id} - \mathbf{P}_\mathbf{J})\mathbf{A}_k$.  We next show that $\{\mathbf{J}_1,\cdots,\mathbf{J}_K\}$ and $\{\mathbf{I}_1,\cdots,\mathbf{I}_K\}$ satisfy the conditions of the lemma. By construction, $\mathcal{C}(\mathbf{J}) \subset\mathcal{C}(\mathbf{A}_k)$ and  $\mathcal{C}(\mathbf{J}) \perp \mathcal{C}(\mathbf{I}_k)$ are satisfied. Next we prove $\mathcal{C}(\mathbf{J}) \subset \mathcal{C}(\mathbf{J}_k), \quad \forall{k} \in \{1,2,\cdots,K\}$. For $\forall{\mathbf{v}} \in \mathcal{C}(\mathbf{J})$, $\exists{\mathbf{u}_1}$, such that $\mathbf{v} = \mathbf{J}\mathbf{u}_1$. Since $\mathcal{C}(\mathbf{J}) = \cap_{j=1}^K{\mathcal{C}(\mathbf{A}_j)}$, $\exists{\mathbf{u}_2}$ so that $\mathbf{v} = \mathbf{A}_k\mathbf{u}_2$. Now we have $\mathbf{v} = \mathbf{A}_k\mathbf{u}_2 = \mathbf{J}\mathbf{u}_1 \Rightarrow \mathbf{J}\mathbf{u}_1 = \mathbf{J}_k\mathbf{u}_2 + \mathbf{I}_k\mathbf{u}_2$. Because $\mathcal{C}(\mathbf{J}_k) \subset \mathcal{C}(\mathbf{J})$, we have $\mathbf{I}_k\mathbf{u}_2 \in \mathcal{C}(\mathbf{J})$, but $\mathcal{C}(\mathbf{J}) \perp \mathcal{C}(\mathbf{I}_k)$, then $\mathbf{I}_k\mathbf{u}_2 = \mathbf{0}$. As a result, $\mathbf{v} = \mathbf{A}_k\mathbf{u}_2 = \mathbf{J}_k\mathbf{u}_2 \in \mathcal{C}(\mathbf{J}_k)$, and thus $\mathcal{C}(\mathbf{J})= \mathcal{C}(\mathbf{J}_k)$. 
Finally, we prove  $\cap_{j=1}^K{\mathcal{C}(\mathbf{I}_j)} = \{\mathbf{0}\}$. Let $\mathbf{b} \in \cap_{j=1}^K{\mathcal{C}(\mathbf{I}_j)}$, then $\mathbf{b} \perp \mathcal{C}(\mathbf{J})$. At the same time, $\forall{k} \in \{1,2,\cdots,K\},\exists{\mathbf{x}_k} \in \mathcal{C}(\mathbf{J})$, ${\mathbf{y}_k} \in \mathcal{C}(\mathbf{A}_k)$ such that $\mathbf{b} = \mathbf{y}_k - \mathbf{x}_k$, which means $\mathbf{b} \in \mathcal{C}(\mathbf{A}_k)$. Then we have $\mathbf{b} \in \cap_{k=1}^K{\mathcal{C}(\mathbf{A}_k)} = \mathcal{C}(\mathbf{J})$. This means that $\mathbf{b} = \mathbf{0}$.\\
$\bf{Uniqueness}:$\\
First we prove that under the conditions of the lemma, for any $k \in \{1,2,\cdots,K\}$, $\mathcal{C}(\mathbf{J}_k) = \mathcal{C}(\mathbf{J}) = \cap_{j=1}^K{\mathcal{C}(\mathbf{A}_j)}$. Since $\mathcal{C}(\mathbf{A}_k) = \mathcal{C}(\mathbf{J}) +  \mathcal{C}(\mathbf{I}_k) $ and $\mathcal{C}(\mathbf{J}) \cap \mathcal{C}(\mathbf{I}_k) = \{\mathbf{0}\}$, we have $\mathcal{C}(\mathbf{A}_k) = \mathcal{C}(\mathbf{J}) \oplus  \mathcal{C}(\mathbf{I}_k) $. Then for any $ \mathbf{v} \in \cap_{j=1}^K{\mathcal{C}(\mathbf{A}_j)}$, there exists unique $\mathbf{u}_k \in \mathcal{C}(\mathbf{J})$ and $\mathbf{w}_k \in \mathcal{C}(\mathbf{I}_k)$ such that $\mathbf{v} = \mathbf{u}_k + \mathbf{w}_k$ for $k \in \{1,2,\cdots,K\}$. Take $k \leq K-1$, then we have $\mathbf{v} = \mathbf{u}_k + \mathbf{w}_k$ and $\mathbf{v} = \mathbf{u}_{k+1} + \mathbf{w}_{k+1}$, which means $\mathbf{u}_k - \mathbf{u}_{k+1} = \mathbf{w}_{k+1} - \mathbf{w}_k \in \mathcal{C}(\mathbf{J})$. From the perpendicularity condition, we have $\mathbf{w}_{k} = \mathbf{w}_{k+1}$. Then $\mathbf{w}_1 = \mathbf{w}_2 = \cdots = \mathbf{w}_K = \mathbf{w}$. Since $\cap_{j=1}^K{\mathcal{C}(\mathbf{I}_j)} = \{\mathbf{0}\}$, we then have $\mathbf{w} = \mathbf{0}$. Thus we conclude that $\mathbf{v} = \mathbf{u}_k \in \mathcal{C}(\mathbf{J})$, which means that $\cap_{j=1}^K{\mathcal{C}(\mathbf{A}_j)} \subset \mathcal{C}(\mathbf{J})$. Because $\mathcal{C}(\mathbf{J}) \subset\mathcal{C}(\mathbf{A}_k)$, we have $ \mathcal{C}(\mathbf{J})\subset \cap_{j=1}^K{\mathcal{C}(\mathbf{A}_j)}$. So $ \mathcal{C}(\mathbf{J}) =  \cap_{j=1}^K{\mathcal{C}(\mathbf{A}_j)}$.\\
Now suppose for any $k \in \{1,2,\cdots,K\}$ we have $\mathbf{A}_k = \mathbf{J}_k + \mathbf{I}_k$ and $\mathbf{A}_k = \mathbf{\Tilde{J}}_k + \mathbf{\Tilde{I}}_k$. For each column of matrix $\mathbf{A}_k$, say $\mathbf{a}$, we have $\mathbf{a} = \mathbf{b} + \mathbf{c} = \mathbf{b}' + \mathbf{c}'$, where $\mathbf{b},\mathbf{b}',\mathbf{c}$ and $\mathbf{c}'$ are the corresponding columns of $\mathbf{J}_k,\mathbf{\tilde{J}}_k,\mathbf{I}_k$ and $\mathbf{\Tilde{I}}_k$. Then $\mathbf{b} - \mathbf{b}' = \mathbf{c}' - \mathbf{c}$, and $(\mathbf{b} - \mathbf{b}')^T(\mathbf{b} - \mathbf{b}') = (\mathbf{b} - \mathbf{b}')^T(\mathbf{c}' - \mathbf{c}) = \mathbf{0}$, since both $\mathbf{b}$ and $\mathbf{b}'$ belong to  $\cap_{j=1}^K{\mathcal{C}(\mathbf{A}_j)} = \mathcal{C}(\mathbf{J})$, $\mathcal{C}(\mathbf{J}) \perp \mathcal{C}(\mathbf{I}_k)$ and $\mathcal{C}(\mathbf{J}) \perp \mathcal{C}(\mathbf{\Tilde{I}}_k)$. Thus we have $\mathbf{b}' = \mathbf{b}$ and $\mathbf{c}' = \mathbf{c}$, and this is true for all the columns and all $k$. We conclude that $\mathbf{J}_k = \mathbf{\Tilde{J}}_k, \quad \mathbf{I}_k = \mathbf{\Tilde{I}}_k$.
\end{proof}

\begin{proof}[Proof of Lemma 2]
First we show that $\mathbf{A}^* = \mathbf{M}\mathbf{M}^T\mathbf{X} + \mathbf{R}\mathbf{R}^T\mathbf{X}$ is feasible. Since the columns of $\mathbf{R}$ are the first $r - r_c$ left singular vectors of $(\mathbf{Id} - \mathbf{M}\mathbf{M}^T)\mathbf{X}$, $\mathbf{M}$ is orthogonal to $\mathbf{R}$. Therefore $\mathcal{C}(\mathbf{M}) \subset \mathcal{C}(\mathbf{A}^*)$ because $\mathbf{M}\mathbf{M}^T\mathbf{X}$ is of rank $r_c$. We then need to show that $\rank(\mathbf{A}^*) = r$, for which it suffices to show that $\rank(\mathbf{R}\mathbf{R}^T\mathbf{X})=r-r_c$. Consider the full SVD of $(\mathbf{Id} - \mathbf{M}\mathbf{M}^T)\mathbf{X} = \mathbf{R}\mathbf{D}_1\mathbf{V}_1^T + \mathbf{U}_2\mathbf{D}_2\mathbf{V}_2^T$, where $\mathbf{D_1}$ has the largest $r-r_c$ singular values. Then we have $\mathbf{R}\mathbf{R}^T\mathbf{X} = \mathbf{R}\mathbf{R}^T\{\mathbf{M}\mathbf{M}^T\mathbf{X} + (\mathbf{Id} - \mathbf{M}\mathbf{M}^T)\mathbf{X}\} = \mathbf{R}\mathbf{R}^T(\mathbf{Id} - \mathbf{M}\mathbf{M}^T)\mathbf{X}= \mathbf{R}\mathbf{R}^T(\mathbf{R}\mathbf{D}_1\mathbf{V}_1^T + \mathbf{U}_2\mathbf{D}_2\mathbf{V}_2^T) =\mathbf{R}\mathbf{D}_1\mathbf{V}_1^T$. Since $\rank(\mathbf{M}\mathbf{M}^T\mathbf{X}) + \rank\{(\mathbf{Id} - \mathbf{M}\mathbf{M}^T)\mathbf{X}\} \geq \rank\{\mathbf{M}\mathbf{M}^T\mathbf{X} + (\mathbf{Id} - \mathbf{M}\mathbf{M}^T)\mathbf{X}\} = \rank(\mathbf{X}) = r$, we know that $ \rank\{(\mathbf{Id} - \mathbf{M}\mathbf{M}^T)\mathbf{X}\} \geq r-r_c$. So $\mathbf{R}\mathbf{D}_1\mathbf{V}_1^T$ is of rank $r-r_c$, which means that $\rank(\mathbf{R}\mathbf{R}^T\mathbf{X})$ is indeed $r-r_c$. 

Next we show that for any feasible $\mathbf{A}_0$, $\|\mathbf{X} - \mathbf{A}^*\|_{F}^2 \leq \|\mathbf{X} - \mathbf{A}_0\|_{F}^2$. Since $\mathbf{A}_0$ is feasible, we can find $\mathbf{R}_0$ such that the column space of $\mathbf{A}_0$ is the same as the column space of an orthonormal matrix $[\mathbf{M},\mathbf{R}_0]$ with $\rank(\mathbf{R}_0) = r - r_c$. Consider the following optimization problem
\begin{align}
    &\min_{\mathbf{A}}{\|\mathbf{X} - \mathbf{A}\|_{F}^2} \label{eq:2} \\
    & \mbox{such that}\quad \mathbf{A} \in \{\mathbf{A} \in \mathbb{R}^{n \times p}|\mathcal{C}(\mathbf{A}) = \mathcal{C}([\mathbf{M},\mathbf{R}_0])\}.\nonumber
\end{align}

Notice that $\{\mathbf{A} \in \mathbb{R}^{n \times p}|\mathcal{C}(\mathbf{A}) = \mathcal{C}([\mathbf{M},\mathbf{R}_0])\}$ is a closed subspace. From classical projection theorem, \eqref{eq:2} has a unique solution $\mathbf{A}^*_0 = \mathbf{M}\mathbf{M}^T\mathbf{X} + \mathbf{R}_0\mathbf{R}_0^T\mathbf{X}$. Hence, we only need to show that $\|\mathbf{X} - \mathbf{A}^*\|_{F}^2 \leq \|\mathbf{X} - \mathbf{A}^*_0\|_{F}^2$ because $\|\mathbf{X} - \mathbf{A}_0^*\|_{F}^2 \leq \|\mathbf{X} - \mathbf{A}_0\|_{F}^2$ for any $\mathbf{A}_0$ with $\mathcal{C}(\mathbf{A}_0) = \mathcal{C}([\mathbf{M},\mathbf{R}_0])$. Furthermore, we have 
\begin{align}
    \|\mathbf{X} - \mathbf{A}^*_0\|_{F}^2 = \|(\mathbf{Id} - \mathbf{M}\mathbf{M}^T)\mathbf{X} - \mathbf{R}_0\mathbf{R}_0^T\mathbf{X}\|_{F}^2. \label{eq:3}
\end{align}

From Lemma~\ref{l:s2}, we know that \eqref{eq:3} is minimized when the columns of $\mathbf{R}_0$ are the first $r - r_c$ left singular vectors of $(\mathbf{Id} - \mathbf{M}\mathbf{M}^T)\mathbf{X}$. 
This means that $\|\mathbf{X} - \mathbf{A}^*\|_{F}^2 \leq \|\mathbf{X} - \mathbf{A}^*_0\|_{F}^2$. Thus $\mathbf{A}^* = \mathbf{M}\mathbf{M}^T\mathbf{X} + \mathbf{R}\mathbf{R}^T\mathbf{X}$ is the global solution to (3).

Furthermore, if $(r - r_c)$-th singular value does not equal to $(r - r_c + 1)$-th singular value of matrix $(\mathbf{Id} - \mathbf{M}\mathbf{M}^T)\mathbf{X}$, then $\mathbf{A}^*$ is unique global solution. Let $\mathbf{A}_1$ be another global solution. Then we can find $\mathbf{R}_1$ such that the column space of $\mathbf{A}_1$ is the same as the column space of an orthonormal matrix $[\mathbf{M},\mathbf{R}_1]$ with $\rank(\mathbf{R}_1) = r - r_c$. Consider the following optimization problem
\begin{align}
    &\min_{\mathbf{A}}{\|\mathbf{X} - \mathbf{A}\|_{F}^2} \label{eq:4} \\
    & \mbox{such that}\quad \mathbf{A} \in \{\mathbf{A} \in \mathbb{R}^{n \times p}|\mathcal{C}(\mathbf{A}) = \mathcal{C}([\mathbf{M},\mathbf{R}_1])\}.\nonumber
\end{align}

Because $\mathbf{A}_1$ minimizes (3), it also minimizes optimization problem \eqref{eq:4}. From classical projection theorem, $\mathbf{A}_1 = \mathbf{M}\mathbf{M}^T\mathbf{X} + \mathbf{R}_1\mathbf{R}_1^T\mathbf{X}$ is the unique solution to~\eqref{eq:4}. Now we have 
\begin{align}
    \|\mathbf{X} - \mathbf{A}_1\|_{F}^2 = \|(\mathbf{Id} - \mathbf{M}\mathbf{M}^T)\mathbf{X} - \mathbf{R}_1\mathbf{R}_1^T\mathbf{X}\|_{F}^2 \label{eq:5}
\end{align}

From Lemma~\ref{l:s2}, we know that $\mathbf{R}^*_1\mathbf{R}_1^{*T}\mathbf{X}$ is unique, where $\mathbf{R}^*_1$ is the solution to \eqref{eq:5}, and the objective value is minimized when $\mathbf{R}^*_1 = \mathbf{R}$, which means that $\mathbf{A}_1 = \mathbf{A}^*$.
\end{proof}

\begin{proof}[Proof of Proposition 1]
For simplicity, we ignore the subscript $k$ in the proof. From Lemma~\ref{l:s1}, the objective function at iteration $t$ of Algorithm 1 is $L^{(t)} = \|\MBX - \widetilde{\MBM}^{(t)}\widetilde{\MBM}^{(t)T}\MBX\widetilde{\MBN}^{(t)}\widetilde{\MBN}^{(t)T}\|^2_F$, and
\begin{align}
L^{(t)} &= L(\mathbf{R}^{(t)},\mathbf{S}^{(t)}) \nonumber \\
&= \|(\mathbf{M}\mathbf{M}^T + \mathbf{R}^{(t)}\mathbf{R}^{(t)T})\mathbf{X}(\mathbf{Id} - \mathbf{N}\mathbf{N}^T - \mathbf{S}^{(t)}\mathbf{S}^{(t)T})\|^2_F + \| (\mathbf{Id} - \mathbf{M}\mathbf{M}^T - \mathbf{R}^{(t)}\mathbf{R}^{(t)T})\mathbf{X} \|^2_F.
\end{align}
From Lemma~\ref{l:s3},  update $\mathbf{S}^{(t+1)}$ with fixed $\mathbf{R}^{(t)}$ corresponds to the solution of the following optimization problem:
\begin{align}
\mathbf{S}^{(t+1)} &= \argmin_{\mathbf{S}}{L(\mathbf{R^{(t)}},\mathbf{S})} \nonumber \\ 
&= \argmin_{\mathbf{S}}{\|(\mathbf{M}\mathbf{M}^T + \mathbf{R}^{(t)}\mathbf{R}^{(t)T})\mathbf{X}(\mathbf{Id} - \mathbf{N}\mathbf{N}^T) - (\mathbf{M}\mathbf{M}^T + \mathbf{R}^{(t)}\mathbf{R}^{(t)T})\mathbf{X}\mathbf{S}\mathbf{S}^T)\|^2_F}
\end{align}
On the other hand,
\begin{align}
L^{(t)} &= L(\mathbf{R}^{(t)},\mathbf{S}^{(t)}) \nonumber \\
&= \|(\mathbf{Id} - \mathbf{M}\mathbf{M}^T - \mathbf{R}^{(t)}\mathbf{R}^{(t)T})\mathbf{X}(\mathbf{N}\mathbf{N}^T + \mathbf{S}^{(t)}\mathbf{S}^{(t)T})\|^2_F + \| \mathbf{X}(\mathbf{Id} - \mathbf{N}\mathbf{N}^T - \mathbf{S}^{(t)}\mathbf{S}^{(t)T}) \|^2_F
\end{align}
From Lemma~\ref{l:s2}, the update $\mathbf{R}^{(t+1)}$ corresponds to the solution of the following optimization problem:
\begin{align}
\mathbf{R}^{(t+1)} &=  \argmin_{\mathbf{R}}{L(\mathbf{R},\mathbf{S^{(t+1)}})} \nonumber \\ 
&= \argmin_{\mathbf{R}}{\|(\mathbf{Id} - \mathbf{M}\mathbf{M}^T)\mathbf{X}(\MBN\MBN^T + \MBS^{(t+1)}\MBS^{(t+1)T}) - \mathbf{R}\mathbf{R}^T\mathbf{X}(\MBN\MBN^T + \MBS^{(t+1)}\MBS^{(t+1)T})\|^2_F}   
\end{align}
Therefore, we have
$$
L^{(t)} = L(\mathbf{R}^{(t)},\mathbf{S}^{(t)}) \geq L(\mathbf{R}^{(t)},\mathbf{S}^{(t+1)}) \geq L(\mathbf{R}^{(t+1)},\mathbf{S}^{(t+1)}) = L^{(t+1)}.
$$
Since the objective function value is awlays bounded by zero, this means that the sequence $L^{(t)}$ is guaranteed to converge.
\end{proof}

\section{Additional lemmas}
\label{s:lemmas}
\begin{lemma}\label{l:s1}
Given $\MBX\in \R^{n\times p}$, $\MBM\in \R^{n\times r}$ and $\MBN\in \R^{p\times r}$ with $r$ orthonormal columns and $r$ with $0 \leq r \leq \min(n, p$). Consider
\begin{equation}\label{eq:simpler-add}
    \min_{\MBA \in \R^{n\times p}}{\|\MBX - \MBA\|_{F}^2}
     \quad \mbox{such that}\quad \mathcal{C}(\MBA) = \mathcal{C}(\MBM), \quad
     \mathcal{R}(\MBA) = \mathcal{C}(\MBN), \quad \rank(\MBA) = r. 
\end{equation}
Let $\MBA^* = \mathbf{M}\mathbf{M}^T\mathbf{X}\mathbf{N}\mathbf{N}^T$. If $\rank(\mathbf{M}\mathbf{M}^T\mathbf{X}\mathbf{N}\mathbf{N}^T) = r$, then $\MBA^{*}$ is the unique global minimizer.
\end{lemma}

\begin{proof}[Proof of Lemma~\ref{l:s1}]
For every feasible $\mathbf{A}$, $\mathbf{M}\mathbf{M}^T\mathbf{A}\mathbf{N}\mathbf{N}^T = \mathbf{A}$ holds. Thus the objective can be written as 
\begin{align}
\|\mathbf{X} - \mathbf{A}\|_{F}^2 &= \|\mathbf{X} - \mathbf{M}\mathbf{M}^T\mathbf{A}\mathbf{N}\mathbf{N}^T\|_{F}^2 \nonumber\\
&= \|\mathbf{M}\mathbf{M}^T(\mathbf{X} - \mathbf{A})\mathbf{N}\mathbf{N}^T + 
(\mathbf{Id} - \mathbf{M}\mathbf{M}^T)\mathbf{X}\mathbf{N}\mathbf{N}^T + \nonumber \\
& \quad \mathbf{M}\mathbf{M}^T\mathbf{X}(\mathbf{Id} - \mathbf{N}\mathbf{N}^T) + 
(\mathbf{Id} - \mathbf{M}\mathbf{M}^T)\mathbf{X}(\mathbf{Id} - \mathbf{N}\mathbf{N}^T)\|_{F}^2. \nonumber\\
&= \|\mathbf{M}\mathbf{M}^T(\mathbf{X} - \mathbf{A})\mathbf{N}\mathbf{N}^T\|^2_F + \text{constant} \nonumber \\
&= \|\mathbf{M}\mathbf{M}^T\mathbf{X}\mathbf{N}\mathbf{N}^T - \mathbf{A}\|^2_F + \text{constant}. \nonumber
\end{align}
The above Frobenius norm is minimized when $\mathbf{A} = \mathbf{M}\mathbf{M}^T\mathbf{X}\mathbf{N}\mathbf{N}^T$. Because we assume that $\rank(\mathbf{M}\mathbf{M}^T\mathbf{X}\mathbf{N}\mathbf{N}^T) = r$, the solution is feasible and unique.
\end{proof}

\begin{lemma}\label{l:s2}
Given $\MBW \in \R^{n\times p}$ and $\MBM\in \R^{n\times r_c}$ with $r_c$ orthonormal columns with $0 \leq r_c \leq r \leq \min(n, p)$,  consider
\begin{equation}\label{eq:simpler2-add}
\min_{\MBR \in \R^{n\times (r-r_c)}}{\|(\mathbf{Id} - \mathbf{M}\mathbf{M}^T)\mathbf{W} - \mathbf{R}\mathbf{R}^T\mathbf{W}\|^2_F}
     \quad \mbox{such that}\quad \MBR^T\MBR = \MBId.
\end{equation}
If $(\mathbf{Id} - \mathbf{M}\mathbf{M}^T)\mathbf{W}$ is of rank at least $r-r_c$, then one optimal $\MBR^*$ is the first $r - r_c$ columns of left singular vectors of $(\mathbf{Id} - \mathbf{M}\mathbf{M}^T)\mathbf{W}$. Furthermore, if $(r - r_c)$-th singular value does not equal to $(r - r_c + 1)$-th singular value of matrix $(\mathbf{Id} - \mathbf{M}\mathbf{M}^T)\mathbf{W}$, then $\mathbf{R}^*\mathbf{R}^{*T}\mathbf{W}$ is unique.
\end{lemma}

\begin{proof}[Proof of Lemma~\ref{l:s2}]
Consider change of variables: $\MBA = \mathbf{R}\mathbf{R}^T\mathbf{W} \in \mathbb{R}^{n \times p}$ with $\rank(\MBA) \leq \rank(\MBR) = r-r_c$. By Eckart-Young-Mirsky theorem, the Frobenius norm in~\eqref{eq:simpler2-add} is minimized when $\MBA$ is the rank-$(r-r_c)$ SVD approximation of $(\mathbf{Id} - \mathbf{M}\mathbf{M}^T)\mathbf{W}$, that is $\mathbf{A}^* = \mathbf{U}_1\mathbf{D}_1\mathbf{V}_1^T$, where $\MBU_1 \in \mathbb{R}^{n \times (r-r_c)}$ are left singular vectors corresponding to $r-r_c$ largest singular values so that the full SVD is $(\mathbf{Id} - \mathbf{M}\mathbf{M}^T)\mathbf{W} = \mathbf{U}_1\mathbf{D}_1\mathbf{V}_1^T + \mathbf{U}_2\mathbf{D}_2\mathbf{V}_2^T$. 

Next we show that choosing feasible $\MBR^* = \MBU_1$ leads to $\mathbf{R}^*\mathbf{R}^{*T}\mathbf{W} = \mathbf{U}_1\mathbf{D}_1\mathbf{V}_1^T$. Since $\mathbf{R}^{*T}\mathbf{M} = \mathbf{0}$ as $\MBU_1$ are singular vectors of $(\mathbf{Id} - \mathbf{M}\mathbf{M}^T)\mathbf{W}$, it follows that $\mathbf{R}^*\mathbf{R}^{*T}\mathbf{W} = \mathbf{R}^*\mathbf{R}^{*T}\{\mathbf{M}\mathbf{M}^T\mathbf{W} + (\mathbf{Id} - \mathbf{M}\mathbf{M}^T)\mathbf{W}\} = \mathbf{R}^*\mathbf{R}^{*T}(\mathbf{Id} - \mathbf{M}\mathbf{M}^T)\mathbf{W}= \mathbf{R}^*\mathbf{R}^{*T}(\mathbf{U}\mathbf{D}_1\mathbf{V}_1^T + \mathbf{U}_2\mathbf{D}_2\mathbf{V}_2^T) =\mathbf{U}_1\mathbf{D}_1\mathbf{V}_1^T$. Thus the Frobenius norm objective value reaches its minimum when $\MBR^* = \MBU_1$. Furthermore, if $(r - r_c)$-th singular value does not equal to $(r - r_c + 1)$-th singular value of matrix $(\mathbf{Id} - \mathbf{M}\mathbf{M}^T)\mathbf{W}$, then by Eckart-Young-Mirsky theorem, the solution $\mathbf{A}^* = \mathbf{U}_1\mathbf{D}_1\mathbf{V}_1^T$ is unique. So $\mathbf{A}^* = \mathbf{R}^*\mathbf{R}^{*T}\mathbf{W}$ is unique.
\end{proof}


\begin{lemma}\label{l:s3}
Given $\MBW \in \R^{n\times p}$ and $\MBN\in \R^{p\times r_r}$ with $r_r$ orthonormal columns with $0 \leq r_r \leq r \leq \min(n, p)$,  consider
\begin{equation}\label{eq:simpler3-add}
\min_{\MBS \in \R^{p\times (r-r_r)}}{\|\mathbf{W}(\mathbf{Id} - \mathbf{N}\mathbf{N}^T) - \mathbf{W}\mathbf{S}\mathbf{S}^T\|^2_F}
     \quad \mbox{such that}\quad \MBS^T\MBS = \MBId.
\end{equation}
If $\mathbf{W}(\mathbf{Id} - \mathbf{N}\mathbf{N}^T)$ is of rank at least $r-r_r$, then one optimal $\MBS^*$ is the first $r - r_r$ columns of right singular vectors of $\mathbf{W}(\mathbf{Id} - \mathbf{N}\mathbf{N}^T)$. Furthermore, if $(r - r_r)$-th singular value does not equal to $(r - r_r + 1)$-th singular value of matrix $\mathbf{W}(\mathbf{Id} - \mathbf{N}\mathbf{N}^T)$, then $\mathbf{W}\mathbf{S}^*\mathbf{S}^{*T}$ is unique.
\end{lemma}
\begin{proof}[Proof of Lemma~\ref{l:s3}]
The proof is analogous to the proof of Lemma~\ref{l:s2}.
\end{proof}

\begin{lemma}\label{l:s4}
Given $\MBX_k \in \R^{n\times p}$, $k=1, 2$, and $\MBR_k \in \R^{n\times (r_k - r_c)}$, $\widetilde \MBN_k = [\MBN, \MBS_k] \in \R^{p \times r_k}$ with orthonormal columns, consider
\begin{align}
    \minimize_{\MBA_k \in \R^{n\times p}}&{\sum_{k=1}^2\|\MBX_k - \MBA_k\|_{F}^2} \label{eq:opt_joint}\\
    \mbox{such that}&\quad \mathcal{R}(\MBA_k) = \mathcal{C}(\widetilde \MBN_k),\nonumber\\
    &\quad \mathcal{C}(\MBA_k) = \mathcal{C}(\widetilde \MBM_k), \quad \widetilde \MBM_k = [\MBM, \MBR_k], \quad \widetilde{\MBM}^\top_k\widetilde{\MBM}_k = \MBId,  \quad k = 1,2.\nonumber
\end{align}
Let
$
\MBY = [(\mathbf{Id} - \MBR_1\MBR_1^\top)\MBX_1\widetilde{\MBN}_1\widetilde{\MBN}^\top_1,(\mathbf{Id} - \MBR_2\MBR_2^\top)\MBX_2\widetilde{\MBN}_2\widetilde{\MBN}^\top_2],
$
and let $\MBM^*$ be the matrix containing first $r_c$ left singular vectors of $\MBY$. Let $\widetilde{\MBM}^*_k = [\MBM^*,\MBR_k]$. Then if $\MBA^*_k = \widetilde{\MBM}^*_k\widetilde{\MBM}^{*\top}_k\MBX_k\widetilde{\MBN}_k\widetilde{\MBN}^\top_k$ has rank $r_k$, it is the global minimizer of \eqref{eq:opt_joint}.
\end{lemma}

\begin{proof}[Proof of Lemma~\ref{l:s4}]
According to Lemma~\ref{l:s1}, for optimal $\MBM$, the solution is
$$
\MBA_k = \widetilde{\MBM}_k\widetilde{\MBM}^\top_k\MBX_k\widetilde{\MBN}_k\widetilde{\MBN}^\top_k.
$$
Thus the minimization in \eqref{eq:opt_joint} can be rewritten as
\begin{align}
\label{eq:opt_joint_sim}
    \minimize_{\MBM \in \R^{n\times r_c}}&\{\|\MBX_1 - \widetilde{\MBM}_1\widetilde{\MBM}^\top_1\MBX_1\widetilde{\MBN}_1\widetilde{\MBN}^\top_1\|_{F}^2 + \|\MBX_2 - \widetilde{\MBM}_2\widetilde{\MBM}^\top_2\MBX_2\widetilde{\MBN}_2\widetilde{\MBN}^\top_2\|_{F}^2\}\\
    \mbox{such that} &\quad \quad \widetilde \MBM_k = [\MBM, \MBR_k], \quad \widetilde{\MBM}^\top_k\widetilde{\MBM}_k = \MBId,  \quad k = 1,2.\nonumber
\end{align}
Observe that
\begin{align}
\|\MBX_1 - \widetilde{\MBM}_1\widetilde{\MBM}^\top_1\MBX_1\widetilde{\MBN}_1\widetilde{\MBN}^\top_1\|_{F}^2 
&= \|\MBX_1\widetilde{\MBN}_1\widetilde{\MBN}^\top_1 - \widetilde{\MBM}_1\widetilde{\MBM}^\top_1\MBX_1\widetilde{\MBN}_1\widetilde{\MBN}^\top_1\|_{F}^2 + \text{constant} \nonumber\\
&= \|(\mathbf{Id} - \widetilde{\MBM}_1\widetilde{\MBM}^\top_1)\MBX_1\widetilde{\MBN}_1\widetilde{\MBN}^\top_1\|_{F}^2 + \text{constant} \nonumber \\
&= \|(\mathbf{Id} - \MBM\MBM^\top - \MBR_1\MBR_1^\top)\MBX_1\widetilde{\MBN}_1\widetilde{\MBN}^\top_1\|_{F}^2 + \text{constant}. \nonumber 
\end{align}
This means we only need to minimize
$$
\|(\mathbf{Id} - \MBM\MBM^\top - \MBR_1\MBR_1^\top)\MBX_1\widetilde{\MBN}_1\widetilde{\MBN}^\top_1\|_{F}^2 + \|(\mathbf{Id} - \MBM\MBM^\top - \MBR_2\MBR_2^\top)\MBX_2\widetilde{\MBN}_2\widetilde{\MBN}^\top_2\|_{F}^2.
$$
Due to orthogonality of $\MBR_k$ and $\MBM$, this is equivalent to
$$
\sum_{k=1}^2\|(\mathbf{Id} - \MBR_k\MBR_k^\top)\MBX_k\widetilde{\MBN}_k\widetilde{\MBN}^\top_k - \MBM\MBM^\top(\mathbf{Id} - \MBR_k\MBR_k^\top)\MBX_k\widetilde{\MBN}_k\widetilde{\MBN}^\top_k\|_{F}^2.
$$
Let
$
\MBY = [(\mathbf{Id} - \MBR_1\MBR_1^\top)\MBX_1\widetilde{\MBN}_1\widetilde{\MBN}^\top_1,(\mathbf{Id} - \MBR_2\MBR_2^\top)\MBX_2\widetilde{\MBN}_2\widetilde{\MBN}^\top_2],
$
then the minimization can be equivalently written as
$$
\minimize_{\MBM \in \R^{n\times r_c}: \MBM^{\top}\MBM = \MBI}{\|\MBY - \MBM\MBM^\top\MBY\|_{F}^2}.
$$
By Eckart-Young-Mirsky theorem, the optimal $\MBM$ is the matrix of the first $r_c$ left singular vectors of $\MBY$.
\end{proof}

\section{Additional simulation results}
\label{s:add_simulation}

\subsection{Data generation details}
Given the sample size $n$, the number of features $p$, the total signal ranks $r_k \leq \min(n, p)$, $k=1,2$, the rank of joint column structure $r_c\leq \min(r_1, r_2)$ and the rank of joint row structure $r_r\leq \min(r_1, r_2)$, we generate the signal matrix $\MBA_k \in \mathbb{R}^{n \times p}$ according to  
$$
\MBA_k = (\MBF_k\MBQ_{1k})\MBD_k(\MBQ_{2k}\MBG_k)^T,\ k = 1,2;
$$
where
\begin{itemize}
    \item $\MBF_k \in \mathbb{R}^{n \times r_k}$ captures the column-space of $\MBA_k$ with columns being the standard bases in $\mathbb{R}^n$. We generate the bases of joint column space of $\MBA_1$ and $\MBA_2$ as the first $r_c$ columns of $\MBF_1$ and $\MBF_2$. These columns have $1$s in the positions sampled from $1,2,\cdots,\frac{n}{2}$ (without replacement). The bases of individual column space of $\MBA_1$ are the remaining $r_1 - r_c$ columns of $\MBF_1$ with positions of $1$s sampled from $\frac{n}{2}+1,\cdots,\frac{3n}{4}$. The individual column space of $\MBA_2$ is the span of the rest of the columns in $\MBF_2$ with positions of $1$s sampled from $\frac{3n}{4}+1,\cdots,n$. Thus, the joint column space is orthogonal to the individual space and the individual spaces have zero intersection. Figure~\ref{fig:datagen} shows an example of $\MBF_k$ with $n=8$, $r_1 = r_2 = 2$, $r_c = 1$. 
    \item $\MBG_k \in \mathbb{R}^{p \times r_k}$ captures the row space of $\MBA_k$ and is generated similarly to $\MBF_k$. Figure~\ref{fig:datagen} shows an example of $\MBG_k$ with $p=4$, $r_1 = r_2 = 2$, $r_r = 1$.
    \item $\MBQ_{1k} \in \mathbb{R}^{r_k \times r_k}$ is an orthogonal matrix. We first generate $\MBH_k\in \R^{r_k \times r_k}$ with independent entries from standard Gaussian distribution, and then set $\MBQ_{1k}$ = $\MBU_k$ from the SVD: $\MBH_k = \MBU_k\mathbf{\Sigma}_k\MBV_k^T$. $\MBQ_{2k}$ is generated similarly.
    \item $\MBD_k \in \R^{r_k \times r_k}$ is a diagonal matrix of singular values which are drawn independently from a uniform distribution on $[0.5, 1.5]$. To control the Frobenius norm of the signal matrix, we scale the singular values so that $\sum_{i=1}^{r_k}d_{kii}^2 = r_k$.
\end{itemize}

\begin{figure}[!t]
\centering
\[
\MBF_1 =
\setstretch{1.2}
 \bordermatrix{\text{} & 1 & 2\cr
    1&\tikzmark{left1} 0                  & \tikzmark{left2}0\cr
    2&                 0                  &                 0\cr
    3&                 1                  &                 0\cr
    4&                 0                  &                 0\cr
    5&                 0                  &                 0\cr
    6&                 0                  &                 1\cr
    7&                 0                  &                 0\cr
    8&                 0\tikzmark{right1} &                 0\tikzmark{right2}
}, 
\MBF_2 =
\setstretch{1.2}
 \bordermatrix{\text{} & 1 & 2\cr
    1&\tikzmark{left3} 0                  & \tikzmark{left4} 0\cr
    2&                 0                  &                  0\cr
    3&                 1                  &                  0\cr
    4&                 0                  &                  0\cr
    5&                 0                  &                  0\cr
    6&                 0                  &                  0\cr
    7&                 0                  &                  1\cr
    8&                 0 \tikzmark{right3}&                  0\tikzmark{right4}
},
\MBG_1 =
\setstretch{1.2}
 \bordermatrix{\text{} & 1 & 2\cr
    1&\tikzmark{left5} 1                   & \tikzmark{left6} 0\cr
    2&                 0                   &                  0\cr
    3&                 0                   &                  1\cr
    4&                 0 \tikzmark{right5} &                  0\tikzmark{right6}
}, 
\MBG_2 =
\setstretch{1.2}
 \bordermatrix{\text{} & 1 & 2\cr
    1&\tikzmark{left7} 1                   & \tikzmark{left8} 0\cr
    2&                 0                   &                  0\cr
    3&                 0                   &                  0\cr
    4&                 0 \tikzmark{right7} &                  1\tikzmark{right8}
}
\]

\DrawBox[thick, red, dotted ]{left1}{right1}{\textcolor{red}{\tiny$Joint$}}
\DrawBox[thick, blue]{left2}{right2}{\textcolor{blue}{\tiny$Ind1$}}
\DrawBox[thick, red, dotted ]{left3}{right3}{\textcolor{red}{\tiny$Joint$}}
\DrawBox[thick, green]{left4}{right4}{\textcolor{green}{\tiny$Ind2$}}

\DrawBox[thick, black, dotted ]{left5}{right5}{\textcolor{black}{\tiny$Joint$}}
\DrawBox[thick, orange]{left6}{right6}{\textcolor{orange}{\tiny$Ind1$}}
\DrawBox[thick, black, dotted ]{left7}{right7}{\textcolor{black}{\tiny$Joint$}}
\DrawBox[thick, purple]{left8}{right8}{\textcolor{purple}{\tiny$Ind2$}}
\caption{An example for $\MBF_1$, $\MBF_2$, $\MBG_1$, $\MBG_2$ when $n = 8, p = 4, r_1 = r_2 = 2, r_c = r_r = 1$}
\label{fig:datagen}
\end{figure}

\subsection{Rank estimation}
\label{s:Rankest-add}

\begin{figure}[!t]
\begin{subfigure}{0.5\textwidth}
\includegraphics[width=1\textwidth]{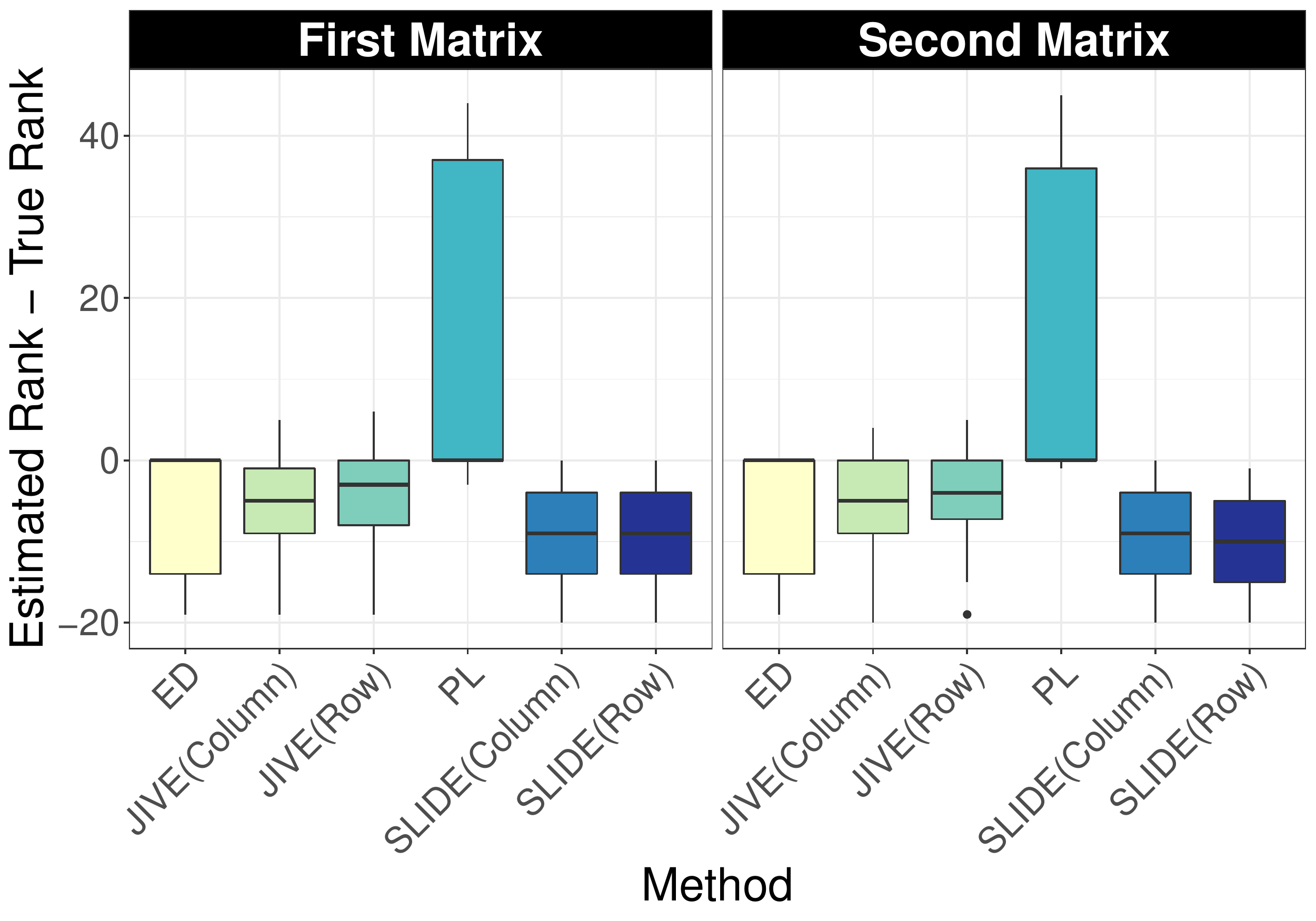}
\caption{Total Rank Estimation}
\label{fig:set2_total} 
\end{subfigure}%
\begin{subfigure}{0.5\textwidth}
\includegraphics[width=1\textwidth]{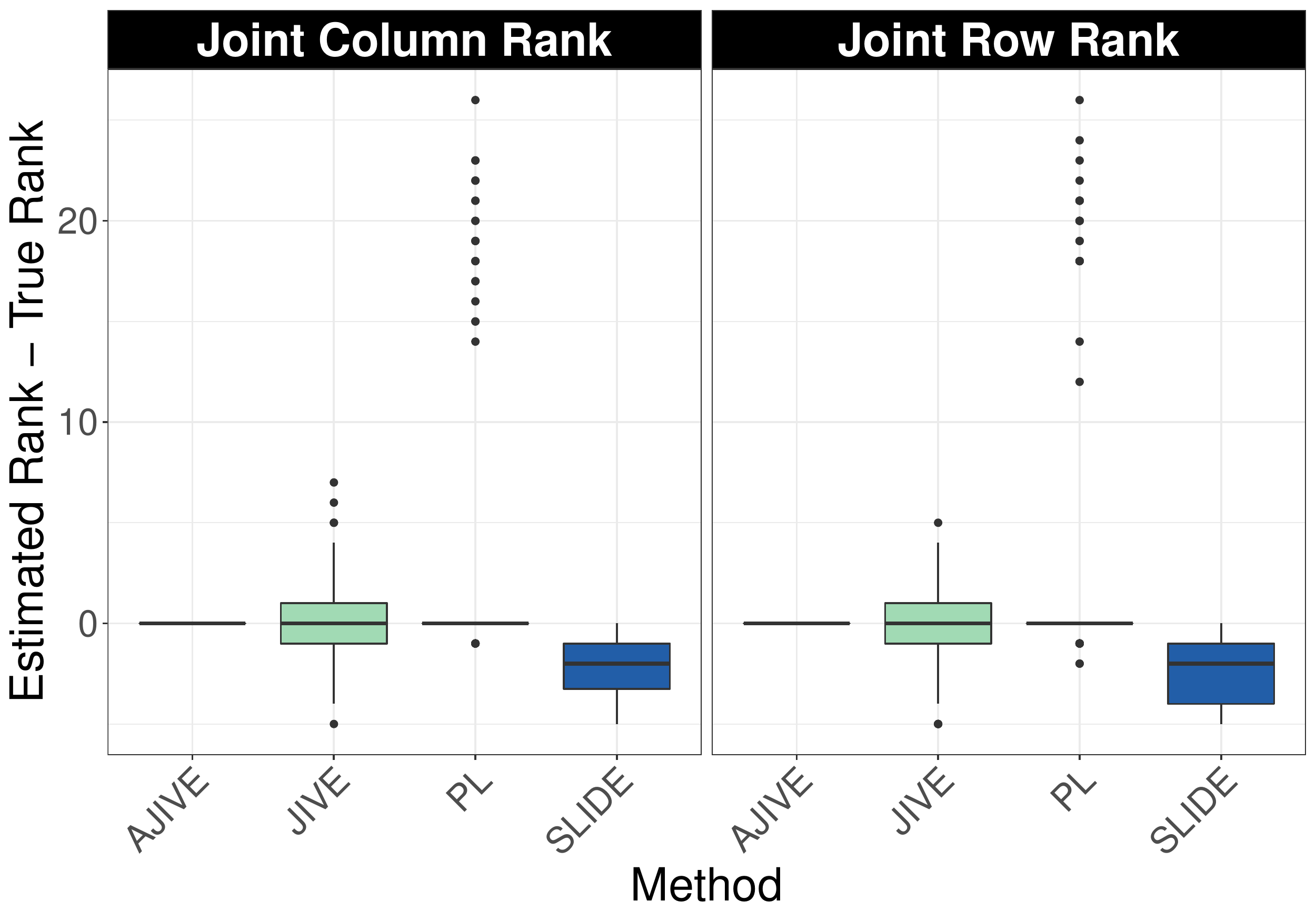}
\caption{Joint Rank Estimation}
\label{fig:set2_joint} 
\end{subfigure}
\caption{Comparison of Rank Estimation for Setting 2}
\end{figure}

Figure~\ref{fig:set2_total} shows the comparison of total rank estimation accuracy between the methods in Setting 2, the corresponding summary statistics are in Table~\ref{tab:total1_setting2}. Compared to Setting~1, the only difference is the lower SNR (changed from 1 to 0.5). JIVE and SLIDE tend to underestimate the ranks and also are not consistent (different ranks are estimated depending on whether the matching is done by rows or by columns). In most of the replications, PL and ED methods work better than JIVE and SLIDE, however they have higher variance. PL tends to overestimate the total rank, while ED tends to underestimate the total rank. The results of estimating joint ranks are shown in Figure~\ref{fig:set2_joint}, the corresponding summary statistics are in Table~\ref{tab:jointc_setting2}. The Wedin bound method used by AJIVE works perfectly, however it relies on the knowledge of true total ranks. PL method works better than JIVE and SLIDE in most settings, however occasionally it overestimates the joint rank. 

\begin{table}[!t]
\small
\caption{Total rank estimation errors, $\widehat r_k - r_k$, in Setting 2 across 140 replications.} \label{tab:total1_setting2}
\begin{center}
\begin{tabular}{l|lrrrrrr}
 & Metric & PL & ED & JIVE(row) & JIVE(col) & SLIDE(row) & SLIDE(col) \\\hline
1st matrix&Min & -3 & -19 & -19 & -19 & -20 & -20 \\
&1st quartile & 0 & -14 & -8 & -9 & -14 & -14 \\
&Median & 0 & 0 & -3 & -5 & -9 & -9\\
&Mean & 12.5 & -5.4 & -4.3 & -5.3 & -9.3 & -9.2\\
&3rd quartile & 37 & 0 & 0 & -1 & -4 & -4\\
&Max & 44 & 0 & 6 & 5 & 0 & 0\\
\hline
2nd matrix &Min & -1 & -19 & -19 & -20 & -20 & -20\\
&1st quartile & 0 & -14 & -7.3 & -9 & -15 & -15\\
&Median & 0 & 0 & -4 & -5 & -10 & -10\\
&Mean & 13.7 & -5.3 & -4.3 & -5.4 & -10.2 & -10.2\\
&3rd quartile & 36 & 0 & 0 & 0 & -5 & -5\\
&Max & 45 & 0 & 5 & 4 & -1 & 0\\
\end{tabular}
\normalsize
\end{center}
\end{table}

\begin{table}[!t]
\small
\caption{Joint rank estimation errors in Setting 2 across 140 replications.\label{tab:jointc_setting2}}
\begin{center}
\begin{tabular}{l|rrrrr}
&Metric & PL & JIVE & AJIVE & SLIDE\\\hline
$\widehat r_c - r_c$& Min & -1 & -5 & 0 & -5 \\
&1st quartile & 0 & -1 & 0 & -3.3 \\
&Median & 0 & 0 & 0 & -2 \\
&Mean & 2.8 & 0.3 & 0 & -2.5 \\
&3rd quartile & 0 & 1 & 0 & -1 \\
&Max & 26 & 7 & 0 & 0 \\
\hline
$\widehat r_r - r_r$& Min & -2 & -5 & 0 & -5 \\
&1st quartile & 0 & -1 & 0 & -4 \\
&Median & 0 & 0 & 0 & -2 \\
&Mean & 2.8 & -0.2 & 0 & -2.5 \\
&3rd quartile & 0 & 1 & 0 & -1 \\
&Max & 26 & 5 & 0 & 0 \\
\end{tabular}
\normalsize
\end{center}
\end{table}

\subsection{Signal identification}
\label{s:Signalid-add}
Figures~\ref{fig:set5_col} and~\ref{fig:set5_row} show relative errors of all methods in Setting 5 for estimated signals based on matched rows ($\MBJ_{ck},\MBI_{ck}$) and matched columns ($\MBJ_{rk},\MBI_{rk}$), respectively. The errors for total signal are the same for DMMD. In contrast, the errors for JIVE, AJIVE and SLIDE depend on matching (by rows or by columns) as it affects the estimated signal. For joint signals, DMMD and SLIDE perform similar, and are both more accurate than JIVE and AJIVE. DMMD has the smallest errors on full signals and individual signals in all scenarios, similar to the simulation results in Setting 1, confirming that taking into account double matching leads to more accurate signal estimation. 

\begin{figure}[!t]
\begin{subfigure}{0.5\textwidth}
\includegraphics[width=1\textwidth]{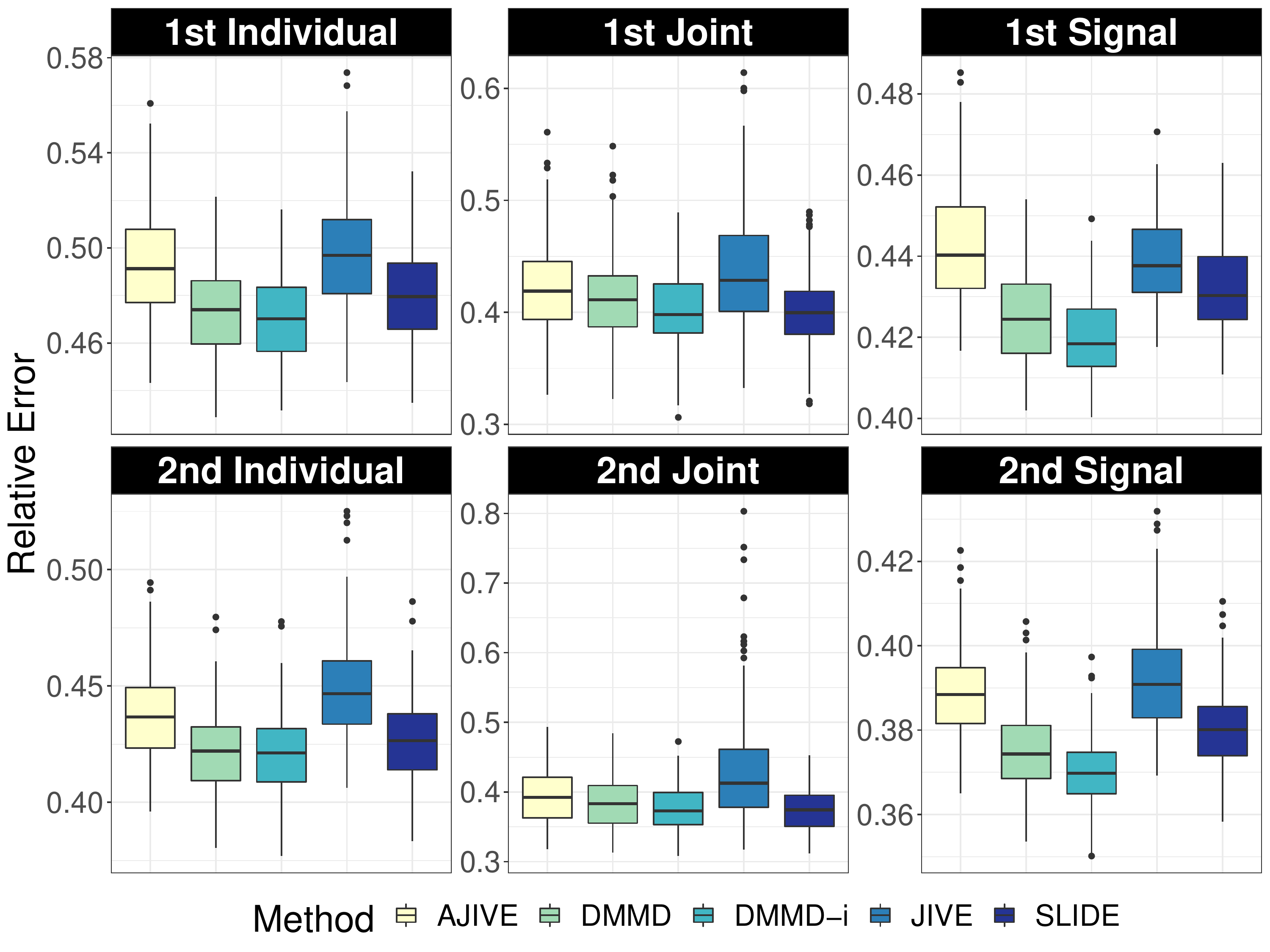}
\caption{Column space decomposition (matched rows)}
\label{fig:set5_col}
\end{subfigure}%
~
\begin{subfigure}{0.5\textwidth}
\includegraphics[width=1\textwidth]{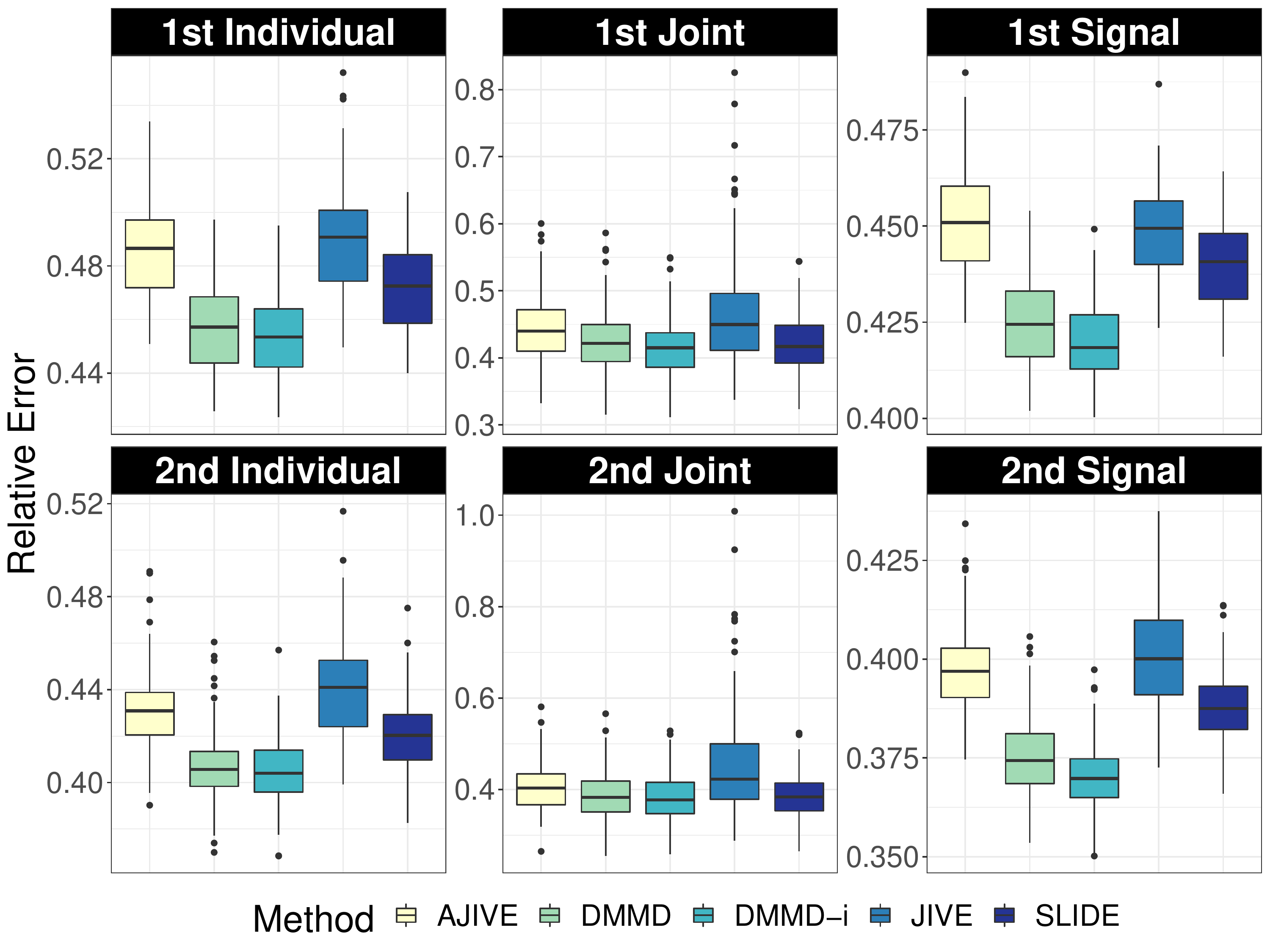}
\caption{Row space decomposition (matched columns)}
\label{fig:set5_row}
\end{subfigure}
\caption{Comparison of signal identification for Setting 5 over 140 replications, $n = 240,\ p = 200$, $r_1 = 20,\ r_2 = 18,\ r_c = 4,\ r_r = 3$, $\mbox{SNR} = 0.5$. }
\end{figure}

\subsection{High-dimensional simulation setting}
We create a high-dimensional simulation setting that mimics the TCGA data set in Section 4.1. We use $n = 88, p = 736, r_1 = 8, r_2 = 6, r_c = 0, r_r = 2, \text{SNR} = 1$ with 140 replications. We compare both performance in ranks estimation and in signal identification. 

The ranks estimation results are shown in Figures~\ref{fig:tcga_total}-\ref{fig:tcga_joint}. The results are similar to other settings. Edge distribution method works best on total rank estimation, however, it requires strong assumptions on the size of the underlying true ranks. These assumptions are satisfied in this simulation, but may not be satisfied in real data. Profile likelihood is second best, but sometimes slightly underestimating the total rank. JIVE works better in column direction rather than the row direction, and SLIDE has consistent bias due to underestimation. For joint rank estimation, AJIVE works the best, however it uses true total ranks, which are in generally uknown. The profile likelihood is slightly worse than AJIVE for joint row rank, but is better than all other methods.

The signal identification results are shown in Figures~\ref{fig:tcga_col}--\ref{fig:tcga_row}. All methods have similar performance on joint structure, with JIVE being worse on joint row structure compared to other methods. The proposed performs best on individual structure, and on total signal, with iterative DMMD being slightly better than DMMD. 

\begin{figure}[!t]
\centering
\begin{subfigure}{0.5\textwidth}
\includegraphics[width=1\linewidth]{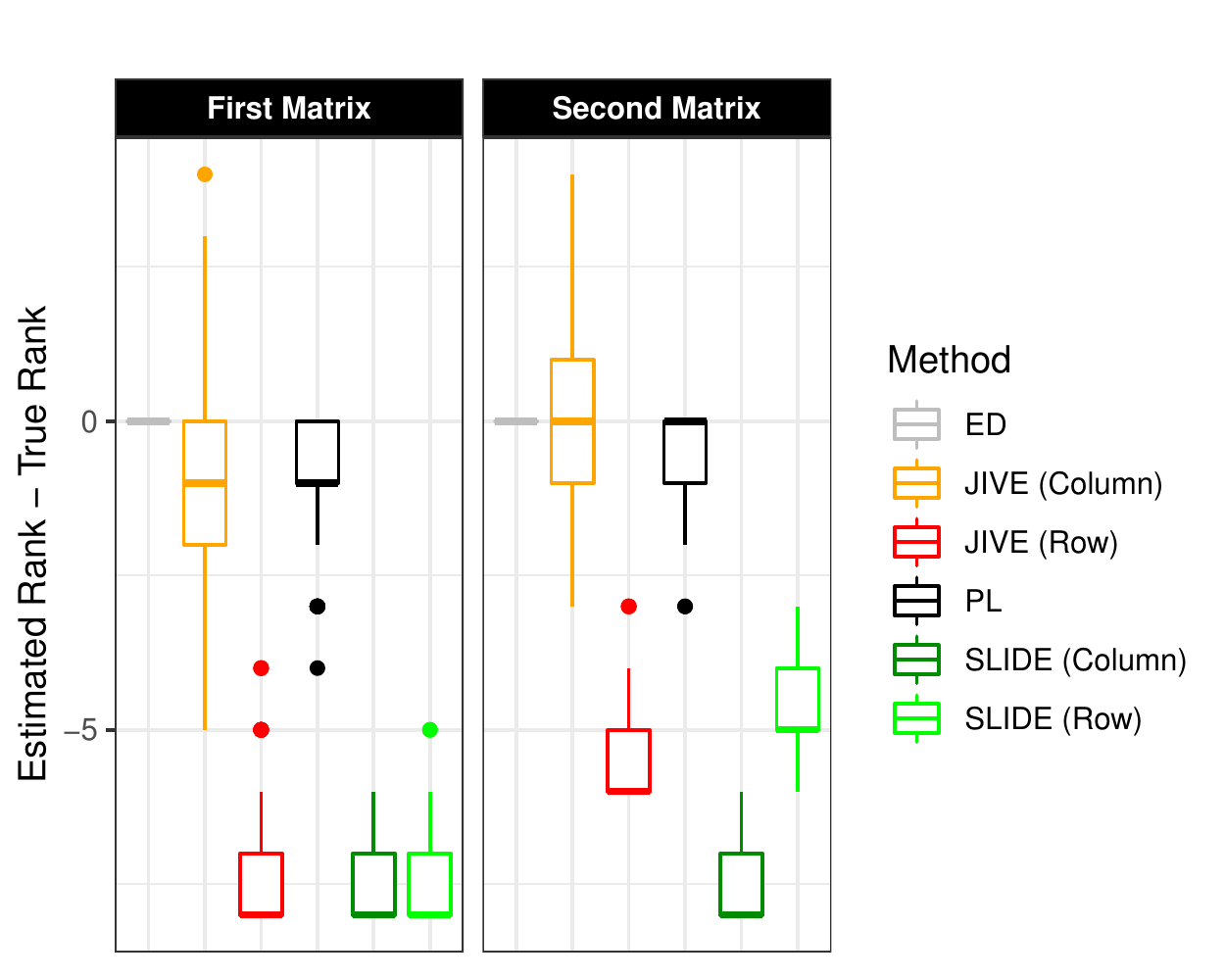}
\caption{Total rank estimation}
\label{fig:tcga_total}
\end{subfigure}%
\begin{subfigure}{0.5\textwidth}
\includegraphics[width=1\linewidth]{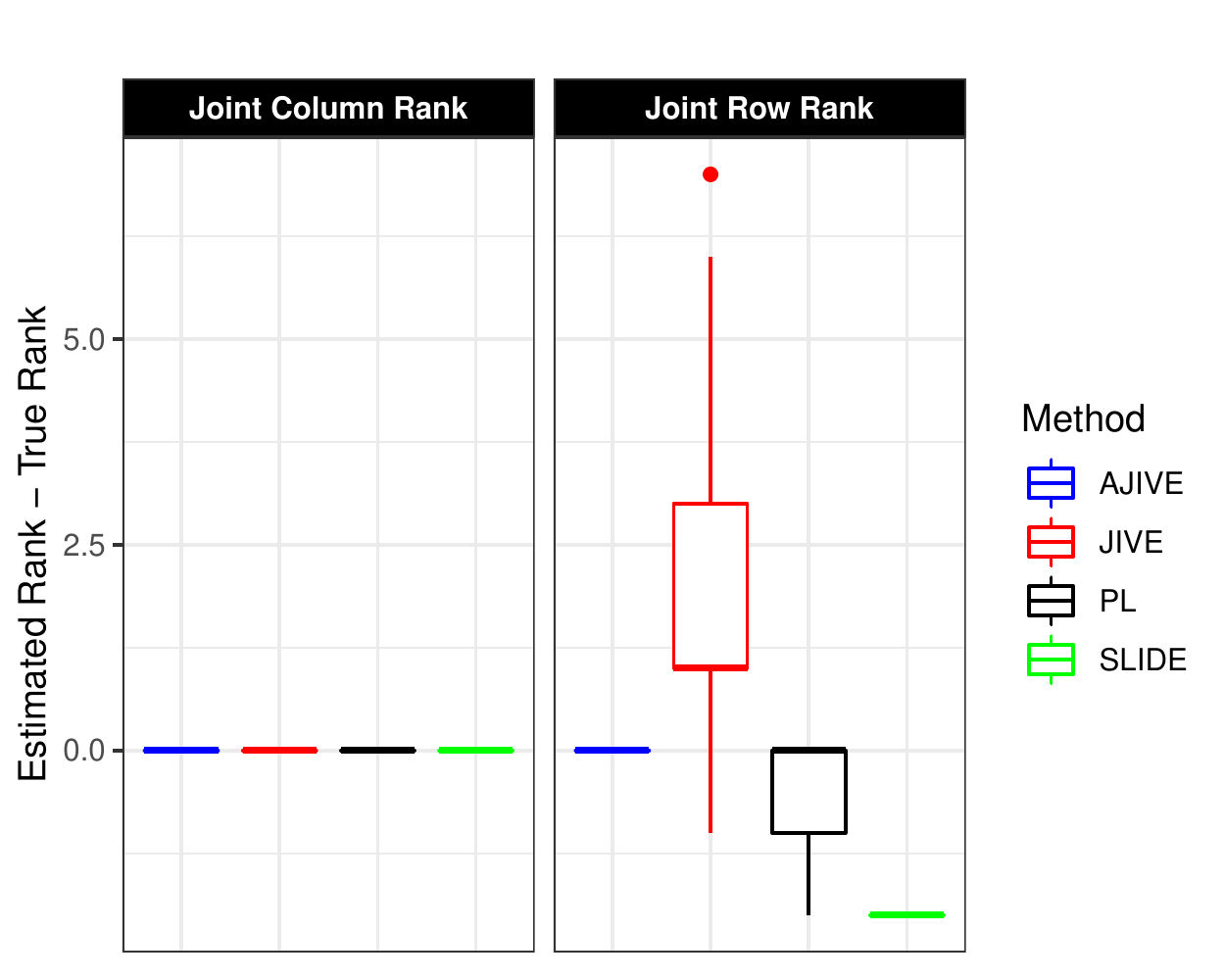}
\caption{Joint rank estimation}
\label{fig:tcga_joint}
\end{subfigure}
\caption{Comparison of rank estimation over 100 replications when $n = 88, p = 736, r_1 = 8, r_2 = 6, r_c = 0, r_r = 2$. JIVE (Column) or SLIDE (Column) estimates the total rank when columns are matched and vice versa.}
\end{figure}

\begin{figure}[!t]
\begin{subfigure}{0.5\textwidth}
\includegraphics[width=1\textwidth]{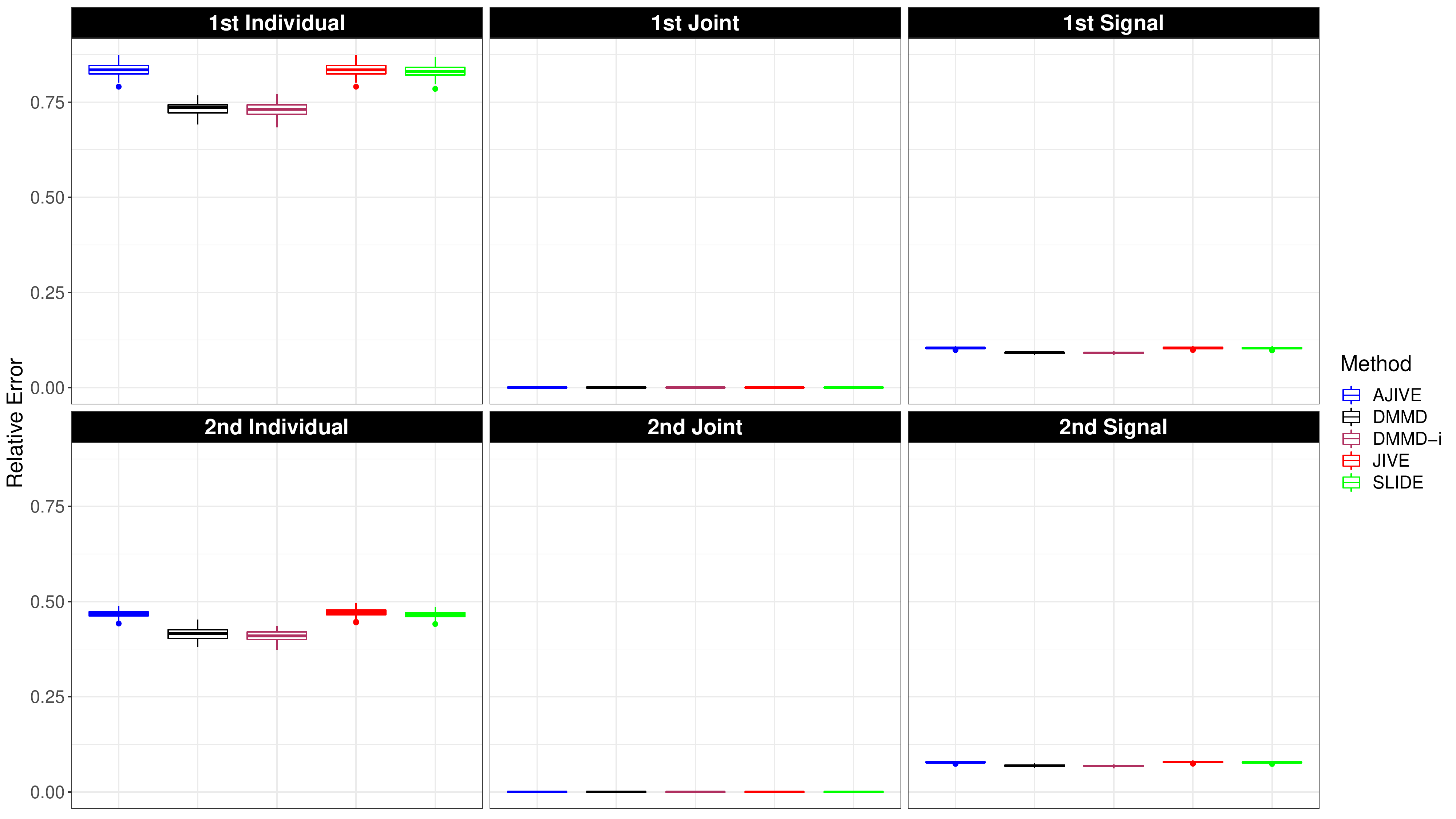}
\caption{Column space decomposition (matched rows)}
\label{fig:tcga_col}
\end{subfigure}%
\begin{subfigure}{0.5\textwidth}
\includegraphics[width=1\textwidth]{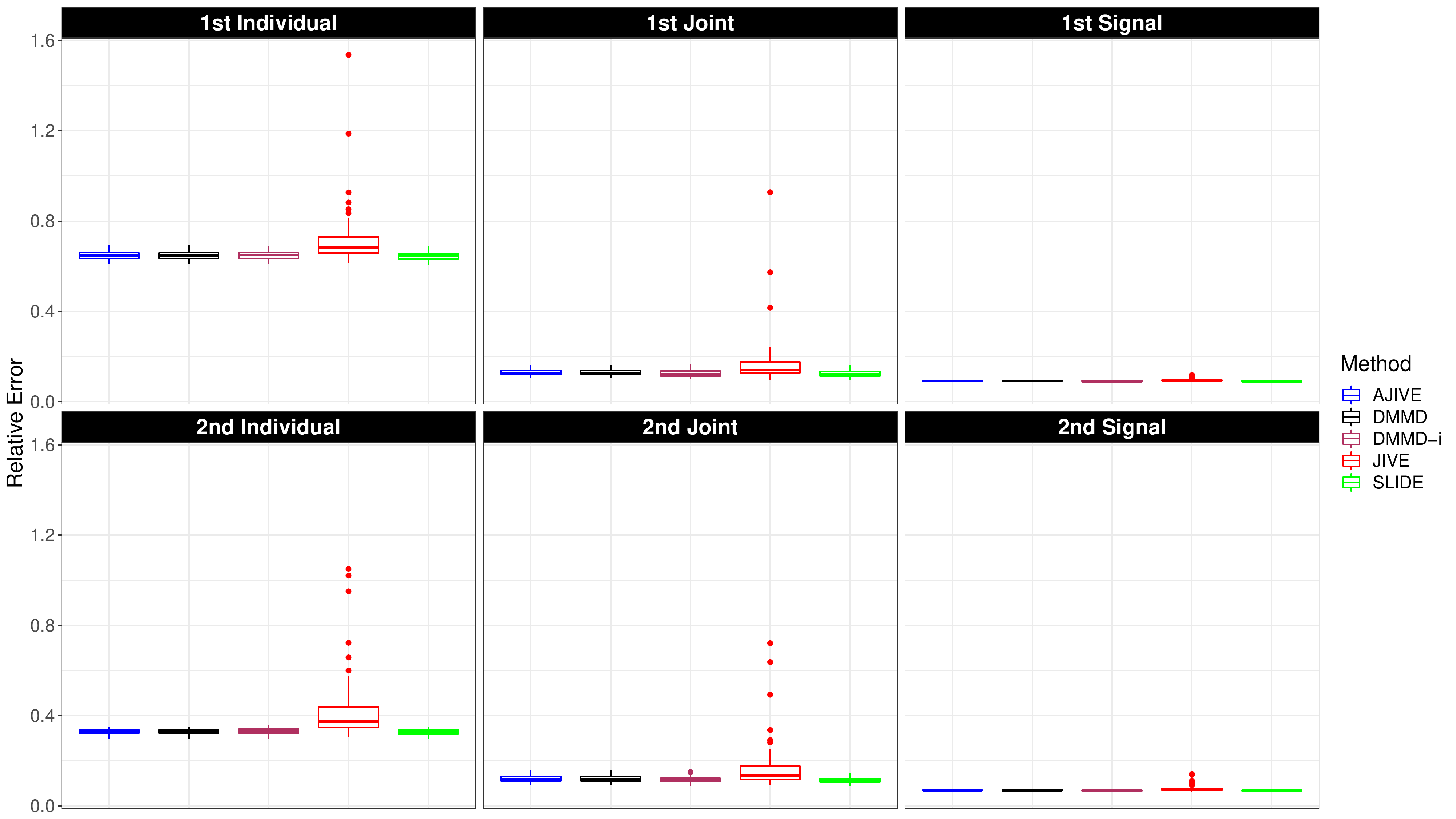}
\caption{Row space decomposition (matched columns)}
\label{fig:tcga_row}
\end{subfigure}
\caption{Comparison of signal identification over 100 replications when $n = 88, p = 736, r_1 = 8, r_2 = 6, r_c = 0, r_r = 2$.}
\end{figure}

\subsection{Signal estimation when ranks are misspecified}
The effect of rank misspecification is studied using simulation setting 4 with $n = 240$, $p = 200$, and true $r_1^* = 20$, $r_2^* = 18$, $r_c^* = 4$, $r_r^* = 3$. We consider three types of misspecifications:
\begin{enumerate}[label=(\alph*)]
    \item The input ranks are all less than the true rank with $r_1 = 19, r_2 = 17, r_c = 3, r_r = 2$. 
    \item The input ranks are all larger than the true rank with $r_1 = 21, r_2 = 19, r_c = 5, r_r = 4$.
    \item The total ranks are correct, but the joint ranks are misspecified to be smaller than the truth. The inputs are $r_1 = 20, r_2 = 18, r_c = 3, r_r = 2$. 
\end{enumerate}
Figures~\ref{fig:misRank_row}--\ref{fig:misRank_column} show relative errors on total signal for all methods under each misspecification scenario. DMMD is consistently better than other methods. Iterative DMMD is very similar to the original one, with slightly better performance in settings (a) and (c). In setting (b), DMMD-i is better than DMMD for the 2nd signal matrix, but is worse for the 1st signal matrix. We suspect that this inconsistency is due to setting (b) using larger ranks than the truth, thus all methods contain noise in the estimated signal.

\begin{figure}[!t]
\begin{subfigure}{0.5\textwidth}
\includegraphics[width=1\textwidth]{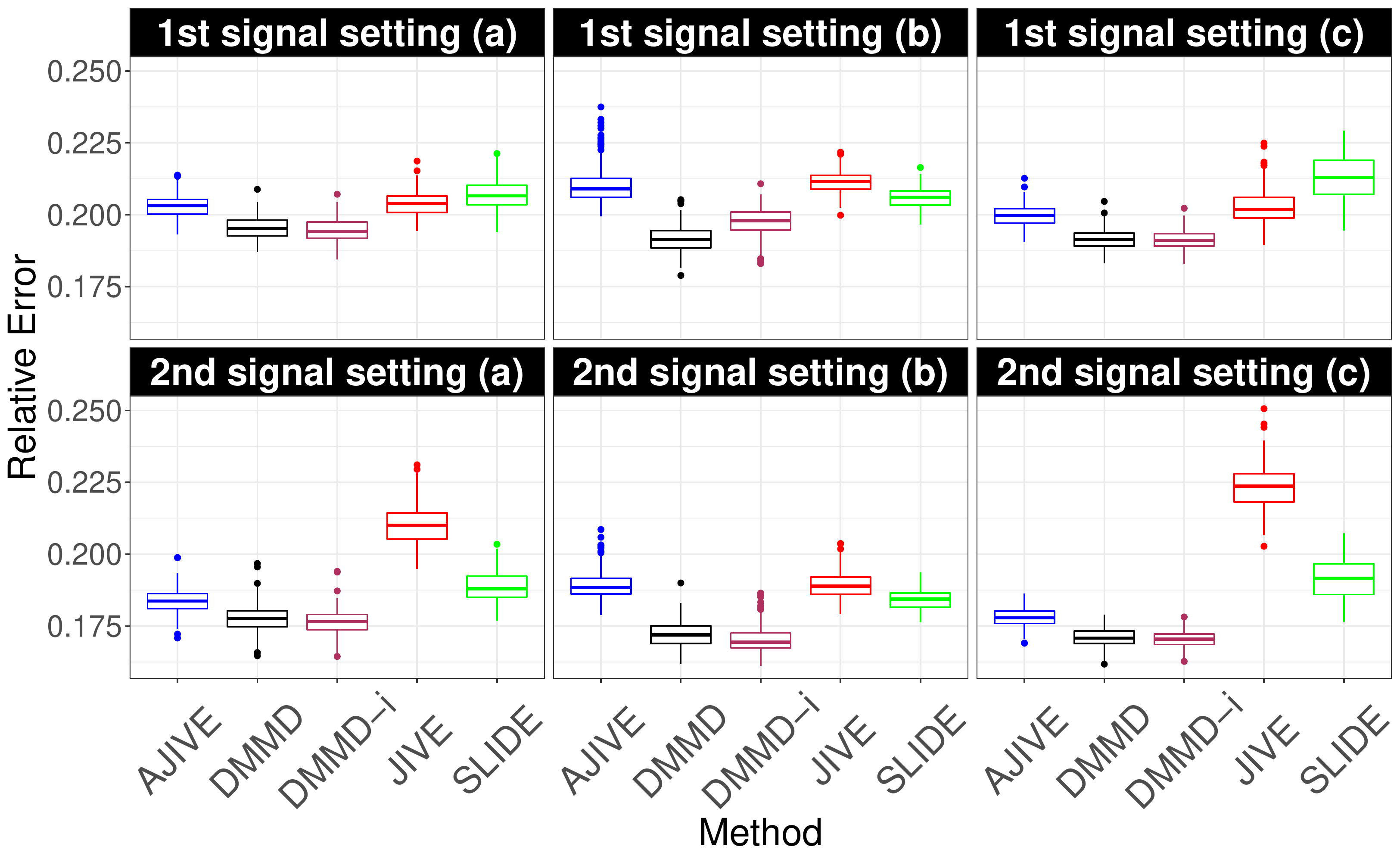}
\caption{Row space decomposition (matched columns)}
\label{fig:misRank_row}
\end{subfigure}%
\begin{subfigure}{0.5\textwidth}
\includegraphics[width=1\textwidth]{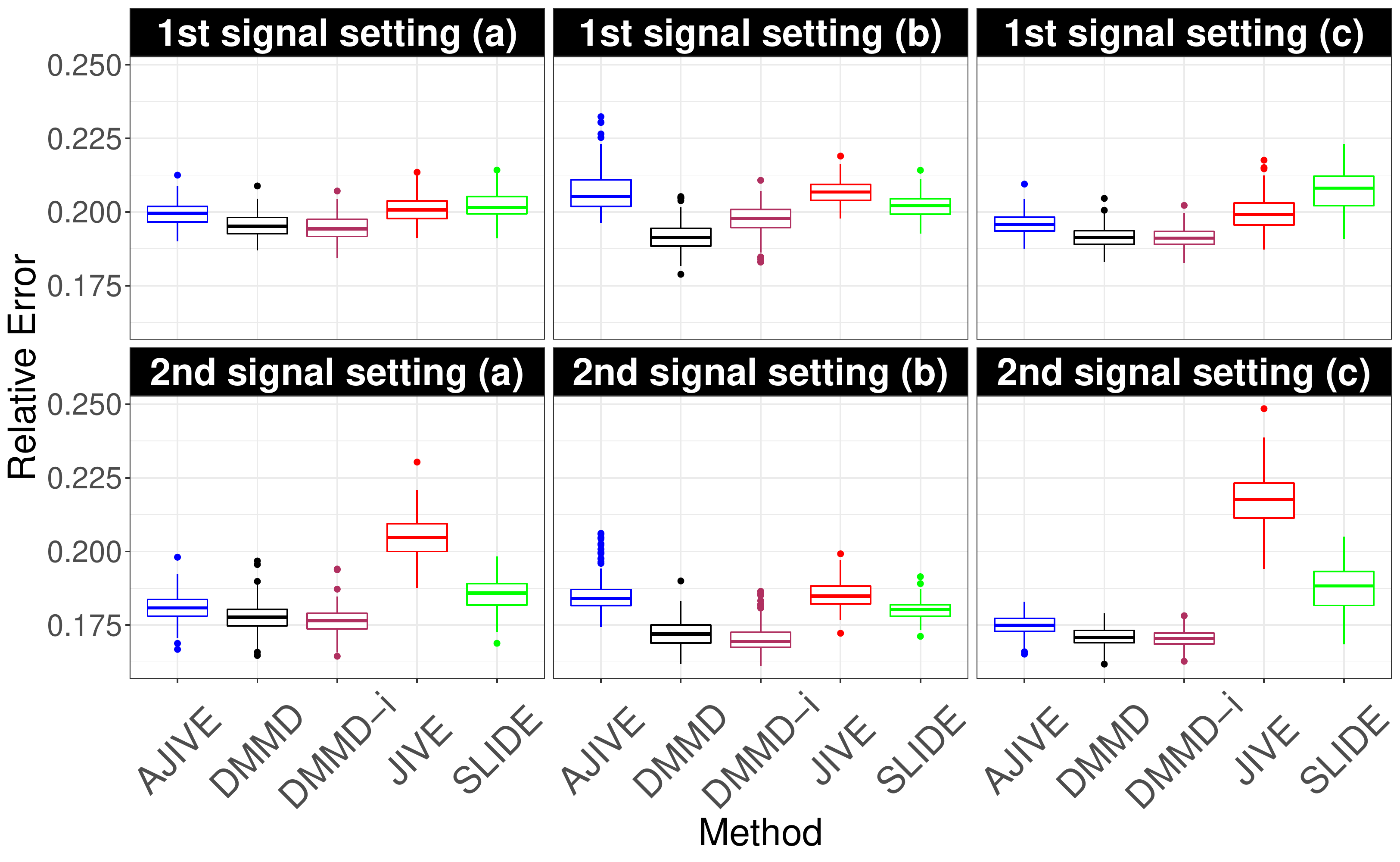}
\caption{Column space decomposition (matched rows)}
\label{fig:misRank_column}
\end{subfigure}
\caption{Comparison of methods on total signal estimation when ranks are misspecified.}
\end{figure}

As the ranks are misspecified, separate comparisons on only joint structures or only individual structures could be misleading. Thus, for all methods, we calculate the chordal distance between subspace associated with true joint structure, and estimated subspace of total signal (Table~\ref{tab:misRank}).

For settings (a) and (c), the supplied joint rank is smaller than the truth, and thus we expect that some of the joint structure would be mistakenly estimated as individual. SLIDE performs the worst in these settings, as it restricts individual signals to be orthogonal, and thus it has difficulties capturing the missing 4th joint basis. DMMD, DMMD-i and AJIVE have the best performance here, with DMMD-i being slightly better than DMMD. Principal angle analysis of the individual signals of DMMD reveals that the first angle is very small (albeit not zero), hence the missing 4th joint basis has been captured by individual structures that are close to each other (albeit not identical). We conclude that DMMD is robust to under-estimation of joint ranks, and in practice we recommend examining principal angles between estimated individual structures for diagnostics of such under-specification.
 
\begin{table}[!t]
\caption{Mean chordal distance between estimated total signal subspace with true joint subspace across 140 replications (standard error is given in brackets).}
\label{tab:misRank}
\begin{center}
\begin{tabular}{c|lrrrrr}
  \hline
 Section & Setting & DMMD & DMMD-i & SLIDE & JIVE & AJIVE \\ 
  \hline
1st column  & (a) & 0.34 (0.004) & 0.33 (0.004) & 0.42 (0.007) & 0.40 (0.003) & 0.33 (0.004) \\ 
subspace & (b) & 0.26 (0.001) & 0.26 (0.001) & 0.26 (0.001) & 0.35 (0.002) & 0.26 (0.001) \\ 
& (c) & 0.30 (0.002) & 0.29 (0.002) & 0.43 (0.005) & 0.36 (0.002) & 0.30 (0.002) \\ 
  \hline
  1st row & (a) & 0.32 (0.004) & 0.32 (0.004) & 0.46 (0.008) & 0.32 (0.004) & 0.33 (0.004) \\ 
 subspace  & (b) & 0.24 (0.001) & 0.24 (0.001) & 0.24 (0.001) & 0.24 (0.001) & 0.24 (0.001) \\ 
   & (c) & 0.28 (0.002) & 0.28 (0.002) & 0.47 (0.006) & 0.30 (0.004) & 0.28 (0.002) \\ 
  \hline
  2nd column & (a) & 0.34 (0.004) & 0.33 (0.004) & 0.45 (0.006) & 0.40 (0.003) & 0.34 (0.004) \\ 
 subspace  & (b) & 0.26 (0.001) & 0.26 (0.001) & 0.26 (0.001) & 0.33 (0.002) & 0.26 (0.001) \\ 
   & (c) & 0.29 (0.002) & 0.28 (0.001) & 0.42 (0.005) & 0.38 (0.003) & 0.29 (0.002) \\ 
  \hline
  2nd row  & (a) & 0.33 (0.005) & 0.32 (0.005) & 0.47 (0.009) & 0.60 (0.002) & 0.33 (0.005) \\ 
subspace   & (b) & 0.24 (0.001) & 0.24 (0.001) & 0.24 (0.001) & 0.26 (0.003) & 0.24 (0.001) \\ 
   & (c) & 0.27 (0.002) & 0.27 (0.002) & 0.44 (0.007) & 0.59 (0.002) & 0.27 (0.002) \\ 
   \hline
\end{tabular}
\end{center}
\end{table} 

In setting (b), all ranks are higher than the truth, and all methods have similar performance in capturing joint structure, except JIVE, which performs the worst. Taking the results on joint structures in conjunction with total signal estimation results in Figure~\ref{fig:misRank_column}, we conclude that when ranks are over-specified, DMMD fits to the noise less than other methods. We believe this robustness to over-specification is due to DMMD's explicit restriction on signals being matched in both row and column directions, which helps to prevent noise overfitting.

\subsection{Computational comparisons}
We compare computational times of all methods using $\mbox{SNR} = 1$ when (i) $n = 100, p = 80, r_1 = 10, r_2 = 8, r_c = 4, r_r = 3$;
   (ii) $n = 100, p = 800, r_1 = 25, r_2 = 20, r_c = 10, r_r = 5$.
We evaluate the total running time (rank estimation $+$ model fitting) as well as model fitting time only given the ranks. All comparisons are done on Intel(R) Core(TM) i5-7300U CPU @ 2.60GHz.
Table~\ref{t:RunningTime} reports running times in seconds. Since all methods except DMMD estimate column and row decompositions separately, their run time is the sum of the two. DMMD is significantly faster than the competitors in total run time, and is quite fast with given ranks. As expected, DMMD-i (Section~\ref{s:DMMDi}) is significantly more costly than DMMD, albeit still faster than SLIDE. 

\begin{table}[!t]
\caption{Comparison of running times (in seconds)}
\begin{center}
\begin{tabular}{ccccc}
Method & Total, p = 80 & Given ranks, p = 80 & Total, p = 800 & Given ranks, p = 800 \\ \hline
DMMD               & 0.8   & 0.4  & 24.1  & 11.8  \\ 
DMMD-i$^{*1}$      & 10.9   & 1.4  & 618.3 & 48.4 \\
JIVE (Row)         & 2.1   & 8.3  & 219.7 & 21.1 \\ 
JIVE (Col)         & 2.2 & 6.1  & 12.6  & 7.8 \\ 
SLIDE (Row)        & 36.2  & 0.0  & 1226.8 & 1.0  \\ 
SLIDE (Col)        & 63.7  & 0.0  & 317.3 & 1.1  \\ 
AJIVE$^{*2}$ (Row) & 6.3   & 0.0  & 66.6  & 0.4  \\ 
AJIVE$^{*2}$ (Col) & 6.5   & 0.0  & 42.5  & 0.3  \\ \hline
\end{tabular}
\end{center}
\footnotesize{$^{*1}$ DMMD-i is the iterative version of DMMD, which updates joint structures.\\
$^{*2}$ Total ranks are necessary inputs for AJIVE. Total running time of AJIVE is measured given total ranks.
}
\label{t:RunningTime}
\end{table}

\section{Additional details on TCGA data application}
Here we present alternative heatmaps of the joint row structure ($\widehat r_r = 2$) corresponding to matched miRNAs that are aligned vertically (Figures~\ref{fig:joint_miRNA_tumor_revised} and~\ref{fig:joint_miRNA_normal_revised}). There are three visually distinguishable clusters of miRNAs, and the clustering is preserved across both tissue types.

\label{s:a_TCGA}
\begin{figure}[!t]
\begin{center}
\begin{subfigure}{0.56\textwidth}
\includegraphics[width=1\textwidth]{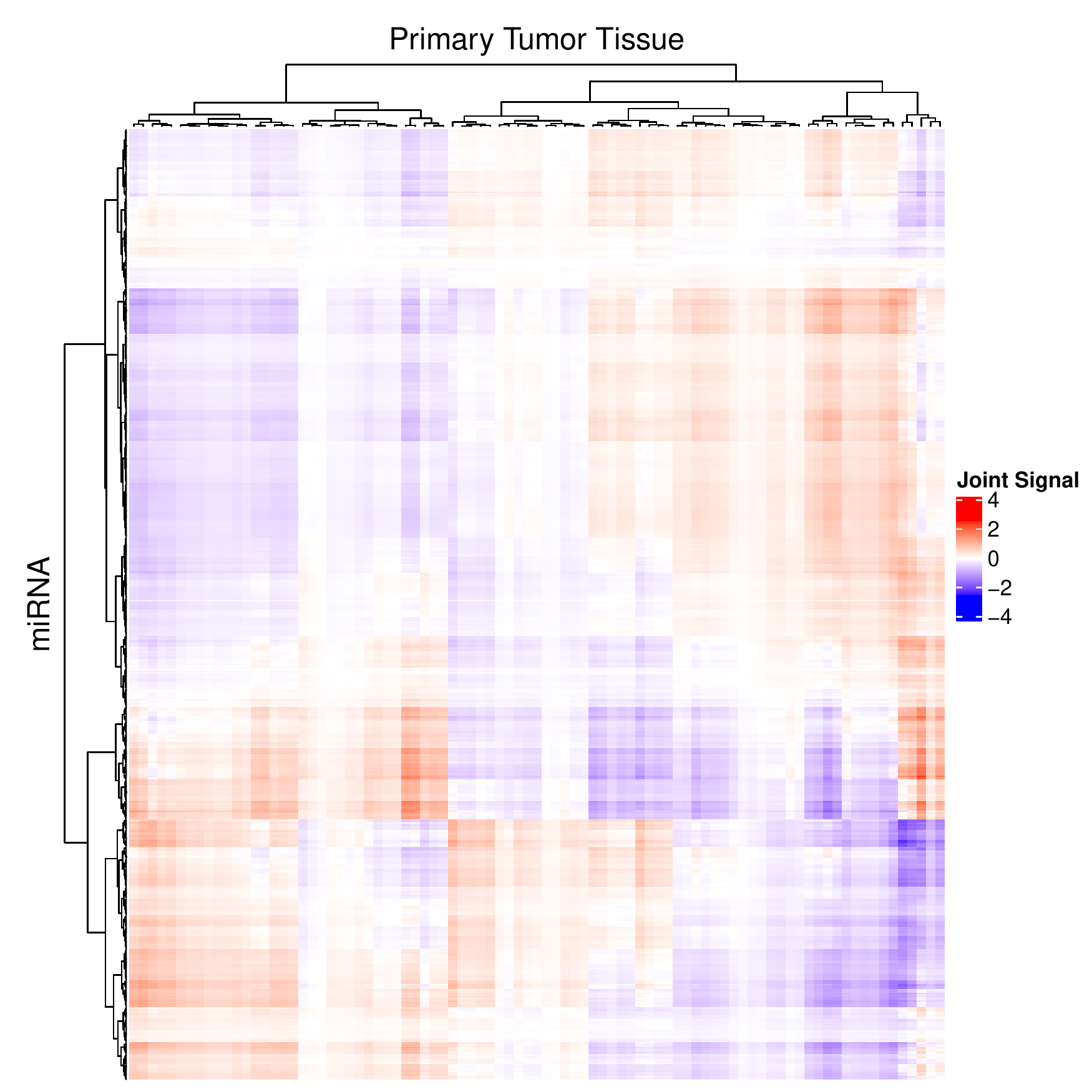}
\caption{}
\label{fig:joint_miRNA_tumor_revised}
\end{subfigure}

\begin{subfigure}{0.56\textwidth}
\includegraphics[width=1\textwidth]{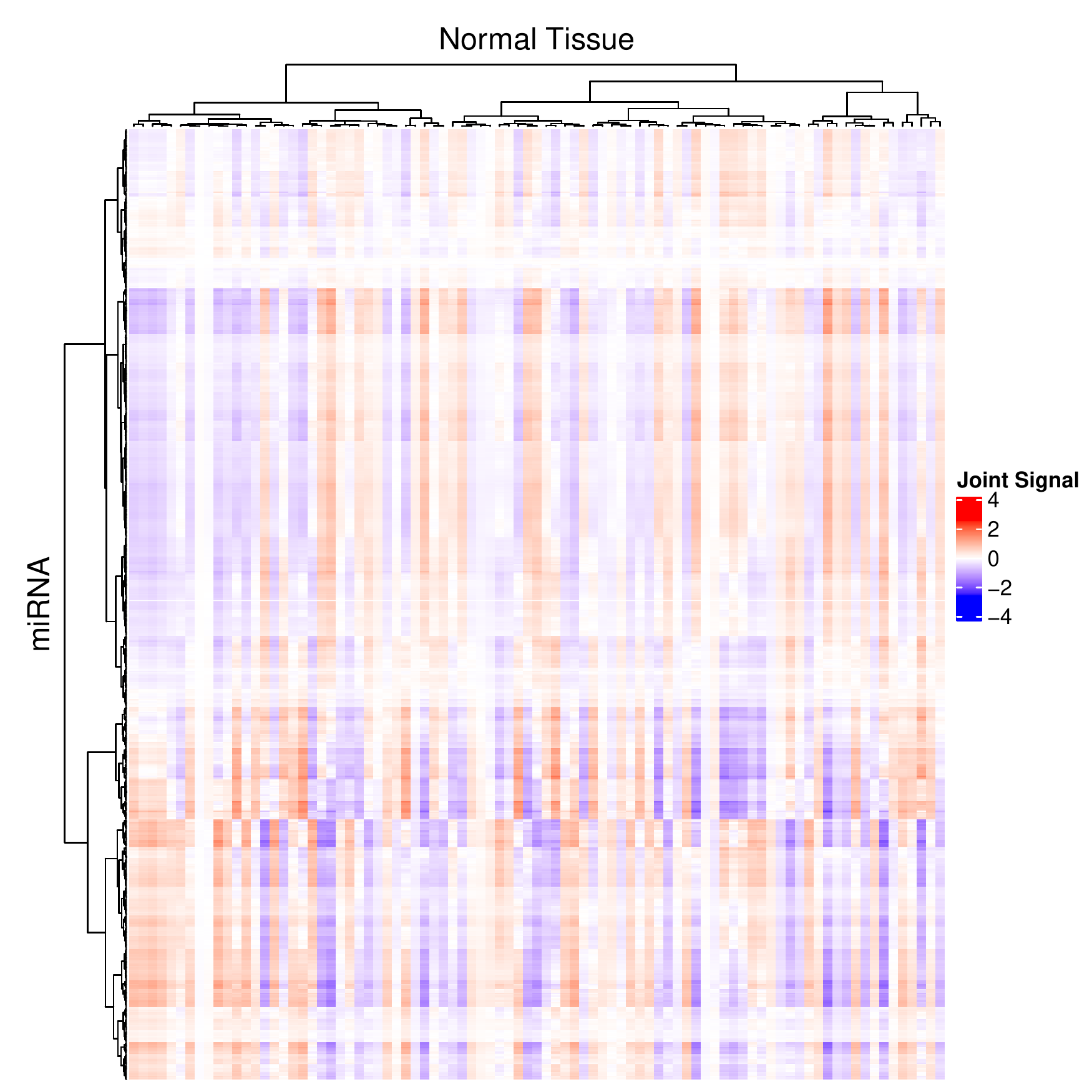}
\caption{}
\label{fig:joint_miRNA_normal_revised}
\end{subfigure}
\caption{Joint row (miRNA) structures extracted by DMMD for primary tumor and normal tissues from matched TCGA-BRCA miRNA data. The order of samples and miRNAs in both figures is the same, which is determined by the joint structure of primary tumor tissue.}
\end{center}
\end{figure}

\section{Iterative DMMD algorithm (DMMD-i)}\label{s:DMMDi}

Consider DMMD signal estimation optimization problem (2):
\begin{align}
    &\minimize_{\MBA_k\in \R^{n\times p}}{\|\MBX_k - \MBA_k\|^2_{F}} \label{eq:DMopt1}\\
    & \mbox{such that}\quad \mathcal{C}(\MBM) \subset \mathcal{C}(\MBA_k),\quad \mathcal{C}(\MBN) \subset \mathcal{R}(\MBA_k),\quad \text{rank}(\MBA_k) = r_k,\quad k = 1,2 \nonumber.
\end{align}
Problem (2) treats joint $\MBM$ and $\MBN$ as fixed. Here we consider an extension of DMMD where the optimization is performed over all parts of the signal, that is
\begin{align}
    &\minimize_{\MBA_k, \MBM, \MBN}{\sum_{k=1}^2\|\MBX_k - \MBA_k\|^2_{F}} \label{eq:DMMD-i}\\
    & \mbox{such that}\quad \mathcal{C}(\MBM) \subset \mathcal{C}(\MBA_k),\quad \mathcal{C}(\MBN) \subset \mathcal{R}(\MBA_k),\quad \text{rank}(\MBA_k) = r_k,\quad k = 1,2 \nonumber.
\end{align}

Following the notation of Algorithm~1, let $\widetilde \MBM_k = [\MBM, \MBR_k]\in \R^{n \times r_k}$ be the matrix of basis vectors for the full column space of $\MBA_k$, and $\widetilde \MBN_k = [\MBN, \MBS_k]\in \R^{p \times r_k}$ be the matrix of basis vectors for the full row space of $\MBA_k$. Then the above problem can be equivalently rewritten as
\begin{align}
    &\minimize_{\MBM, \MBN, \MBR_k, \MBS_k}{\sum_{k=1}^2\|\MBX_k - \widetilde \MBM_k\widetilde \MBM_k^{\top} \MBX_k^{\top}\widetilde \MBN_k\widetilde \MBN_k^{\top}\|^2_{F}} \label{eq:DMit}\\
    & \mbox{such that}\quad \widetilde \MBM_k = [\MBM, \MBR_k],\quad \widetilde \MBN_k = [\MBN, \MBS_k],\quad \widetilde \MBM_k^{\top}\widetilde \MBM_k= \widetilde \MBN_k^{\top}\widetilde \MBN_k = \MBI, \quad k = 1,2 \nonumber.
\end{align}
Given joint $\MBM$ and $\MBN$, the minimization with respect to individual $\MBR_k$ and $\MBS_k$ is performed using Algorithm~1 (DMMD). Iterative DMMD allows to update the initial $\MBM$ and $\MBN$ by further optimizing~\eqref{eq:DMit} with respect to joint $\MBM$ and $\MBN$ with given individual $\MBR_k$ and $\MBS_k$.

Given $\MBR_k$, $\MBN$ and $\MBS_k$, the optimal $\MBM$ is given by Lemma~\ref{l:s4}. Using $\MBX_k^{\top}$ instead of $\MBX_k$, the same Lemma~\ref{l:s4} gives optimal $\MBN$ given $\MBS_k$, $\MBM$ and $\MBR_k$. Combining updates of $\MBR_k$, $\MBS_k$ from Algorithm~1 with updates of $\MBM$, $\MBN$ according to Lemma~\ref{l:s4} gives rise to iterative DMMD (DMMD-i) Algorithm~\ref{a:DMMD-i}.

\begin{algorithm}[!t]
\caption{Iterative DMMD algorithm (DMMD-i)}
\label{a:DMMD-i}
\begin{algorithmic}[1]
\State Given: $\MBX_k \in \R^{n\times p}$, $r_k$, $k=1, 2$; $\MBM^{(0)} \in \R^{n\times r_c}$, $\MBN^{(0)} \in \R^{p\times r_r},t_{max},\epsilon > 0$
\For{$k = 1,2$}
\State SVD: $(\MBId - \MBM^{(0)}\MBM^{(0)T})\MBX_k = \MBU_k\MBD_k\MBV_k^T$
\State $\MBR_k^{(0)} \gets$ first $r_k - r_c$ columns of $\MBU_k$
\State SVD: $\MBX_k(\MBId - \MBN^{(0)}\MBN^{(0)T}) = \widehat{\MBU}_k\widehat{\MBD}_k\widehat{\MBV}_k^T$
\State $\MBS_k^{(0)} \gets$ first $r_k - r_r$ columns of $\widehat{\MBV}_k$
\State $\widetilde{\MBM}_k^{(0)} \gets [\MBM^{(0)},\MBR_k^{(0)}]$, $\widetilde{\MBN}_k^{(0)} \gets [\MBN^{(0)},\MBS_k^{(0)}]$
\EndFor
\State $t \gets 0$
\While{$t\neq t_{max}$ and $\max_k |L_k^{(t)} - L_k^{(t-1)}|>\epsilon$}
\State Update of $\MBM$:
\State $\quad$ $\MBY^{(t)} \gets [(\mathbf{Id} - \MBR^{(t)}_1\MBR_1^{(t)\top})\MBX_1\widetilde{\MBN}^{(t)}_1\widetilde{\MBN}^{(t)\top}_1,(\mathbf{Id}-\MBR^{(t)}_2\MBR_2^{(t)\top})\MBX_2\widetilde{\MBN}^{(t)}_2\widetilde{\MBN}^{(t)\top}_2]$
\State $\quad$ SVD: $\MBY^{(t)}=\MBU^{(t)}_k\MBD^{(t)}_k\MBV^{(t)T}_k$
\State $\quad$ $\MBM^{(t+1)} \gets$ first $r_c$ columns of $\MBU^{(t)}_k$, $\widetilde{\MBM}_k^{(t+1)} \gets [\MBM^{(t+1)},\MBR_k^{(t)}]$
\State Update of $\MBN$:
\State $\quad$ $\MBZ^{(t)} \gets [\widetilde{\MBM}^{(t+1)}_1\widetilde{\MBM}^{(t+1)\top}_1\MBX_1(\mathbf{Id} - \MBS^{(t)}_1\MBS_1^{(t)\top}),\widetilde{\MBM}^{(t+1)}_2\widetilde{\MBM}^{(t+1)\top}_2\MBX_2(\mathbf{Id}-\MBS^{(t)}_2\MBS_2^{(t)\top})]$
\State $\quad$ SVD: $\MBZ^{(t)}=\widehat{\MBU}^{(t)}_k\widehat{\MBD}^{(t)}_k\widehat{\MBV}^{(t)T}_k$
\State $\quad$ $\MBN^{(t+1)} \gets$ first $r_r$ columns of $\widehat{\MBU}^{(t)}_k$, $\widetilde{\MBN}_k^{(t+1)} \gets [\MBN^{(t+1)},\MBS_k^{(t)}]$

\State Update of $\MBR_k$ and $\MBS_k$:

\State $\quad$ Apply Algorithm 1 to get $\MBR_k^{(t)}$, $\MBS_k^{(t)}$, $k=1, 2$
\State $\quad$ $\widetilde{\MBM}_k^{(t+1)} \gets [\MBM^{(t+1)},\MBR_k^{(t+1}]$, $k=1, 2$
\State $\quad$ $\widetilde{\MBN}_k^{(t+1)} \gets [\MBN^{(t+1)},\MBS_k^{(t+1}]$, $k=1, 2$



\State $t \gets t + 1$
\State $L_k^{(t)} = \|\MBX_k - \widetilde{\MBM}_k^{(t)}\widetilde{\MBM}_k^{(t)T}\MBX_k\widetilde{\MBN}_k^{(t)}\widetilde{\MBN}_k^{(t)T}\|^2_F$, $k=1, 2$
\EndWhile
\State \Return {$\MBA_k^* = \widetilde{\MBM}_k^{(t)}\widetilde{\MBM}_k^{(t)T}\MBX_k\widetilde{\MBN}_k^{(t)}\widetilde{\MBN}_k^{(t)T}$, $k=1,2$}
\end{algorithmic}
\end{algorithm}

We compare the results of original DMMD with iterative DMMD-i on soccer dataset (Tables~\ref{t:Soccer11_Irina}--\ref{t:Soccer21_Irina}). The results are very similar without affecting the main conclusions.

\begin{table}[!t]
\caption{Comparison on joint row basis for winning and losing teams in English Premier League when $r_1 = r_2 = r_r = 1$.}
\begin{center}
\footnotesize
\begin{tabular}{p{1.3cm}p{1.0cm}p{0.9cm}p{0.9cm}p{1.1cm}p{1.1cm}p{1.2cm}p{1.2cm}p{1.2cm}p{1cm}p{1cm}}
Signal & Full Time Goals & Half Time Goals & Shots & Shots on Target & Hit Woodwork & Corners & Fouls Commited & Offsides & Yellow Cards & Red Cards\\\hline
DMMD joint & 1.00 & 0.45 & 8.08 & 3.90 & 0.23 & 4.03 & 9.80 & 2.50 & 1.10 & 0.07 \\
DMMD-i joint & 1.00 & 0.45 & 7.42 & 3.63 & 0.21 & 3.64 & 9.05 & 2.38 & 1.01 & 0.06 \\\hline
\end{tabular}
\label{t:Soccer11_Irina}
\end{center}
\end{table}

\begin{table}[!t]
\caption{Joint row basis and individual row basis for winning teams in English Premier League when $r_1 = 2, r_2 = r_r = 1$.}
\begin{center}
\footnotesize
\begin{tabular}{p{1.3cm}p{1.0cm}p{0.9cm}p{0.9cm}p{1.1cm}p{1.1cm}p{1.2cm}p{1.2cm}p{1.2cm}p{1cm}p{1cm}}
Signal & Full Time Goals & Half Time Goals & Shots & Shots on Target & Hit Woodwork & Corners & Fouls Commited & Offsides & Yellow Cards & Red Cards \\\hline
DMMD Joint & 1.00 & 0.47 & 7.85 & 3.78 & 0.21 & 4.01 & 11.35 & 2.76 & 1.30 & 0.08 \\
DMMD-i Joint & 1.00 & 0.47 & 8.12 & 3.85 & 0.22 & 4.16 & 11.95 & 2.95 & 1.40 & 0.09 \\
DMMD Win & 1.00 & 0.33 & 5.09 & 2.93 & 0.21 & 1.62 & -4.98 & -0.39 & -0.88 & -0.08 \\
DMMD-i Win & 1.00 & 0.34 & 4.79 & 2.83 & 0.20 & 1.52 & -4.59 & -0.42 & -0.86 & -0.09 \\ \hline
\end{tabular}
\label{t:Soccer21_Irina}
\end{center}
\end{table}

\section{Generalization of DMMD to more than two views}\label{s:largeK}

We consider $K$ double-matched data matrices $\MBX_k\in \mathbb{R}^{n \times p}$, $k=1, \dots, K$, with DMMD decomposition according to Lemma 1:
\begin{equation}\label{eq:dmmdK}
    \MBX_k = \underbrace{\MBJ_{ck} + \MBI_{ck}}_{\MBA_k}+ \MBE_k = \underbrace{\MBJ_{rk} + \MBI_{rk}}_{\MBA_k} + \MBE_k, \quad k=1,\dots, K.
\end{equation}
All DMMD estimation steps can be applied to the case $K>2$ with the exception of joint structure estimation (Step 2 described in Section 2.3.2). This step is specific to $K=2$ case as it determines joint ranks based on principal angles between two subspaces. Furthermore, the joint basis vectors in $\MBM$ and $\MBN$ are computed based on averaging corresponding principal vectors. Thus, both the joint rank determination, and the computation of $\MBM$, $\MBN$ require adjustment when $K>2$.

 In our numerical studies of rank estimation performance in Section 3.2, we found that an alternative joint rank estimation approach of AJIVE \citep{AJIVE} works quite well. The latter can be applied with any number of views $K$, however requires supplying the total ranks as the input. Since total ranks can be estimated using profile likelihood as in Section 2.3.1, we recommend to estimate the ranks by combining profile likelihood method with AJIVE joint rank estimation approach when applying DMMD in $K>2$ setting.
 
 Given the joint column rank $r_c$ and row rank $r_r$, we propose to construct $\MBM$ and $\MBN$ based on SUM-PCA \citep{CPCA}, that is low-rank SVD on views either concatenated column-wise (for $\MBM$) or row-wise (for $\MBN$). Using $\MBM$ as an example, this approach is equivalent to finding the solution to
 \begin{equation}
     \minimize_{\MBM}\Big\{\sum_{k=1}^K\|\MBX_k - \MBM\MBM^{\top}\MBX_k\|_F^2\Big\}\quad \mbox{s.t.}\quad \MBM^{\top}\MBM = \MBI_{r_c}.
 \end{equation}
 Since the objective function uses squared Frobenius loss, this minimization coincides with iterative DMMD algorithm update of $\MBM$ when the individual structures are initialized as zero. In case computational time is not a constraint, iterative DMMD can be directly used when $K>2$ to further modify these initial $\MBM$ and $\MBN$, which based on our simulations leads to slightly improved performance.



\section{Difficulties in capturing joint structure in Tucker decomposition}\label{s:tensor}
Let $\MBX_1, \MBX_2\in \R^{n \times p}$ be double-matched. We can view these data alternatively as a three-way tensor  $\bm{\mathscr{X}} \in \mathbb{R}^{n \times p \times 2}$ with the frontal slices:
$
\bm{\mathscr{X}}_{::1} = \MBX_1, \quad
\bm{\mathscr{X}}_{::2} = \MBX_2
$. We can consider the Tucker decomposition \citep{Kolda:2009dh} where a core tensor is multiplied by a matrix along each mode, that is
$$\bm{\mathscr{X}} \approx \bm{\mathscr{G}} \times_1 \mathbf{A} \times_2 \mathbf{B} \times_3 \mathbf{C},$$
where $\mathbf{A} \in \mathbb{R}^{n \times P}, \mathbf{B} \in \mathbb{R}^{p \times Q}, \mathbf{C} \in \mathbb{R}^{2 \times R}$ are the factor matrices and the tensor $\bm{\mathscr{G}} \in \mathbb{R}^{P \times Q \times R}$ is  the core tensor. To obtain the Tucker decomposition, we can consider the low-rank SVD for each mode-d matricization of tensor $\bm{\mathscr{X}}$, that is
$$\bm{\mathscr{X}}_{(1)} \approx \mathbf{U}_1\mathbf{\Sigma}_1\mathbf{V}^T_1 \in \R^{n \times 2p}\quad \mbox{with rank} \ P$$
$$\bm{\mathscr{X}}_{(2)} \approx \mathbf{U}_2\mathbf{\Sigma}_2\mathbf{V}^T_2\in \R^{p \times 2n}\quad \mbox{with rank} \ Q$$
$$\bm{\mathscr{X}}_{(3)} \approx \mathbf{U}_3\mathbf{\Sigma}_3\mathbf{V}^T_3\in \R^{2 \times pn} \quad \mbox{with rank} \ R$$
and calculate 
$$\bm{\mathscr{S}} = \bm{\mathscr{X}} \times_1 \mathbf{U}^T_1 \times_2 \mathbf{U}^T_2 \times_3 \mathbf{U}^T_3.$$
Then the Tucker decomposition becomes
$$\bm{\mathscr{X}} \approx \bm{\mathscr{S}} \times_1 \mathbf{U}_1 \times_2 \mathbf{U}_2 \times_3 \mathbf{U}_3.$$

From this decomposition, it is natural to encode joint column space information in $\mathbf{U}_1 \in \R^{n \times P}$ and the joint row space information in $\mathbf{U}_2\in \R^{p \times Q}$. However, we found that such encoding does not always lead to joint structure that aligns with the matrix case, as illustrated in the following toy example.

Consider a noiseless case with
\setlength{\arraycolsep}{2.5pt}
\medmuskip = 2mu 

\[
  \MBX_1 =
  \left[ {\begin{array}{ccc}
   0 & 0 & 0 \\
   0 & 0 & 1 \\
   1 & 1 & 0 \\
  \end{array} } \right] \quad
  \MBX_2 =
  \left[ {\begin{array}{ccc}
   0 & 1 & 0 \\
   1 & 0 & 0 \\
   0 & 0 & 0 \\
  \end{array} } \right]
\]
corresponding to $3 \times 3 \times 2$ tensor. By direct calculation, the joint column and row spaces between $\MBX_1$ and $\MBX_2$ are:
$$
\mathcal{C}(\MBX_1) \cap \mathcal{C}(\MBX_2) = Span\{(0,1,0)^\top\}, \quad 
\mathcal{R}(\MBX_1) \cap \mathcal{R}(\MBX_2) = Span\{(1,1,0)^\top\}.
$$

Thus, we use 
\[
  \MBU_1 =
  \left[ {\begin{array}{c}
   0 \\
   1 \\
   0 \\
  \end{array} } \right], \quad
  \MBU_2 =
  \left[ {\begin{array}{c}
   1 \\
   1 \\
   0 \\
  \end{array} } \right], \quad
  \MBU_3 =
  \left[ {\begin{array}{cc}
   1 & 0 \\
   0 & 1 \\
  \end{array} } \right]
\]
to capture the information of joint column and row spaces (note that we use the largest rank for $\MBU_3$ to have the most flexible model that is possible). Then the joint core tensor is obtained as
$$\bm{\mathscr{S}}_{\text{joint}} = \bm{\mathscr{X}} \times_1 \mathbf{U}^T_1 \times_2 \mathbf{U}^T_2 \times_3 \mathbf{U}^T_3$$
and 
$$\bm{\mathscr{X}}_{\text{joint}} = \bm{\mathscr{S}}_{\text{joint}} \times_1 \mathbf{U}_1 \times_2 \mathbf{U}_2 \times_3 \mathbf{U}_3.$$

The resulting tensor $\bm{\mathscr{X}}_{\text{joint}}$ of size $3 \times 3 \times 2$ has the following two frontal slices: 
\[
  \left[ {\begin{array}{ccc}
   0 & 0 & 0 \\
   0 & 0 & 0 \\
   0 & 0 & 0 \\
  \end{array} } \right]
, \quad
  \left[ {\begin{array}{ccc}
   0 & 0 & 0 \\
   1 & 1 & 0 \\
   0 & 0 & 0 \\
  \end{array} } \right].
\]
Surprisingly, the resulting slices in $\bm{\mathscr{X}}_{\text{joint}}$ no longer have any joint information. This example tells us that Tucker decomposition may have difficulties capturing joint row and column structures simultaneously. 

\end{document}